\documentclass[a4paper,onecolumn,accepted=2020-01-19,11pt]{quantumarticle}
\pdfoutput=1

\usepackage[numbers,sort&compress]{natbib}
\usepackage[utf8]{inputenc}
\usepackage[english]{babel}
\usepackage[T1]{fontenc}

\usepackage{amsmath,amssymb,amsthm}
\usepackage{graphicx}
\usepackage{braket}
\usepackage{empheq}
\usepackage{enumitem}
\usepackage{subfigure}
\usepackage{comment}
\usepackage{upgreek}
\usepackage[normalem]{ulem}

\usepackage[pdfstartview=FitH]{hyperref}

\usepackage{doi}

\newcommand{\ignore}[1]{}

\newcommand{\bes} {\begin{subequations}}
\newcommand{\ees} {\end{subequations}}
\newcommand{\beq}{\begin{equation}}
\newcommand{\eeq}{\end{equation}}
\newcommand{\bea}{\begin{eqnarray}}
\newcommand{\eea}{\end{eqnarray}}

\newcommand\Tr{\mathrm{Tr}}

\newcommand{\pen}[1]{\left(#1\right)}								
\newcommand{\ben}[1]{\left[#1\right]}								
\newcommand{\cen}[1]{\left\{#1\right\}}								
\newcommand{\expv}[1]{\langle #1\rangle}							

\newtheorem*{theorem*}{Theorem}

\newtheorem*{lemma*}{Lemma}

\newtheorem{mylemma}{Lemma}

\newcommand{\mc}{\mathcal}
\renewcommand{\Re}{\textrm{Re}}
\renewcommand{\Im}{\textrm{Im}}

\def\a{\alpha}
\def\b{\beta}
\def\g{\gamma}

\def\o{\omega}

\def\beq{\begin{equation}}
\def\eeq{\end{equation}}

\DeclareMathOperator{\sinc}{sinc}
\DeclareMathOperator{\sgn}{sgn}

\begin{document}



\title{Completely positive master equation for arbitrary driving and small level spacing}



\author{Evgeny Mozgunov}
	\affiliation{Center for Quantum Information Science \&
		Technology, University of Southern California, Los Angeles, California 90089, USA}
\author{Daniel Lidar}
	\affiliation{Center for Quantum Information Science \&
		Technology, University of Southern California, Los Angeles, California 90089, USA}
	\affiliation{Department of Electrical and Computer Engineering, University of Southern California, Los Angeles, California 90089, USA}
	\affiliation{Department of Chemistry, University of Southern California, Los Angeles, California 90089, USA}
	\affiliation{Department of Physics and Astronomy, University of Southern California, Los Angeles, California 90089, USA}

\begin{abstract}
Markovian master equations are a ubiquitous tool in the study of open quantum systems, but deriving them from first principles involves a series of compromises. On the one hand, the Redfield equation is 
valid for fast environments (whose correlation function decays much faster than the system relaxation time) regardless of the relative strength of the coupling to the system Hamiltonian, but is notoriously non-completely-positive. On the other hand,
the Davies equation preserves complete positivity but is valid only in the ultra-weak coupling limit and for systems with a finite level spacing, {which makes it incompatible with arbitrarily fast time-dependent driving}. 
Here we show that a recently derived {Markovian coarse-grained master equation (CGME)}, already known to be completely positive, has a much expanded range of applicability compared to the Davies equation, and moreover, is locally generated and can be generalized to accommodate arbitrarily fast driving. This generalization, which we refer to as the time-dependent CGME, is thus suitable for the analysis of fast operations in gate-model quantum computing, such as quantum error correction and dynamical decoupling. Our derivation proceeds directly from the Redfield equation and allows us to place rigorous error bounds on all three equations: Redfield, Davies, and coarse-grained. Our main result is thus a completely positive Markovian master equation that is a controlled approximation to the true evolution for any time-dependence of the system Hamiltonian, and works for systems with arbitrarily small level spacing. We illustrate this with an analysis showing that dynamical decoupling can extend coherence times even in a strictly Markovian setting.

\end{abstract}

\maketitle








\section{Introduction}
Modeling experiments requires taking into account that physical systems are open, i.e., not ideally isolated from their environments. Usually an environment contains many more degrees of freedom than the system itself, but is not in any interesting phase of matter where the detailed modeling of each of those degrees of freedom is required. This makes the problem of modeling open systems tractable. More precisely, the problem is tractable under a list of assumptions on the parameters of the bath (environment) and its interaction with the system~\cite{alicki_quantum_2007,Breuer:book,Gardiner:book}.

One natural approach is to consider the case when this interaction is weak compared to all other energy scales in the problem. In this limit, Davies showed~\cite{Davies:74} that the dynamics of the system are given by a Markovian master equation with what is now called Davies generators. Rigorous bounds on the error between the solution of this equation and the true evolution can be obtained [presented here in Eq.~\eqref{LblBound}]. The largest contribution to the error comes from making the rotating wave approximation (RWA).\footnote{This approximation is stated  explicitly in Ref.~\cite{Gardiner:book}, page 87, around Eq.~(3.6.67).} After the system evolves for a time set by a relevant relaxation timescale, the error in its density matrix is given by:
\beq
    \textrm{system-bath coupling strength}  \times ~\sqrt{\textrm{bath correlation time}/\textrm{system level-spacing}}\ .
    \label{eq:1}
\eeq
This quantity, in general, becomes exponentially large in the system size due its denominator, so that the 
coupling strength or the bath correlation time need to be sufficiently small for the use of Davies generators to be justified, up to some upper system size limit. 
In practice, for convenience and also because it may apply in a range that is larger than the rigorous but conservative bounds imply, the Lindblad master equation~\cite{Lindblad:76} with Davies generators is often used outside its formal range of applicability in modeling of experiments (e.g., Ref.~\cite{q-sig2,Boixo:2014yu,Albash:2015pd,Mishra:2018}). In such cases it should be understood as a semi-empirical model of open system dynamics, that possesses all the physical properties of the dynamics without reproducing them exactly. Unfortunately, since we naturally wish to apply modeling tools to problems for which we do not know a priori what the correct answer is, this approach cannot provide correctness guarantees. However, other methods are even more susceptible to this problem, e.g., a path-integral (non-master equation) approach based on integrating out the bath~\cite{Feynman:63,Caldeira1983,Makarov:94}: this involves integration over a very high-dimensional space, so methods~\cite{Eunji:2001} that bring it back into the realm of numerical feasibility usually involve uncontrollable approximations.

One of the important requirements of open system dynamics is complete positivity~\cite{Kraus:book}, under which density matrices whose eigenvalues are by definition non-negative are mapped to other such matrices, even when applied to a subsystem of a larger system.\footnote{It should be noted that complete positivity can be relaxed without losing physicality; see, e.g., Ref.~\cite{Dominy:2016xy}.}  
This property is equivalent to the Markovian master equation being in Lindblad canonical form~\cite{Lindblad:76,gorini_properties_1978}. There are numerous other master equations (e.g., Refs.~\cite{Nakajima:58,Zwanzig:60a,Redfield:66,PhysRevA.60.1944,Wonderen:2000yi,Lidar200135,Daffer:03,ShabaniLidar:05,ManPet06,piilo2008non,Breuer:2008aa,Whitney:2008aa,Wu09,Kossakowski:10,ABLZ:12-SI,Mozgunov:2016aa,Smirnov:2018,Kosloff:2018,Venuti:2018ab,McCauley:19,Benatti_2009,Benatti_2010,Merkli_2019}), some in Lindblad form and some not, whose domains of validity and ranges of applicability were not always studied in detail. One promising candidate is a coarse-grained master equation (CGME)~\cite{Majenz:2013qw}, 
which is in Lindblad form and thus automatically completely positive (CP). In this work we show that a suitably generalized variant of this equation{, that includes an arbitrarily time-dependent system Hamiltonian (hence called the ``time-dependent CGME''),} has a much milder error than~\eqref{eq:1}:
\beq
    \textrm{system-bath coupling strength}\times \textrm{bath correlation time}\ ,
    \label{eq:2}
\eeq
which does not diverge with the system size $n$ for local observables, and has only a polynomial [$O(n^\alpha)$ with $1/2 \leq \alpha \leq 1$] prefactor for nonlocal observables. We note that the full form of the error estimate~\eqref{eq:2} is not uniform over the evolution time, unlike some of the error bounds developed elsewhere~\cite{Merkli_2019,Cubitt:2015}. We also show that this new master equation is an improvement over the Lindblad equation with Davies generators in that it has a local structure, compatible with approximate simulation methods with asymptotically smaller costs [$O(n)$ \textit{vs}. $O(\exp(n))$], and that the error estimate~\eqref{eq:2} is valid for an arbitrary time-dependence of the system Hamiltonian. In particular, the time-dependent CGME is 
compatible with the assumption of arbitrarily fast gates, often made in the circuit model of quantum computing, e.g., in the analysis of fault-tolerant quantum error correction~\cite{Knill:05}. This assumption was shown to be incompatible with the derivation of the Davies master equation~\cite{PhysRevA.73.052311}.

This paper is structured as follows. We summarize our main results in Sec.~\ref{sec:Three}, where we present the main master equations involved in this work: Redfield and Davies-Lindblad (well known), and coarse-grained, both time-independent and time-dependent (new). We also present bounds that provide the range of applicability of each type of master equation (new). We show that these bounds can be expressed entirely in terms of only two timescales, namely the fastest system decoherence timescale, and the characteristic timescale of the decay of the bath correlation function.
The reader interested only in the results and not in the details of derivations can safely read only this section and then skip to the conclusions in Sec.~\ref{sec:conc}. Derivations begin in Sec.~\ref{deriv1}, where we present a simple, new derivation of the time-independent CGME. We express the equation in CP (Lindblad) form, state the range of validity for different approximations made in the derivation, and thus for the equation itself. We also compare the range of applicability with other master equations, 
and note the spatial locality of the Lindblad generators of the CGME. At this point we are ready to address the case of time-dependent Hamiltonians. The derivation of the equation for this case (which happens to result in exactly the same form) is given in Sec.~\ref{TD}. In particular, this equation is well suited for the simulation of open system adiabatic quantum computation, 
but also for dynamical decoupling, which involves the opposite limit of very fast system dynamics. We note that while we will sometimes refer to qubits, none of the results we discuss in this work are limited to qubits, and in fact any finite multi-level system, or interacting set of such systems, is captured by the formalism. We study some applications and examples in Sec.~\ref{scalSec}, including error suppression using a dynamical decoupling protocol, and
numerical results of the comparison of our master equation with the Redfield and Lindblad-Davies master equations. In the remainder of the paper we derive the range of validity of the various master equations we study. We give a detailed treatment of the Born approximation and the other approximations involved, in terms of rigorous bounds presented in Sec.~\ref{errSec}, where we also discuss the generalization to multiple coupling terms and analyze the scaling of the error bounds for large system sizes. 
Conclusions are presented in Sec.~\ref{sec:conc}.
Various additional technical details are given in the Appendix.

\section{Three master equations}
\label{sec:Three}

This section provides a summary of our main results. We first present a brief background to define basic concepts and notation, followed by the definition and properties of the two main timescales we use later to provide bounds and ranges of applicability. We then summarize the three master equations we focus on in this work, followed by upper bounds on the distance of their solutions from the true state. 

\subsection{Background}

Consider a system interacting with its environment as described by the total Hamiltonian 
\begin{equation}
H_{\textrm{tot}} = H\otimes I_\text{b} + A\otimes B + I\otimes H_\text{b}\ .
\label{eq:Htot}
\end{equation}
Here $H$ is the time-independent system Hamiltonian (we treat the time-dependent case in Sec.~\ref{TD}), $H_\text{b}$ is the bath Hamiltonian, and $A$ and $B$ are, respectively, Hermitian system and bath operators. The more general situation, with $V=A\otimes B$ replaced by $\sum_i A_i\otimes B_i$, is easily treated as well (see Sec.~\ref{nTerms}), but we avoid it here to simplify the notation. We choose $A$ to be dimensionless with $\|A\| =1$ [the operator norm is defined in Eq.~\eqref{opaNorm}] and $B$ to have energy units, but we deliberately do \emph{not} set $\|B\| =1$, since we wish to account for baths (such as harmonic oscillator baths) for which $\|B\|$ can be infinite. We also set $\hbar\equiv 1$ throughout, so that energy and frequency have identical units.

Let the eigenstates of $H$ be $\{|n\rangle\}$ , and the corresponding eigenvalues $\{E_n\}$. Equivalently, $H = \sum_n E_n \Pi_n$, where $\Pi_n$ is a projector onto the eigenspace with eigenvalue $E_n$, and $\Pi_m\Pi_n = \delta_{mn}\Pi_n$. 
Note that in the interaction picture $A(t) = U^\dagger(t) A U(t)= \sum_{nm} A_{nm} e^{-iE_{mn}t} |n\rangle \langle m |$ where $E_{mn} = E_m - E_n$, $A_{nm} = \bra{n}A\ket{m}$, and $U$ is the solution of $\dot{U}=-iHU$. 
Let 
\beq
A_{\omega} = \sum_{mn: E_{mn} = \omega} A_{nm}|n\rangle \langle m |  = \sum_{mn: E_{mn} = \omega}\Pi_n A \Pi_{m} = A_{-\omega}^\dag\ , 
\label{eq:A_om}
\eeq  
so that\footnote{The choice of the sign for this notation, as well as other notation choices, follow the textbook~\cite{Breuer:book}, p.133-134.} 
\begin{equation}
A(t) = \sum_\omega A_{\omega}e^{-i\omega t}\ .
\label{eq:A(t)}
\end{equation}
Here $\omega$ runs over all the energy differences (Bohr frequencies) between the eigenvalues of $H$. 

We assume henceforth that the bath is stationary --- $[\rho_\text{b},H_\text{b}]=0$ where $\rho_\text{b}$ is the bath state --- and introduce the bath correlation function $\mathcal{C} (t)$ and its Fourier transform 
$\gamma(\omega)$, known as the bath spectral density:
\bes
\begin{align}
    \mathcal{C}(t) &= \Tr[\rho_\text{b} B(t) B] = \mathcal{C} ^*(-t) = \frac{1}{2\pi}\int_{-\infty}^{\infty} e^{-i\o t} \gamma(\o)d\o
    \label{eq:C(t)} \\
\label{eq:gamma31}
\g(\o) &= \int_{-\infty}^\infty \mathcal{C}(t)e^{i\omega t}dt = f(\o) + f^*(\o) \ ,\end{align}
\ees
where $B(t) = e^{iH_\text{b} t} B e^{-iH_\text{b} t}$ and 
\beq
f(\o) = \int_0^\infty C(t)e^{i\o t} dt = \frac{1}{2}\gamma(\omega)+iS(\omega) \ , \quad S(\omega) = \frac{1}{2i}\left[f(\omega) - f^*(\omega)\right] = S^*(\omega)
\ .
\label{eq:6c}
\eeq
Note that the real-valued functions $S(\o)$ and $\g(\o)$ are related via a Kramers-Kronig transform
\beq
S(\o)= \frac{1}{2\pi} \int_{-\infty}^\infty \mc{P}\left(\frac{1}{\o-\o'}\right) [\gamma(\omega')] d \omega'\ , 
\label{eq:S-gamma}
\eeq
where $\mc{P}\left(\frac{1}{x}\right)[f] = \lim_{\epsilon\to 0} \int_{-\epsilon}^\epsilon \frac{f(x)}{x}dx$ is the Cauchy principal value.%
\footnote{This follows from the identity $\int_{0}^\infty d\tau e^{ix\tau} = \pi \delta(x) + i \mc{P}\left(\frac{1}{x}\right)$.}
When $A\otimes B$ is replaced by $\sum_i A_i\otimes B_i$, we note the positivity of $\gamma_{ij}(\omega) = \int dt e^{i\omega t}\Tr[\rho_\text{b} B_i(t) B_j] $ as a matrix in the indices $i,j$ (of course this reduces to $\gamma(\o) \geq 0$ in the scalar case). This is a non-trivial fact that is usually associated with Bochner's theorem, as e.g., in the textbook~\cite{Breuer:book}. In Appendix~\ref{app:A} we give two different proofs, one using Bochner's theorem and another that is direct. If we assume not only that the bath state is stationary, but that it is also in thermal equilibrium at inverse temperature $\beta$, i.e., $\rho_\text{b}=e^{-\beta H_\text{b}}/\mathcal{Z}$, $\mathcal{Z}=\Tr e^{-\beta H_\text{b}}$, then it follows that the correlation function satisfies the time-domain {Kubo-Martin-Schwinger (KMS) condition}~\cite{Breuer:book,Lidar:2019aa}:
\beq
C(t) = \Tr[\rho_\text{b} B(0) B(t+i\beta)]\ .
\label{eq:KMS-C}
\eeq
If in addition the correlation function is analytic in the strip between $t=-i\b$ and $t=0$, then it follows that the Fourier transform of the bath correlation function satisfies the {frequency domain KMS condition}~\cite{Breuer:book,Lidar:2019aa}:
\beq
\gamma(-\omega) = e^{-\beta \omega} \gamma(\omega)\ \forall \o \geq 0\ .
\label{eq:KMS}
\eeq
We note that the thermal equilibrium assumption is not necessary for the results we derive in this work, and we never use it in our proofs. We mention it here for completeness, and also since we use it in one of our dynamical decoupling examples later on. Finally, note that $\mathcal{C} (t)$ has units of frequency squared.


\subsection{The bath correlation time $\tau_B$ and the fastest system decoherence timescale $\tau_{SB}$}
\label{sec:times}

We define the two key quantities
\bes
\label{eq:T1tauB}
\begin{align}
\label{eq:T1tauB-a}
   \frac{1}{\tau_{SB}} &= \int_0^{\infty}|\mathcal{C} (t)|dt \\
\label{eq:T1tauB-b}
\tau_B &= \frac{\int_0^{T}t|\mathcal{C} (t)|dt}{\int_0^{\infty}|\mathcal{C} (t)|dt}
\end{align}
\ees
Here $T$ is the total evolution time, used as a cutoff which can often be taken as $\infty$. We  show later that in the interaction picture $\|\dot{\rho}\|_1\leq 4\|\rho\|_1/\tau_{SB}$, where $\rho$ is the system density matrix and the trace norm defined in Eq.~\eqref{treNorm}, so that we can interpret $\tau_{SB}$ as the fastest system decoherence timescale, or timescale over which $\rho$ changes due to the coupling to the bath, in the interaction picture (i.e., every other system decay timescale, including the standard $T_1$ and $T_2$ relaxation and dephasing times, must be longer). The quantity $\tau_B$ is the characteristic timescale of the decay of $\mathcal{C} (t)$. We return in Sec.~\ref{errSec} to why we define these timescales in this manner, but note that Eq.~\eqref{eq:T1tauB-b} becomes an identity if we choose $|\mc{C}(t)| \propto e^{-t/\tau_B}$ and take the limit $T\to\infty$. 

We further note that $\tau_{SB}$ and $\tau_B$ are the only two parameters relevant for determining the range of applicability of the various master equations; in particular, nothing about the norm or time-dependence of $H_{\text{tot}}$ affects the accuracy of our time-dependent CGME, given below. In particular, $H_{\text{tot}}$ can be arbitrarily large or small in the operator norm (strong coupling, unbounded bath), or have an arbitrarily large derivative (non-adiabatic regime). This remark is important in light of previous approaches, so we expand on it some more.

In previous work it was common to extract a dimensionful coupling parameter $g$, i.e., to replace $B$ by $g\tilde{B}$, where $\|\tilde{B}\|=1$. One then defines $\tau_B = \int_0^\infty |\tilde{C}(t)|dt$ in terms of a dimensionless correlation function $\tilde{C}(t) = C(t)/g^2$~\cite{ABLZ:12-SI}. One problem with this approach is the arbitrariness of distributing the numerical factors between $g$ and $\tilde{B}$. Another is that it precludes unbounded baths, such as oscillator baths, for which $\|B\|=\infty$. In contrast, the formalism we present here is applicable even when $\|B\|$ diverges. Furthermore, 
$\|B\|$ contains an extra scale that does not carry any new information about the range of applicability of the master equations we discuss and derive. By not introducing this extra scale we highlight that there are only two free parameters in our analysis of the error: $\tau_B$ and $\tau_{SB}$, and no other information about the bath is needed. I.e., even though different baths will lead to different equations, our results are universal for any bath with the same values of these two parameters (in fact only their dimensionless ratio matters).

%

Note that we can relate $\tau_B$ and $\tau_{SB}$ to the spectral density. First, using $\g(\o)>0$ and $C(t) = C^*(-t)$:
\beq
\g(\o) = |\g(\o)| \leq \int_{-\infty}^\infty |C(t)|dt = 2 \int_{0}^\infty |C(t)|dt = \frac{2}{\tau_{SB}} \ .
\label{eq:22new}
\eeq
The dephasing time $T_2$ and relaxation time $T_1$ for a single qubit are usually defined as $\propto 1/\g(0)$ and $\propto 1/\g(\o)$ respectively, where $\omega$ is the qubit operating frequency (e.g., see Ref.~\cite{Albash:2015nx}), so we see that $\tau_{SB} \lesssim T_1,T_2$, which illustrates our earlier comment that $\tau_{SB}$ is the fastest system decoherence time-scale.
Second, $\g'(\o) = i \int_{-\infty}^\infty t C(t) e^{i\o t} dt$, so that
\beq
|\g'(\o)| \leq \int_{-\infty}^\infty |t| |C(t)| dt = 2 \int_{0}^\infty t |C(t)| dt \stackrel{\tiny{T\to\infty}}{=} 2\frac{\tau_B}{\tau_{SB}}\ ,
\label{eq:23}
\eeq
where the limit is taken in Eq.~\eqref{eq:T1tauB-b}.
And lastly, 
$\frac{d\g(-\omega)}{d\o} =-\frac{d\g(\omega)}{d\o}|_{-\omega}$, while differentiating the KMS condition yields $\frac{d\g(-\omega)}{d\o} = -\b e^{-\b \o} \g(\o) +e^{-\b \o}\frac{d\g(\omega)}{d\o} $, so that:
\begin{equation}
    \gamma'(0) = \frac12 \beta \gamma(0) \ .
\end{equation}
This implies that $\gamma'(0) >0$, so that under the KMS condition~\eqref{eq:KMS} we can replace Eq.~\eqref{eq:23} by
\beq
\beta \gamma(0) \leq 4\frac{\tau_B}{\tau_{SB}} \ .
\label{eq:25new}
\eeq
The upper bounds on the approximation errors of all the master equations we present below involves the ratio $\tau_B/\tau_{SB}$ [see, e.g., Eqs.~\eqref{err:Redfield}-\eqref{eq:59b}]. The smallness of this ratio is sufficient for the CGME to be useful, while additional assumptions are required for the Davies-Lindblad equation and some versions of the Redfield equation.
It follows from Eq.~\eqref{eq:25new} that for our bound to be small it is necessary that the temperature $\b^{-1}$ is not too low and/or the spectral density at $\o=0$ (roughly the same as the dephasing rate $T_2^{-1}$) is not too large. 
The simple condition $\beta \gamma(0) \ll 1$ already rules out diverging spectral densities such as the case of $1/f$ noise, though this can be ameliorated by introducing a low-frequency cutoff. 

We now present the three master equations, in the Schr\"odinger picture.

\subsection{Redfield master equation} 
\label{sec:Redfield}

For the derivation see Sec.~\ref{sec:Redfield-deriv}.
This equation was known earlier than the Davies-Lindblad equation~\cite{Redfield:66}, and is often used in quantum chemistry. It  is not CP and hence cannot be put in Lindblad form:
\begin{equation}
{\dot\rho}_R(t) =-i[H,\rho_R(t)] +(A \rho_R A_f - \rho_R  A_f A +\textrm{h.c.})\ ,
\label{eq:Red}
\end{equation}
where we defined the ``filtered" operator: 
\begin{equation}
    A_f =\int_{0}^\infty \mathcal{C} (-t)A(-t)dt=  \sum_\omega A_\omega \int^{\infty}_0 \mathcal{C} (-t) e^{i\omega t}dt =  \sum_\omega f^*(-\o) A_\omega \ ,
    \label{eq:Af}
\end{equation}
where $f(\o)$ is defined in Eq.~\eqref{eq:6c}, and we used the subscript $R$ to denote that the density matrix satisfies the Redfield equation (we use a similar subscript notation below to distinguish the solutions of different master equations). There is no natural way to separate the Lamb shift. One of the benefits of the Redfield equation is that there is only one generator $A_f$ per interaction with the environment $A\otimes B$. We will show that as long as it preserves positivity, the Redfield equation is more accurate than the Davies-Lindblad master equation.


\subsection{Davies-Lindblad master equation} 

For the derivation see Sec.~\ref{app:derivAll}.
This is the familiar result~\cite{Davies:74}
\bes
\label{eq:Davies-Lind}
\begin{align}
{\dot\rho}_D(t) &=-i[H +H_{\mathrm{LS}} ,\rho_D] +\sum_{\omega} \gamma(\omega)  (A_{\omega} \rho_D A^\dagger_{\omega} - \frac{1}{2}\{\rho_D, A^\dagger_{\omega} A_{\omega} \}) \ ,
\\
    H_{\mathrm{LS}} &= -\sum_\omega S(\o) A^\dagger_{\omega}A_{\omega}\ , \quad S(\o) = \frac{1}{2i} \int_{-\infty}^\infty \sgn(t)\mathcal{C}(t)e^{i\omega t} dt \ ,
    \label{eq:LS-final}
\end{align}
\ees
where $H_{\mathrm{LS}}$ is the Lamb shift. We note that ad hoc forms of the Lindblad equation are often written down and used without justification, e.g., with just one Lindblad term per qubit (e.g., $\sigma_-$ or $I-\sigma_z$~\cite{Znidaric:10,Medvedyeva:16}). In reality the Davies generators are derived from first principles, i.e., from the description of the total Hamiltonian of the system and bath. The number of Davies generators is unfortunately exponential in the number of qubits $n$: $\sum_\omega$ is over all $4^n$ energy differences.

\subsection{Coarse-grained master equation (CGME) for time-independent or time-dependent system Hamiltonians} 

This is our main new set of results, generalizing Ref.~\cite{Majenz:2013qw} to the time-dependent case. 

\subsubsection{The time-independent case}

For the derivation see Sec.~\ref{sec:CGME}.
First, considering time-independent system Hamiltonians, we obtain:
\begin{equation}
{\dot\rho}_C(t) =-i[H +H_{\mathrm{LS}},\rho_C(t)] + \int_{-\infty}^{\infty}d\epsilon \left(A_{\epsilon} \rho_C(t) A_{\epsilon}^\dag - \frac{1}{2}\left\{ \rho_C(t) , A_{\epsilon}^\dag A_{\epsilon}\right\} \right)\ ,
\label{eq:Lind}
\end{equation}
where the Lindblad operators are
\begin{subequations}
    \label{eq:newLindops}
\begin{align}
    \label{Atime}
    A_\epsilon &=  \sqrt{\frac{\gamma(\epsilon)}{2\pi T_a}}\int_{-T_a/2}^{T_a/2}e^{i\epsilon t}A(t)dt \\
    &= \sum_\omega f(\epsilon,\omega) A_\omega \ , \quad f(\epsilon,\omega) =  \sqrt{\frac{\gamma(\epsilon)T_a}{2\pi}} \textrm{sinc}[T_a(\epsilon-\omega)/2] \ ,
\end{align}
\end{subequations}
with $\sinc(x) \equiv \sin(x)/x$, and where $T_a$ is the averaging, or coarse-graining time, discussed below. The continuous family of terms labeled by $\epsilon$ is an unusual form of the Lindblad equation, but is manifestly CP. The more standard form in terms of a discrete sum over transition frequencies is equivalent to this one, and is given in Eq.~\eqref{eq:CGME} below. The range of applicability of the CGME is only slightly smaller than that of the Redfield ME, as we discuss in Sec.~\ref{rangesDiscuss}. Much like for Davies generators, there are exponentially many frequency differences, but unlike the Davies case the discretization of the continuous integral $\int d\epsilon$ makes it possible to keep the number of generators constant in the system size for each $A\otimes B$ term. These results are presented in Sec.~\ref{posLoc}.

The same applies for the Lamb shift, which is 
\begin{subequations}
\label{eq:HLS-CG}
\begin{align}
    H_{\mathrm{LS}}&=\sum_{\omega\omega'}F_{\omega\omega'}A_{\omega'}A_{\omega}\\
    F_{\omega\omega'}  &=\frac{1}{2T_a\o_+} \Re  \int_{0}^{T_a} d\theta \left( e^{i(\o\theta-T_a\o_+)} - e^{-i(\o'\theta-T_a\o_+)} \right) \mathcal{C} (\theta) = F_{\omega\omega'}^* = F_{-\omega',-\omega}\ ,
    \label{eq:Fw1w2}
\end{align}
\end{subequations}
where $\omega_+ = \frac{\omega + \omega'}{2}$.
Note that $F_{\omega\omega'}$ is well defined in the limit $\omega_+ \to 0$. The results above are generalized to the $n$-qubit setting in Sec.~\ref{nTerms}.

\subsubsection{The time-dependent case}

For the derivation see Sec.~\ref{TD}.
In the case of a time-dependent system Hamiltonian $H(t)$ we obtain the same form, and even the same range of applicability, except that all the operators are now time-dependent:
\begin{equation}
\label{tdCoarse}
{\dot\rho}_C(t) =-i[H(t) +H_{\mathrm{LS}}(t),\rho_C] + \int_{-\infty}^{\infty}d\epsilon \left(A_{\epsilon}(t) \rho_C A_{\epsilon}^\dag(t) - \frac{1}{2}\left\{ \rho_C , A_{\epsilon}^\dag(t) A_{\epsilon}(t)\right\} \right)\ ,
\end{equation}
where the time-dependent Lindblad operators are
\begin{equation}
\label{eq:14}
    A_\epsilon(t) = \sqrt{\frac{\gamma(\epsilon)}{2\pi T_a}}\int_{-T_a/2}^{T_a/2} e^{i\epsilon t_1}A(t+t_1,t)dt_1  \ ,
\end{equation}
where $A(t',t) = U^\dagger (t',t)A U(t',t)$ with $U(t',t) = \mathcal{T} \exp[-i\int_t^{t'} H(s) ds]$ (the forward time-ordered exponential, from $t'$ on the left to $t$ on the right), and the time-dependent Lamb shift is
\beq
H_{\mathrm{LS}}(t) = \frac{i}{2T_a}\int_{-T_a/2}^{T_a/2}dt_1\int_{-T_a/2}^{T_a/2}  dt_2 \sgn(t_1-t_2)\mathcal{C} (t_2-t_1)  A(t+t_2,t) A(t+t_1,t)\ .
\label{eq:15}
\eeq
Although the coarse-graining time $T_a$ is a free parameter, we show that 
to minimize the upper bound we derive on the distance between the solutions of the Redfield and coarse-grained master equations
it is nearly optimal to choose 
\begin{equation}
    T_a = \sqrt{\tau_{SB} \tau_B/5}\ .
    \label{eq:T_a-opt}
\end{equation}

\subsection{Error bounds and range of applicability}
\label{sec:err-bounds}
The error bound of the Redfield master equation is
\beq
   \|\rho_{\text{true}}(t) -\rho_R(t) \|_1 \le O\left(\frac{\tau_B}{\tau_{SB}}   e^{12t/\tau_{SB}}\right)\text{ln}\left(\frac{\tau_{SB}}{\tau_B} \right)\ , 
   \label{err:Redfield}
\eeq
where $\rho_{\text{true}}(t)$ denotes the true (approximation-free) state.

The error bound of the Davies-Lindblad master equation is
\beq
   \|\rho_{\text{true}}(t) -\rho_D(t) \|_1 \le  O\left( \left(\frac{\tau_B}{\tau_{SB}}  +\frac{1}{\sqrt{\tau_{SB}\delta E}}\right)e^{12t/\tau_{SB}}\right)\ ,
       \label{eq:72}
\eeq
where $\delta E = $min$_{i\ne j}|E_i-E_j|$ is the level spacing, with $E_i$ the eigenenergies of the system Hamiltonian $H$. 

The error bound of both the time-independent and time-dependent 
coarse-grained master equation is:\footnote{Strictly, our proof is only for the time-independent case, but there do not seem to be any obstacles for its generalization to the time-dependent case.}
\begin{align}
    \|\rho_{\text{true}}(t) -\rho_{C}(t) \|_1 \le O\left(\sqrt{\frac{\tau_B}{\tau_{SB}}}e^{6t/\tau_{SB}}\right) \ .
   \label{eq:59b}
\end{align}
In the above expressions, $O(X)$ is understood for $X\to 0$, specifically for CGME $X =\sqrt{\tau_B/\tau_{SB}}e^{6t/\tau_{SB}}$. There is a subtlety in the definition of big-$O$ notation that we would like to emphasize. By definition, “$f(X) = O(X)$ at $X\to0$” means there exist $X_0, M$ s.t. for any $X\leq X_0$ we have $f(X)\leq MX$.
Applying this to Eq.~\eqref{eq:59b}, we have the following unpacking of the above statement:
there exist universal constants $x_0, M$ such that if $\sqrt{\tau_B/\tau_{SB}}\leq X_0$ and $t\leq (\tau_{SB}/6)\text{ln}(X_0/\sqrt{\tau_B/\tau_{SB}})$, then $\|\rho_{\text{true}}(t) -\rho_{C}(t) \|_1 \le M\sqrt{\tau_B/\tau_{SB}}e^{6t/\tau_{SB}}$ holds. Note that for sufficiently small $\tau_B/\tau_{SB}$, a dimensionless quantity $t/\tau_{SB}$ can have an arbitrarily large constant value in this bound. The error can be made arbitrarily small at any fixed value of $t/\tau_{SB}$ just by choosing a sufficiently small $\tau_B/\tau_{SB}$.  This is what is commonly referred to as the weak coupling limit.
We emphasize that the constants $X_0, M$ involved in this statement are universal, i.e. independent of the parameters of the equation such as the Hamiltonian and the initial state. In that sense, our bound is uniform. It does depend on time $t$ as explicitly stated, so it is not uniform in time.\footnote{Note that in both the Davies-Lindblad and CGME cases the l.h.s. is bounded above by $2$ due to the positivity of the density matrices; this is not necessarily true in the Redfield case.} Thus, comparing these bounds, we observe that the Redfield bound is the tightest, which is natural given that it involves the fewest approximations. However, recall that the Redfield master equation is not CP. When comparing the bounds on the solutions of the two CP master equations, we observe that unlike the Davies-Lindblad equation the CGME does not involve the energy gap, which means that the CGME in principle has a much larger range of applicability, in particular for systems whose gap shrinks with growing system size.\footnote{By range of applicability we mean the range of parameters over which the approximation is accurate. This range can be defined formally in terms of an upper bound on the right-hand side (r.h.s.), e.g., $\sqrt{{\tau_B}/{\tau_{SB}}} e^{6t/\tau_{SB}} < \epsilon$ for some fixed $\epsilon<1$.}

While the equations and inequalities can be derived for any bath, we made some extra assumptions to prove the error bounds as presented above. Specifically, we assumed a Gaussian bath and the convergence of our series expansion for the Born error, as defined and discussed in detail in Sec.~\ref{Born}. For a non-Gaussian bath, extra time scales relating to higher correlation functions generally appear, and the error bound can be derived by a straightforward generalization of our approach. It will contain, besides $\tau_B$ and $\tau_{SB}$, additional terms involving those time scales.

\section{(Re-)Derivation of the coarse-grained master equation}
\label{deriv1}

Our goal in this section is to present a compact and simplified derivation of the CGME found in Ref.~\cite{Majenz:2013qw}. This master equation is in Lindblad form but avoids making use of the RWA, which contributes the largest approximation error. Instead, the key idea is to introduce a coarse-graining (averaging) timescale, first exploited in this context in Refs.~\cite{PhysRevA.60.1944,Lidar200135}. In this section we consider the case of a time-independent system Hamiltonian, and later generalize our derivation to the time-dependent case, which includes adiabatic evolution as a special case. 

Along the way we also derive the Redfield and Davies-Lindblad master equations. 

\subsection{Setting up}
\label{sec:derivation}

We start from the total Hamiltonian given in Eq.~\eqref{eq:Htot}, and let $V=A\otimes B$. Recall that $A$ is dimensionless and $B$ has units of energy.
The first few steps in our derivation are standard~\cite{Breuer:book}. In the double system-bath interaction picture $V(t) = U_{\textrm{0}}^\dagger(t) V U_{\textrm{0}}(t)$ [where $U_{\textrm{0}}$ is the solution of $\dot{U}_{\textrm{0}} = -i(H\otimes I_\text{b} + I\otimes H_\text{b})U_{\textrm{0}}, ~ U_{0}(0) = 1$], and the state satisfies
\bes
\label{eq:16}
\begin{align}
\label{eq:16a}
{\dot\rho}_{\text{tot}}(t) &=-i[V(t),\rho_{\text{tot}}(t)]\\
 \rho_{\text{tot}}(t) &=\rho_{\text{tot}}(0) - i \int_0^t [V(\tau),\rho_{\text{tot}}(\tau)]d\tau .
\label{eq:16b}
\end{align}
\ees
Substituting this back into the original equation yields:
\begin{equation}
{\dot\rho}_{\text{tot}}(t) =-i[V(t),\rho_{\text{tot}}(0)]-[V(t),\int_0^t [V(\tau),\rho_{\text{tot}}(\tau)]]d\tau .
\end{equation}
The reduced system state in the interaction picture is defined as $\rho_{\textrm{true},I}(t) = \Tr_\text{b}[\rho_{\text{tot}}(t)]$, where the subscript ``true" denotes that this is the correct, completely approximation-free state. The corresponding true system state in the Schr\"odinger picture is $\rho_{\textrm{true}} = U_0(t) \rho_{\textrm{true},I} U^\dag_0(t)$. The Born approximation,
\beq
\rho_{\text{tot}}(t)= \rho_{\textrm{true,I}}(t)\otimes \rho_\text{b} + \delta\rho_{\text{tot}} ,
\eeq 
together with the shift of $B$ such that $\Tr[\rho_\text{b}B] =0$ allows one to write:
\begin{equation}
{\dot\rho}_{\textrm{true},I}(t) =-\Tr_\text{b}[A(t)\otimes B(t),\int_0^t [A(\tau)\otimes B(\tau),\rho_{\textrm{true,I}}(\tau) \otimes\rho_\text{b}]]d\tau  + \mathcal{E}_B \ ,
\label{eq:Born20}
\end{equation}
which can be understood as the definition of the Born approximation error $\mathcal{E}_B$:
\beq
\mathcal{E}_B \equiv {\dot\rho}_{\textrm{true},I}(t) + \Tr_\text{b}[A(t)\otimes B(t),\int_0^t [A(\tau)\otimes B(\tau),\rho_{\textrm{true,I}}(\tau) \otimes\rho_\text{b}]]d\tau , \quad \|\mathcal{E}_B\|_1 = O\left(\frac{\tau_B}{\tau_{SB}^2}\right)\ .
\label{eq:Born21}
\eeq
Henceforth we assume that the bath correlation function decays rapidly, namely that $\tau_B \ll \tau_{SB}$, where $\tau_{SB}$ and $\tau_B$ were defined in Eq.~\eqref{eq:T1tauB}.

Now, let $\rho_{B,I}$ denote the solution of the master equation after the Born approximation:
\begin{equation}
{\dot\rho}_{B,I}(t) =-\Tr_\text{b}[A(t)\otimes B(t),\int_0^t [A(\tau)\otimes B(\tau),\rho_B(\tau) \otimes\rho_\text{b}]]d\tau \ .
\end{equation}
Later we collect other errors as letters next to $B$. The error estimate given in Eq.~\eqref{eq:Born21} for $\|\mathcal{E}_B\|_1$ is derived in Sec.~\ref{Born} [Eq.~\eqref{eq:193b}]. 

We digress briefly to carefully explain the meaning of the norms and big-$O$ notation used here, since this is important for the remainder of this work. $\mathcal{E}_B$ is an operator acting on the system Hilbert space. 
The trace norm $\|A\|_1$ is:
\begin{equation}
    \|A\|_1 \equiv \Tr\sqrt{A^\dagger A} \label{treNorm}
\end{equation}
The operator norm $\|X\|$ is defined as follows:
\begin{equation}
    \|X\| = \max_i \lambda_{X,i} \label{opaNorm}
\end{equation}
where $\lambda_{X,i}$ are the eigenvalues of $|X|\equiv \sqrt{X^\dag X}$ (singular values of $X$) indexed by $i$.
The big-$O$ notation $\mathcal{E}=O(x)$ is taken for times such that $t/\tau_{SB} \leq M$ and $x\to 0$, i.e., there exist an $M$-dependent number $\epsilon_M$ and an $M$-dependent constant $C_M$  such that for any $x< \epsilon_M$ the error $\|\mathcal{E}\|_1 \leq C_M x$. 

It is natural to use the trace norm for density matrices. Indeed, if we want to find the deviation in an observable $O$ relative to the difference between two states, $\delta\rho = \rho_1-\rho_2$, then
\begin{equation}
   | \Tr O\delta \rho| \leq 2^n \|O\| \|\delta \rho\|, \quad  | \Tr O\delta \rho| \leq \|O\| \|\delta \rho\|_1\ .
\end{equation}
The second expression gives a tighter bound, if we manage to have the same bound on $\|\delta \rho\|_1$ as on $\|\delta \rho\|$. Fortunately, this will turn out to be the case in this work, and was already observed for the Markov error in Ref.~\cite{ABLZ:12-SI}. The key property used to prove the inequality above is submultiplicativity, valid for any unitarily invariant norm~\cite{Bhatia:book}:
\begin{equation}
     \|AB\|_1 \leq \|A\|\|B\|_1\ .
     \label{eq:submult}
\end{equation}

For simplicity, we assume the bath state to be stationary, such that $\mathcal{C} (t,t') = \Tr[\rho_B B(t) B(t')]$ is invariant with respect to shifts in time of both arguments.\footnote{The derivation can potentially be extended to the general case of a two-time correlation function, i.e., without assuming translational invariance; indeed the derivation in \cite{Majenz:2013qw} does not.} We can then replace $\mathcal{C} (t,t')$ by $\mathcal{C} (t) = \Tr[\rho_B B(t)B]$, as in Eq.~\eqref{eq:C(t)}. Recall that $\mathcal{C} ^*(t) = \mathcal{C} (-t)$. After expanding the commutators, one arrives at
\begin{equation}
{\dot\rho}_{B,I}(t) =\int_0^t \mathcal{C} (\tau-t) [A(t) \rho_{B,I}(\tau) A(\tau) - \rho_{B,I}(\tau) A(\tau) A(t)] d\tau +\textrm{h.c.} \equiv \int_0^t K^{B,2}_{t-\tau} (\rho_{B,I}(\tau))d\tau\ ,
\label{eq:rho_Born}
\end{equation}
where $K^{B,2}_{t-\tau}$ is a superoperator.
Taking the trace norm, we note that:
\beq
\|\dot{\rho}_{B,I} \|_1 
\leq 4c_B \int_0^\infty |\mathcal{C} (t)| dt =4c_B/\tau_{SB} ,
\label{eq:|L|B}
\eeq
where we used our earlier choice of setting $\|A\| =1$, and the constant $c_B$ is chosen as an upper bound on $\|\rho_{B,I}\|_1$ (if this were a CP map, $c_B=1$ would hold). Under the condition $\forall t, ~ \|\rho_{\text{test}}(t)\|_1=1$ we have
\begin{equation}
   \left\| \int_0^t K^{B,2}_{t-\tau} (\rho_{\text{test}}(t))d\tau\right\|_1 \leq 4/\tau_{SB} \label{eq:|L|}\ .
\end{equation}

We next introduce the Markov approximation $\rho_{B,I}(\tau) \mapsto \rho_{B,I}(t)$:
\bes
\begin{align}
{\dot\rho}_{B,I}(t) &=\int_0^t \mathcal{C} (\tau-t) [A(t) \rho_{B,I}(t) A(\tau) - \rho_{B,I}(t) A(\tau) A(t)] d\tau +\textrm{h.c.} +\mathcal{E}_M, \quad \|\mathcal{E}_M\|_1=O\left(\frac{\tau_B}{\tau_{SB}^2}\right) \\
{\dot\rho}_{BM,I}(t) &=\int_0^t \mathcal{C} (\tau-t) [A(t) \rho_{BM,I}(t) A(\tau) - \rho_{BM,I}(t) A(\tau) A(t)] d\tau +\textrm{h.c.} =\mathcal{L}^{BM,I}_t(\rho_{BM,I}(t))
\label{RedStart}
\\ &\qquad\qquad\qquad\qquad\qquad\qquad\qquad\qquad\forall t, ~ \|\rho_{\text{test}}\|_1=1 : ~ \|\mathcal{L}^{BM,I}_t(\rho_{\text{test}})\|_1 \leq 4/\tau_{SB} \label{newLambda}
\end{align}
\ees
Equation~\eqref{RedStart} is the Redfield equation~\cite{Breuer:book}, 
which is notoriously non-CP  (though various fixes have been proposed \cite{Gaspard:1999aa,Whitney:2008aa}). The Markov approximation error $\mathcal{E}_M$ is of same order as the Born approximation error given in Eq.~\eqref{eq:Born21} (for more detail see Sec.~\ref{errSec}). 

We note that the errors on the r.h.s. are not generally additive: if we try to write a Markovian master equation for the true (approximation-free) evolution $\rho_{\text{true},I}$, then the error on the r.h.s. will, besides $ \mathcal{E}_B +\mathcal{E}_M$, contain a correction to the Markov error from the Born error:
\bes
\begin{align}
\label{errorAd}
    {\dot\rho}_{\text{true},I}(t) =\int_0^t \mathcal{C} (\tau-t) [A(t) \rho_{\text{true},I}(t) A(\tau) - \rho_{\text{true},I}(t) A(\tau) A(t)] d\tau +\textrm{h.c.} +\mathcal{E}_B + \mathcal{E}_M(\rho_{\text{true}}), \\ 
   \mathcal{E}_M(\rho_{\text{true}}) \ne \mathcal{E}_M
\end{align}
\ees
In this particular case, using methods from Section \ref{errSec} the difference between bounds on $\|\mathcal{E}_M(\rho_{\text{true}})\|_1$ and $\|\mathcal{E}_M\|_1$ can be found to be subleading in $\tau_B/\tau_{SB}$. Since we wish to present higher orders of the error, we do not attempt to collect errors as in Eq.~\eqref{errorAd}. Instead, our derivation will present a sequence of equations, e.g., $\rho_{\text{true}}, \rho_{B},\rho_{BM}$, and an error estimate for each step.

After using Eq.~\eqref{eq:A(t)}, i.e., in the frequency representation, Eq.~\eqref{RedStart} takes the following form: 
\begin{equation}
{\dot\rho}_{BM,I}(t) =\sum_{\omega, \omega'}\int_0^t d\tau \mathcal{C} (\tau-t)e^{-i(\omega t+\omega' \tau)}  (A_{\omega} \rho_{BM,I} A_{\omega'} - \rho_{BM,I} A_{\omega'} A_{\omega}) +\textrm{h.c.}  
\label{LindStart}
\end{equation}
We sometimes omit the explicit time dependence of $\rho$ on the r.h.s. since it is always $\rho(t)$ from hereon.

\subsection{Redfield master equation}
\label{sec:Redfield-deriv}

Let us digress briefly in order to establish the form of the Redfield equation presented in Sec.~\ref{sec:Three}.
Rotating Eq.~\eqref{RedStart} back to the Schr\"odinger picture we obtain: 
\begin{equation}
{\dot\rho}_{BM}(t) =-i[H,\rho_{BM}(t)] +\int_0^t \mathcal{C} (\tau-t) [A \rho_{BM}(t) A(\tau-t) - \rho_{BM}(t)  A(\tau-t) A] d\tau +\textrm{h.c.} 
\end{equation}
Introducing a new integration variable $t' = t- \tau$:
\begin{equation}
\label{pureRedfield}
{\dot\rho}_{BM}(t) =-i[H,\rho_{BM}(t)] +\int_0^t \mathcal{C} (-t') [A \rho_{BM}(t) A(-t') - \rho_{BM}(t)  A(-t') A] dt' +\textrm{h.c.} 
\end{equation}
We extend the upper integration limit from $t$ to $\infty$, which introduces an additional error:
\bes
    \label{eq:31}
    \begin{align}
    \label{eq:31a}
    {\dot\rho}_{BM}(t) &=-i[H,\rho_{BM}(t)] +\int_0^\infty \mathcal{C} (-t') [A \rho_{BM}(t) A(-t') - \rho_{BM}(t)  A(-t') A] dt' +\textrm{h.c.} + \mathcal{E}_l\\
    \|\mathcal{E}_l\|_1 &\leq 4c_{BM}\int_t^{\infty} |\mathcal{C} (t')|dt' \leq \frac{4c_{BM}}{t}\int_t^{T}t' |\mathcal{C} (t')|dt' + 4c_{BM}\int_T^{\infty} |\mathcal{C} (t')|dt' =O\left(\frac{\tau_B}{\tau_{SB}^2} +\frac{\epsilon_T}{\tau_{SB}} \right) ,
    \label{eq:31c}
\end{align}
\ees
where for the first summand in Eq.~\eqref{eq:31c} we assumed a lower cutoff on the evolution time $t>\tau_{SB} \delta >0$ where $\delta$ is a small constant number which we absorbed in the big-$O$ notation [i.e., $1/t = O(1/\tau_{SB})$]. 
We analyze the transients happening for $t<\tau_{SB} \delta $ in Sec.~\ref{errLimits}. The constant $c_{BM}$ is chosen as an upper bound on $\|\rho_{BM,I}\|_1$ (again, if this were a CP map, $c_{BM}=1$ would hold). We will see in Sec.~\ref{errSec} that $c_{BM} =O(1)$ in terms of our big-$O$ notation, specified in Eq.~\eqref{eq:bigO}.  For the second summand, we introduced a new bath parameter 
\beq
\epsilon_T \equiv \tau_{SB} \int_T^\infty |\mathcal{C} (t)|dt \ .
\label{eq:eps_T}
\eeq 
For most physically motivated choices of bath correlation functions $\int_0^\infty t|\mathcal{C} (t)|dt$ converges, so we may replace the total evolution time $T\to\infty$ and then $\epsilon_T=0$. But, in case $\tau_B/\tau_{SB}$ diverges as a function of $T$ [Eq.~\eqref{eq:T1tauB-b}], the above bound should be used. 

The big-$O$ notation in terms of all of these parameters $\tau_{SB}\delta \leq t \leq \tau_{SB}M, \epsilon_T, \tau_B/\tau_{SB}, c_{BM}$ should be understood as follows: there exist constants $C(M,\delta), \epsilon(M,\delta)$ such that for $\tau_B/\tau_{SB} + \epsilon_T \leq \epsilon(M,\delta)$
\begin{equation}
    \|\mathcal{E}_l\|_1 \leq C(M,\delta, \epsilon)\left(\frac{\tau_B}{\tau_{SB}^2} +\frac{\epsilon_T}{\tau_{SB}} \right)\ .
    \label{eq:bigO}
\end{equation}
As seen in the derivation above, $C(\delta,\epsilon) = \max(4c_{BM}/\delta,4c_{BM})$ in fact does not depend on $\epsilon$. 

Note that the change of integration limit $t\to \infty$ is optional: the Redfield equation can be used without it, but the change does simplify its form. Eq.~\eqref{eq:31a} without the error gives:
\begin{equation}
{\dot\rho}_{BMl}(t) =-i[H,\rho_{BMl}(t)] +\int_0^{\infty} \mathcal{C} (-t') [A \rho_{BMl}(t) A(-t') - \rho_{BMl}(t)  A(-t') A] dt' +\textrm{h.c.}
\label{eq:25}
\end{equation}
Introducing the ``filtered" operator $A_f = \int_0^{\infty} \mathcal{C} (-t')A(-t') dt'$ immediately yields Eq.~\eqref{eq:Red}, so that in fact $\rho_{BMl} = \rho_R$.
Replacing $A(-t')$  in $A_f$ by its Fourier decomposition [Eq.~\eqref{eq:A(t)}] gives Eq.~\eqref{eq:Af}, for which 
the integral can be expressed in multiple ways:
\begin{equation}
    \int_0^{\infty} \mathcal{C} (-t') e^{i\omega t'}dt'  ~ ~ \stackrel{\mathclap{\normalfont \mbox{\tiny{$\theta$=-$t'$}}}}{=} ~  ~  \frac{1}{2}\int_{-\infty}^\infty (\mathcal{C} (\theta) - \mathcal{C} (\theta) \sgn(\theta)) e^{-i\omega \theta}d\theta = \frac{1}{2}  \gamma(-\omega)  -iS(- \omega) = f^*(-\o)\ .
    \label{eq:35}
\end{equation}

\subsection{Davies-Lindblad master equation}
\label{app:derivAll}

As a second digression let us establish the form of the Davies-Lindblad equation presented in Sec.~\ref{sec:Three}.
We start from the form given in Eq.~\eqref{LindStart}, and use 
$\o t +\o' \tau = \o(t-\tau)+(\o+\o')\tau$ and a change of variables to $t'=\tau-t$ to write it as
\begin{equation}
{\dot\rho}_{BM,I}(t) =\sum_{\omega, \omega'}\int_{-t}^0 dt' \mathcal{C} (t')e^{i[\omega t'-(\o+\omega')(t+t')]}  (A_{\omega} \rho_{BM,I} A_{\omega'} - \rho_{BM,I} A_{\omega'} A_{\omega}) +\textrm{h.c.}
\label{eq:D1}
\end{equation} 
Next we apply the RWA, in which one considers the term $e^{-i(\o+\o')(t+t')}$ to be rapidly oscillating and hence self-cancelling, so the non-negligible contribution is only due to the terms with $\omega'= -\omega$. This allows us to write:
%
\begin{equation}
{\dot\rho}_{D,I}(t) =\sum_{\omega}\int_{-\infty}^0 dt' \mathcal{C} (t')e^{i\omega t'}  (A_{\omega} \rho_{D,I} A_{-\omega} - \rho_{D,I} A_{-\omega} A_{\omega}) +\textrm{h.c.} , 
\label{eq:D2}
\end{equation}
where we replaced $\rho_{BM,I}$ by $\rho_{D,I}$ and $\o'$ by $-\o$.
The error $\mathcal{E}\sim O(1/\tau_{SB})$ between the right-hand sides of Eqs.~\eqref{eq:D1} and~\eqref{eq:D2} is large but rapidly oscillating, so the solutions $\rho_{BM,I}(t)$ and $\rho_{D,I}(t) $ remain close (see Sec.~\ref{rangesDiscuss}). 

We now transform back to the Schr\"odinger picture (recall that this requires multiplying each $A_\o$ by $e^{i\omega t}$, but these cancel now due to the corresponding $A_{-\o}$):
\begin{align}
{\dot\rho}_D(t) 
 &=-i[H ,\rho_D] +\sum_{\omega}\int_{-\infty}^0  \left(\mathcal{C} (t')e^{i\omega t'}  (A_{\omega} \rho_D A_{\omega}^\dag - \rho_D A_{\omega}^\dag A_{\omega}) +  \mathcal{C} ^*(t')e^{-i\omega t'}  (A_{\omega} \rho_D A_{\omega}^\dag  -   A_{\omega}^\dag A_{\omega}\rho_D)\right)dt'  .
\label{eq:D3}
 \end{align}
It is convenient to separate the real and imaginary parts, using $\mathcal{C} ^*(t) = \mathcal{C} (-t)$:
\begin{subequations}
\label{eq:32}
\begin{align}
\int_{-\infty}^0 dt' \mathcal{C} (t')e^{i\omega t'} &= \frac12 \g(\o)  -\frac12\int_{-\infty}^\infty dt' \sgn(t')\mathcal{C} (t')e^{i\omega t'}\\
\int_{-\infty}^0 dt' \mathcal{C} ^*(t')e^{-i\omega t'} &= \frac12 \g(\o) +\frac12\int_{-\infty}^\infty dt' \sgn(t')\mathcal{C} (t')e^{i\omega t'}\ .
\end{align}
\end{subequations} 
Using the decomposition given by Eq.~\eqref{eq:32} in Eq.~\eqref{eq:D3} directly yields
the form of the Davies-Lindblad master equation presented in Eq.~\eqref{eq:Davies-Lind}. We note that versions of the Davies-Lindblad master equation that allow a time-dependent drive have been derived before (see, e.g., Refs.~\cite{PhysRevA.73.052311,ABLZ:12-SI}), but such derivations must always assume that the driving is adiabatic.


\subsection{Coarse-grained master equation}
\label{sec:CGME}

Let us now return to our main goal in this section, the derivation of the CGME. 
At this point we deviate from the standard derivations and instead follow the coarse-graining approach~\cite{PhysRevA.60.1944,Lidar200135,Majenz:2013qw}. There are two main steps: (i) Time-averaging of the state, which is done in lieu of the RWA (and reduces to the RWA in the limit of large time-averaging scale; see Appendix~\ref{app:D}), (ii) removal of a small part of the integration domain, which restores complete positivity (new in this work). We describe each in turn, along with the associated error.

\subsubsection{Time-averaging of the state}
Before we leave the interaction picture, we coarse-grain, or time-average over $T_a$ such that $\tau_B \ll T_a \ll \tau_{SB}$. Specifically, we consider another equation satisfied by a new state $\rho_{BMT,I}(t)$, where the time-averaging operation $\frac{1}{T_a}\int_{t-T_a/2}^{t+T_a/2}dt'$ is applied to the r.h.s. of Eq.~\eqref{LindStart} taken at time $t'$, except that the state is still at time $t$ on the l.h.s.:
\begin{equation}
   {\dot\rho}_{BMT,I}(t) =\sum_{\omega, \omega'}\frac{1}{T_a}\int_{t-T_a/2}^{t+T_a/2} dt' \int_0^{t'} d\tau \mathcal{C} (\tau-t')e^{-i(\omega t'+\omega' \tau)} (A_{\omega} \rho_{BMT,I}(t) A_{\omega'} - \rho_{BMT,I}(t) A_{\omega'} A_{\omega}) +\textrm{h.c.} 
   \label{eq:22} 
\end{equation}
Alas, the associated error $\mathcal{E}_T$ has $\|\mathcal{E}_T\|_1 =O(1/\tau_{SB})$, which is not small. The reason is that the fast-rotating terms may have a large derivative. What can be proven is that the solutions of Eqs.~\eqref{LindStart} and~\eqref{eq:22} remain close: 
\beq
\|\rho_{BM,I}(t) -\rho_{BMT,I}(t)\|_1=O(T_a/\tau_{SB})\ ,
\label{eq:E_T}
\eeq 
where $t\leq c\tau_{SB}$, and the r.h.s. $O(T_a/\tau_{SB})$ does not depend on $t$, only on $c$, which is considered a constant for the purposes of big-$O$ notation. 
The proof is given in Sec.~\ref{sec:t-ave-err2}, in Lemma \ref{LemTave} (this analysis first appeared in Ref.~\cite{Mozgunov:2016aa}). 
The error we track is now the error of the solution, not its derivative. 

We use the following fact about time-local (Markovian) and time-nonlocal (non-Markovian) differential equations:
\begin{mylemma}
\label{factDif}
Assume that
\begin{equation}
    \dot{x}(t) =\mathcal{L}(x(t)) +\mathcal{E}, \quad \dot{y}(t) = \mathcal{L}(y(t)), \quad x(0) =y(0) \ ,
\end{equation}
where $\mathcal{L}$ is a linear superoperator and 
$\Lambda$ is a positive constant such that $\sup_{\tau,x:\|x\|_1 = 1}\|\mathcal{L}_\tau(x)\|_1\leq \Lambda$. 
Or, assume that
\begin{equation}
    \dot{x}(t) =\int_0^tK_{t-\tau}(x(\tau))d\tau +\mathcal{E}, \quad \dot{y}(t) = \int_0^tK_{t-\tau}(y(\tau))d\tau,\quad x(0) =y(0) \ ,
\end{equation}
where $K_{t-\tau}(x)$ is a linear superoperator and $\Lambda$ is a positive constant such that $\sup_{t,x(\tau):\|x(\tau)\|_1 = 1}\int_0^t \|K_{t-\tau}(x(\tau))\|_1 d\tau \leq \Lambda$. 
Then:
\beq
\forall t\leq \frac{c}{ \Lambda}: ~ \|x(t) -y(t) \|_1 \leq (e^c-1)\frac{\|\mathcal{E}\|_1}{\Lambda}\ .
\eeq
\end{mylemma}
For the proof see Sec.~\ref{sec:t-ave-err}. 
We can take $\Lambda = 4/\tau_{SB}$ because of Eq.~(\ref{eq:|L|}, \ref{newLambda}). Using this, we rewrite the result of the Lemma as follows, for brevity: if $t \leq c/\Lambda = c\tau_{SB}/4$, then $\|x-y\|_1 = O(\tau_{SB} \|\mathcal{E}\|_1)$. Thus, using the previously noted results that $\|\mathcal{E}_{B}\|_1,\|\mathcal{E}_{M}\|_1=O\left({\tau_B}/{\tau_{SB}^2}\right)$:
\bes
\begin{align}
    \|\rho_{\text{true},I}(t) -\rho_{BMT,I}(t)\|_1&\leq \|\rho_{\text{true},I}(t) -\rho_{B,I}(t)\|_1 +\|\rho_{B,I}(t) -\rho_{BM,I}(t)\|_1 +\|\rho_{BM,I}(t) -\rho_{BMT,I}(t)\|_1 \\
    &=  O(\tau_{SB}\|\mathcal{E}_{B}\|_1)+O(\tau_{SB}\|\mathcal{E}_{M}\|_1) +O(T_a/\tau_{SB})  = O\left(\frac{\tau_B +T_a}{\tau_{SB}}\right) .  
    \label{solErr}
\end{align}
\ees
\subsubsection{Neglecting part of the integration domain to regain complete positivity}
Upon changing variables to $s=t'-t$, and renaming $s$ as $t'$, Eq.~\eqref{eq:22} becomes:
\begin{equation}
{\dot\rho}_{BMT,I}(t) =\sum_{\omega, \omega'}\frac{1}{T_a}\int_{-T_a/2}^{T_a/2} dt' \int_0^{t+t' } d\tau \mathcal{C} (\tau-t - t')e^{-i\omega (t+t')-i\omega' \tau}  (A_{\omega} \rho_{BMT,I} A_{\omega'} - \rho_{BMT,I} A_{\omega'} A_{\omega}) +\textrm{h.c.}
\end{equation}
We rotate out of the system interaction picture into the system Schr\"odinger picture by defining the coarse-grained state $\rho_{BMT} =U(t)\rho_{BMT,I}  U^\dag(t) = e^{-iHt}\rho_{BMT,I} e^{iHt}$, leaving only the bath interaction picture, by multiplying each summand by a $e^{i(\omega+\omega')t}$ factor:\footnote{ 
Note that the transformation back from the interaction picture is $U\dot{\rho}_IU^\dag$. Thus, e.g., $A_{\omega} \rho_I A_{\omega'} \mapsto U A_{\omega} \rho_I A_{\omega'}U^\dagger = U A_{\omega} U^\dagger \rho U A_{\omega'} U^\dagger$, where $U = \sum_k e^{-iE_kt}\Pi_k$. Using Eqs.~\eqref{eq:A_om} and \eqref{eq:A(t)}, we have $U A_{\omega} U^\dagger = \sum_{kl} \sum_{mn: E_{mn} = \omega}e^{-i (E_k-E_{l}) t}\Pi_k \Pi_n A  \Pi_m \Pi_{l} = \sum_{mn: E_{mn}=\o} e^{-i(E_n-E_m) t} \Pi_n A\Pi_m = e^{i\omega t} A_\omega$, with the last equality holding due to the fact that $E_n-E_{m}$ is fixed at $\omega$ for all combinations of $m$ and $n$. The claim now follows.
}
\begin{equation}
{\dot\rho}_{BMT}(t) =-i[H,\rho_{BMT}] + \sum_{\omega, \omega'}\frac{1}{T_a}\int_{-T_a/2}^{T_a/2}dt' \int_0^{t+t'} d\tau \mathcal{C} (\tau - t - t')e^{-i\omega t'-i\omega' (\tau-t)}   (A_{\omega} \rho_{BMT} A_{\omega'} - \rho_{BMT} A_{\omega'} A_{\omega}) +\textrm{h.c.}
\label{eq:27S}
\end{equation}
Next, we change variables to $\tau' = \tau-t$, so that
\begin{equation}
{\dot\rho}_{BMT}(t) =-i[H,\rho_{BMT}] + \sum_{\omega, \omega'}\frac{1}{T_a}\int_{-T_a/2}^{T_a/2} dt' \int_{-t}^{t'} d\tau'\mathcal{C} (\tau' -t')e^{-i(\omega t'+\omega' \tau')}   (A_{\omega} \rho_{BMT} A_{\omega'} - \rho_{BMT} A_{\omega'} A_{\omega}) +\textrm{h.c.}
\label{eq:11}
\end{equation}
%
\begin{figure}[t]
\includegraphics[width=0.8\linewidth]{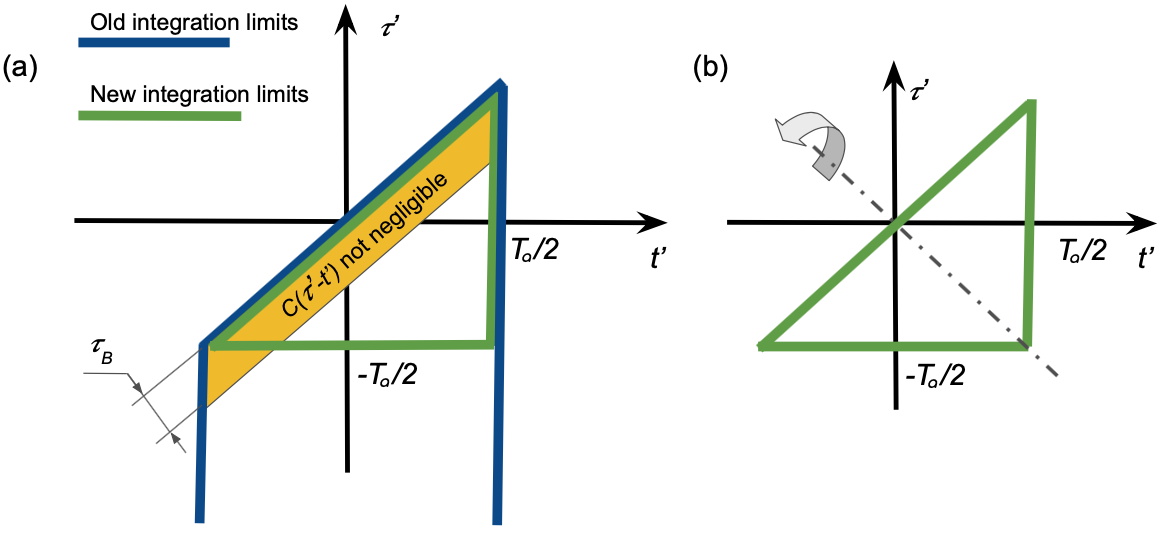}
\caption{(a) The box outside the green line is neglected in the integration, which restores complete positivity. The area of the non-negligible part (yellow) is $\sim \tau_B T_a$ while the area we nevertheless neglect is $\sim\tau_B^2$. (b) Illustration of the interchange of the order of integration limits used to prove Eq.~\eqref{eq:29b}.}
\label{cut1}
\end{figure}
%
We define $\mathcal{L}^{BMT}_t(\rho_{BMT})$ to be the superoperator on the r.h.s. The key trick now is to replace $-t$ by $-T_a/2$, which lets us write (note that Lemma \ref{factDif} will work for $\mathcal{E}_p$ either in the initial or in the resulting equation, like we do here):
\begin{subequations}
\begin{align}
\label{eq:26}
{\dot\rho}_C(t) &=\mathcal{L}^{BMT}_t(\rho_C) - \mathcal{E}_p=-i[H,\rho_C] + \sum_{\omega, \omega'}
x_{\omega\omega'}(A_{\omega} \rho_C A_{\omega'} - \rho_C A_{\omega'} A_{\omega}) +\textrm{h.c.}  \\
\label{eq:xww'}
x_{\omega\omega'} &\equiv \frac{1}{T_a}\int_{-T_a/2}^{T_a/2} dt' \int_{-T_a/2}^{t'} d\tau'\ \mathcal{C} (\tau'-t')e^{-i(\omega t'+\omega' \tau')} \ ,
\end{align}
\end{subequations} 
This replacement cuts a corner of area $\sim \tau_B^2$ out of the integration domain with non-vanishing $\mathcal{C} (t)$, itself of area $\sim T_a \tau_B$, as illustrated in Fig.~\ref{cut1}(a) for $t>T_a/2$ (times $t<T_a/2$ introduce a transient, as discussed in more detail in Sec.~\ref{sec:t-ave-err2}).
Thus, the corner being cut for $t>T_a/2$ is just a $\tau_B/T_a$ fraction of the whole integral in Eq.~\eqref{eq:xww'} (ignoring factors of $\sqrt{2}$), which [using Eq.~\eqref{eq:T1tauB-a}] is itself upper bounded by $1/\tau_{SB}$ in absolute value. Thus, the error introduced by the corner removal is 
\beq
\|\mathcal{E}_p\|_1 = O(\tau_B/(T_a\tau_{SB}))\ \  \mathrm{ for }\ \ t>T_a/2
\label{eq:Ep}
\eeq 
[for $t<T_a/2$ it can be $O(1/\tau_{SB})$ for a short transient]. The constant in big-$O$ notation here is independent of $\tau_{SB},\tau_B, T_a$. As we show below, $\mathcal{E}_p$ is the cost of recovering complete positivity, in the sense that the master equation after subtraction of $\mathcal{E}_p$ is in canonical Lindblad form.

Eq.~\eqref{eq:26} is the coarse-grained master equation from Ref.~\cite{Majenz:2013qw}. 
To see this we need some properties of the coefficient $x_{\omega\omega'}$.
In particular, we need the properties
\bes
\begin{align}
\label{eq:29a}
x_{\omega\omega'}^* &=\frac{1}{T_a}\int_{-T_a/2}^{T_a/2} dt' \int_{t'}^{T_a/2} d\tau'\ \mathcal{C} (\tau'-t')e^{-i(\omega t'+\omega' \tau')}  \\
x_{\omega\omega'} &=x_{-\omega'-\omega}  \ ,
\label{eq:29b}
\end{align}
\ees
with the first following from $\mathcal{C} ^*(t) = \mathcal{C} (-t)$, and the second from interchanging integration limits after a reflection w.r.t. the $\tau' = -t'$ axis, as illustrated in Fig.~\ref{cut1}(b).
This allows us to reshuffle the $\textrm{h.c.}$ part of Eq.~\eqref{eq:26}, taking the $-\omega', -\omega$ term in the sum for the $\omega, \omega'$ term in the original part. The resulting equation is:
\begin{equation}
{\dot\rho}_C(t) =-i[H,\rho_C] + \sum_{\omega, \omega'}x_{\omega\omega'} (A_{\omega} \rho_C A_{\omega'} - \rho_C A_{\omega'} A_{\omega}) +x_{\omega\omega'}^* (A_{\omega} \rho_C A_{\omega'} - A_{\omega'} A_{\omega}\rho_C )\  .
\end{equation}
The $\Im x_{\o\o'}$ part cancels for $A_{\omega} \rho A_{\omega'}$ and sends $A_{\omega'} A_{\omega}$ into the Lamb shift:
\beq
H_{\mathrm{LS}} \equiv -\sum_{\o\o'}\Im(x_{\o\o'})A_{\o'}A_{\o} =\frac{i}{2T_a}\int_{-T_a/2}^{T_a/2}dt_1\int_{-T_a/2}^{T_a/2}  dt_2\ \sgn(t_1-t_2)\mathcal{C} (t_2-t_1)  A(t_2) A(t_1)\ ,
\label{eq:HLS1}
\eeq
where to get the second form we used $\Im(x_{\o\o'}) = \frac{1}{2i}(x_{\o\o'}-x^*_{\o\o'})$, combined Eqs.~\eqref{eq:xww'}, \eqref{eq:29a} into a single integral using the $\sgn(t_1-t_2)$ function, and used Eq.~\eqref{eq:A(t)}. The Lamb shift can be further simplified as shown in Appendix~\ref{app:LS-simp}. The simplified form is the one we presented as a part of our main result [Eq.~\eqref{eq:HLS-CG}] in Sec.~\ref{sec:Three}.

We define $\gamma_{\omega\omega'} \equiv x_{\omega\omega'}^* +x_{\omega\omega'} =2\Re x_{\omega\omega'}$ and obtain:
\begin{equation}
\gamma_{\omega\omega'} = \frac{1}{T_a}\int_{-T_a/2}^{T_a/2}dt' \int_{-T_a/2}^{T_a/2} d\tau' \mathcal{C} (\tau' -t')e^{-i(\omega t'+\omega' \tau')}  \ .
\label{eq:18}
\end{equation}
Therefore 
\begin{equation}
{\dot\rho_C}(t) =-i[H +H_{\mathrm{LS}},\rho_C] + \sum_{\omega, \omega'}\gamma_{\omega\omega'} (A_{\omega} \rho_C A_{\omega'} - \frac{1}{2}\left\{ \rho_C , A_{\omega'} A_{\omega}\right\} ) \ ,
\label{eq:CGME}
\end{equation}
which is Eq.~(50) from~\cite{Majenz:2013qw} up to a shift in the center of averaging and changing variables $A_{\omega'} \to A_{\omega'}^\dag, ~ \omega' \to -\omega'$. 

\subsubsection{Complete positivity}

Complete positivity is not immediately apparent from Eq.~\eqref{eq:CGME}. Let us demonstrate it next.
Substituting $\mathcal{C} (t) = \frac{1}{2\pi}\int_{-\infty}^{\infty} e^{-i\epsilon t} \gamma(\epsilon) d\epsilon$  [Eq.~\eqref{eq:C(t)}], where $\gamma(\epsilon)\geq 0$ (see Appendix~\ref{app:A}), into Eq.~\eqref{eq:18}, we have:
\begin{equation}
\gamma_{\omega\omega'} =\int_{-\infty}^{\infty}d\epsilon \frac{\gamma(\epsilon)}{2\pi T_a}\int_{-T_a/2}^{T_a/2}dt' \int_{-T_a/2}^{T_a/2} d\tau' e^{-i\epsilon(\tau'-t')}e^{-i\omega t'-i\omega' \tau'}  = \int_{-\infty}^{\infty}d\epsilon f(\epsilon,\omega) f^*(\epsilon,-\omega')\ ,
\label{eq:34}
\end{equation}
where 
\begin{equation}
    f(\epsilon,\omega) \equiv 
    \sqrt{\frac{\gamma(\epsilon)}{2\pi T_a}}\int_{-T_a/2}^{T_a/2}dt e^{i(\epsilon-\omega)t} 
    \label{eq:f}
\end{equation}
is a filter function whose integrated form is given in Eq.~\eqref{eq:newLindops}.
By using the above decomposition for $\gamma_{\omega\omega'}$ in Eq.~\eqref{eq:CGME}, we obtain the manifestly CP Lindblad form given in Eq.~\eqref{eq:Lind}.

Note further that it was shown in Ref.~\cite{Majenz:2013qw} that the Davies-Lindblad equation [Eq.~\eqref{eq:Davies-Lind}] arises as the $T_a \to \infty$ limit of the CGME. Similar ideas have appeared in \cite{Benatti_2009,Benatti_2010}. We provide this proof for completeness in Appendix~\ref{app:D}. This derivation has the advantage that it arrives at the Davies-Lindblad equation without invoking the (uncontrolled) RWA, and instead shows that the Davies-Lindblad equation is the limit of infinite coarse graining time of the CGME. One of our new results is to show how the error of the Davies-Lindblad equation can be controlled for a sufficiently large (but not infinite) coarse graining time; see Sec.~\ref{DavSec}.


\subsection{Ranges of applicability of the three master equations}
\label{rangesDiscuss}

We now discuss the ranges of applicability of the master equations derived above in more rigor than in Sec.~\ref{sec:err-bounds}, while postponing the derivations to Sec.~\ref{errSec}. Recall that the ranges are dependent on the timescales $\tau_{SB}$ and $\tau_B$ defined in Eq.~\eqref{eq:T1tauB}.

Errors are incurred from different approximations made along the way. The very first approximation made in the derivation of all three master equations is the Born approximation (see Sec.~\ref{sec:derivation}). Unlike the other approximations, we do not prove a rigorous error bound for the Born approximation error; we just provide a bound on the first order contribution to this error. However, under the assumption that the error converges in the first place, the Born error will turn out to be subleading. Therefore we settle for an estimate of this error for the sake of simplicity, but we note that a rigorous bound can be obtained~\cite{Davies:74}.

We first collect all the errors introduced in the derivation of the CGME: the difference in solutions $\|\rho_{\text{true},I}(t) -\rho_{BMT,I}(t)\|_1$ [Eq.~\eqref{solErr}, which includes the error due to time-averaging], and the error due to enforcing complete positivity, $ \|\mathcal{E}_p\|_1 = O(\tau_B/(T_a\tau_{SB}))$ [Eq.~\eqref{eq:Ep}]. The difference between the solution of the CGME $\rho_C(t)$ and the true time evolved state $\rho_{\text{true}}(t)$ thus satisfies (note that the norm is unitarily invariant and hence unaffected by transformation to or out of the interaction picture):
\bes
\begin{align}
    \|\rho_{\text{true}}(t) -\rho_{C}(t) \|_1 \leq \|\rho_{\text{true},I}(t) -\rho_{BMT,I}(t)\|_1 +  \|\rho_{BMT}(t) -\rho_C(t)\|_1 &=    O\left(\frac{T_a +\tau_B}{\tau_{SB}}\right) + O(\tau_{SB}\|\mathcal{E}_p\|_1) \\
    &= O\left(\frac{\tau_B}{T_a} +\frac{T_a +\tau_B}{\tau_{SB}}\right)  
\end{align}
\ees
Here we used Lemma~\ref{factDif} to multiply $\|\mathcal{E}_p\|_1$ by $\tau_{SB}$. 
Optimizing $T_a$ to minimize this error, we obtain $T_a = \sqrt{\tau_{SB} \tau_B}$, and: 
\begin{equation}
    \|\rho_{\text{true}}(t) -\rho_{C}(t) \|_1 =O\left(2\sqrt{\frac{\tau_B}{\tau_{SB}}} +\frac{\tau_B}{\tau_{SB}}\right) =   O\left(\sqrt{\frac{\tau_B}{\tau_{SB}}} \right) \quad \text{for} \  t>T_a/2 \ .
    \label{eq:57new}
\end{equation}
Note that if we define the coupling $g$ strength via $1/\tau_{SB} = g^2 \tau_B$ (recall the discussion in Sec.~\ref{sec:times}), then the error bound of Eq.~\eqref{eq:57new} is $O(g\tau_B)$.

We can in fact improve on the result presented in Eq.~\eqref{eq:57new}. A careful derivation presented in Sec.~\ref{errSec} generalizes the above result from the $t>T_a/2$ case discussed here [recall Eq.~\eqref{eq:Ep}] to $t>0$. The big-$O$ notation in Eq.~\eqref{eq:57new} means that there exist numbers $(C,\delta, c)$ such that for all $0<t < c\tau_{SB}$  and $\tau_B/\tau_{SB} < \delta$  the error is  $\|\rho_{\text{true}}(t) -\rho_{C}(t) \|_1 \leq C(c,\delta) \sqrt{\tau_B/\tau_{SB}}$. We do not know the function $C(c,\delta)$ explicitly, because we do not keep track of it in our analysis of the Born error [the details of which are given in Sec.~\ref{Born}]. The contributions from all the other errors are known explicitly. Our derivation also allows us to state a stronger result exhibiting the coefficients and the $t$-dependence of the first few leading terms of the error as a series in $\tau_B/\tau_{SB}$, since the unknown Born contributions start with $O(\tau_B^2/\tau_{SB}^2)$. The improved bound is:

\begin{equation}
      \|{\rho}_{\text{true}}(t) - \rho_{C}(t)  \|_1  \leq \frac{13e^{\frac{4t}{\tau_{SB}}}\sqrt{\tau_B}\left(1 +  \frac{29\tau_B e^{\frac{8t}{\tau_{SB}}}}{\tau_{SB}}\right)}{\sqrt{\tau_{SB}}}+ \frac{\left(e^{\frac{4t}{\tau_{SB}}} -1\right)e^{\frac{8t}{\tau_{SB}}} \tau_B\left(12+  O\left(e^{\frac{4t}{\tau_{SB}}}\frac{\tau_B}{\tau_{SB}}\right)\right)}{\tau_{SB}}   \ . \label{eq:58}
\end{equation}
The time $t$ can now be arbitrarily large, as long as the combination $e^{\frac{4t}{\tau_{SB}}}\frac{\tau_B}{\tau_{SB}}$ is small. Specifically, the big-$O$ notation  in Eq.~\eqref{eq:58} means that there exist constants $C,\epsilon$ such that for all $e^{\frac{4t}{\tau_{SB}}}\frac{\tau_B}{\tau_{SB}}\leq \epsilon$, $O\left(e^{\frac{4t}{\tau_{SB}}}\frac{\tau_B}{\tau_{SB}}\right)  \leq C e^{\frac{4t}{\tau_{SB}}}\frac{\tau_B}{\tau_{SB}}$.

The time-averaging window leading to Eq.~\eqref{eq:57new} was chosen sub-optimally, as the more careful analysis yielding Eq.~\eqref{eq:58} uses $T_a = \sqrt{\tau_B \tau_{SB}/5}$. Including higher order corrections in $\tau_B/\tau_{SB}$ leads to a minor improvement in the subleading terms in the bound, but does not affect the leading term (see Section \ref{meatOptim}).

Let us now discuss the time-dependence of the error. From the above expression it follows that for any $t$,
    $\|\rho_{\text{true}}(t) -\rho_{C}(t) \|_1 = O(\sqrt{\tau_B/\tau_{SB}}e^{6t/\tau_{SB}})$, which is Eq.~\eqref{eq:59b}.
For the sake of completeness, we present the errors of the Redfield equation (for $\epsilon_T =0$ [recall Eq.~\eqref{eq:eps_T}]; for the general case, see Sec.~\ref{errLimits}):
   $\|\rho_{\text{true}}(t) -\rho_R(t) \|_1 =   O\left(\frac{\tau_B}{\tau_{SB}}  \text{ln}\frac{\tau_{SB}}{\tau_B}e^{12t/\tau_{SB}} \right)$, which is Eq.~\eqref{err:Redfield}, 
    and the Davies-Lindblad equation:
   $\|\rho_{\text{true}}(t) -\rho_D(t) \|_1 =  O\left(\left(\frac{\tau_B}{\tau_{SB}}  +\sqrt{\frac{1}{\tau_{SB}\delta E}}\right)e^{12t/\tau_{SB}}\right)$, which is Eq.~\eqref{eq:72}.

We note that for Eq.~\eqref{pureRedfield} there is no transient before the introduction of $\mathcal{E}_l$:
\begin{equation}
   \|\rho_{\text{true}}(t) -\rho
   _{BM}(t) \|_1 = (e^{4t/\tau_{SB}}-1)O\left(\frac{\tau_B}{\tau_{SB}}e^{8t/\tau_{SB}}\right)\ ,
\end{equation}
thus the error grows from zero linearly at small times, and the slope of the upper bound is $O(\tau_B/\tau_{SB}^2)$. In contrast, for the Redfield, CGME and Davies-Lindblad equations, there is a nonzero transient error bound at $t\to 0$. The error itself starts at zero (our bounds are not tight), but the slope is $O(1/\tau_{SB})$ instead.

It is important to emphasize that the solutions of open system master equations do not actually grow exponentially. In fact, for time-independent system Hamiltonians all of the equations presented have steady states to which they converge in the operator and trace norm at some rate. If that rate is sufficiently fast, one can prove a much stronger result \cite{Cubitt:2015} than our error bounds: a stability of the open system dynamics to small perturbations. The stability can be expressed in terms of a time-independent error bound on the difference in solutions.

We remark that the coefficients we presented and the time-dependence of the error are explicit, while other work hides them in big-$O$ notation. Our results allow one to put error bars on the numerical solution of the above equations. 

\subsection{Discretization, spatial locality, and the Lieb-Robinson bound}
\label{posLoc}

Different forms of the CGME are preferred depending on the application. The form given in Eq.~\eqref{eq:CGME} with a discrete sum $\sum_{\omega,\omega'}$ is ready for numerical implementation but at exponential cost, since for a Hilbert space of dimension $2^n$ the number of terms in the sum is $2^{4n}$. Numerical solution of the equation already requires at least $2^{2n}$ numbers just to store the density matrix. It is therefore desirable to reduce the number of terms on the r.h.s. to a constant in $n$ (or polynomial in $n$ for the general case of multiple interaction terms; see Sec.~\ref{scalSec}). We first present a discretization that achieves this goal. It is also desirable to have each term locally supported, instead of acting on the whole system. This speeds up the numerics, and also allows one  to meaningfully use the equation with methods that use resources that are polynomial in $n$ to store the density matrix, e.g., matrix product states (MPS) \cite{Vidalsolo:03,VerstraeteGarcia:04}. For a recent implementation of these methods see \cite{Znidaric:10,Medvedyeva:16}. We present a way to write down a local equation at the end of this subsection. 

In Eq.~\eqref{eq:Lind}, the number of terms is infinite because of the integration over $\epsilon$. In Appendix~\ref{app:C} we explain how a discrete approximation can be derived. The result is:
\begin{equation}
\label{eq:discrete-approx}
   {\dot\rho}_C(t) =-i[H +H_{\mathrm{LS}},\rho_C] + \Delta \epsilon 
  \sum_{-k^*<k<k^*}\sum_{\epsilon = k \Delta \epsilon} (A_{\epsilon} \rho_C A_{\epsilon}^\dag - \frac{1}{2}\left\{ \rho_C , A_{\epsilon}^\dag A_{\epsilon}\right\} ) + \mathcal{E}_{k}
\end{equation}
where $A_\epsilon$ is given by Eq.~\eqref{eq:newLindops} and the values of $\Delta \epsilon$, $k^*$ and $\|\mathcal{E}_{k}\|_1$ are given in Appendix \ref{app:C}, and are system-size-independent as long as $[H,A]=O(1)$. We do not present a similar discrete approximation for the Lamb shift, leaving it to future work.

We now proceed to show that storing the Lindblad operators $A_\epsilon$ does not require resources exponential in the system size $n$. Recall that $A_\epsilon = 
    \sqrt{\frac{\gamma(\epsilon)}{2\pi T_a}}\int_{-T_a/2}^{T_a/2}e^{i\epsilon t}A(t)dt$.
This form is interesting because $A(t)$ is evaluated over the interval $[-T_a/2, T_a/2]$. An operator $A(t)$ at those times is local (with exponentially decaying tails) with a locality radius $v T_a/2$, where $v$ is the Lieb-Robinson velocity ($v \sim \|H_{\textrm{local}}\|$, the coupling strength in the Hamiltonian) \cite{Lieb:72}. More rigorously, we consider a system defined on a graph. Each vertex of the graph supports a finite-dimensional Hilbert space, and the system Hilbert space is a tensor product of vertex Hilbert spaces. The edges of the graph are present if there are nontrivial interactions between the corresponding vertices in the Hamiltonian. That is, for vertices $i,j$ we can define a trace over the degrees of freedom at other vertices, denoted $\Tr_{\overline{i,j}}H-\Tr_{\overline{i}}H -\Tr_{\overline{j}}H$. If this operator is nonzero (for traceless $H$), then there is a nontrivial interaction and the edge $(i,j)$ is present in the graph. The Lieb-Robinson bound can be used (see, e.g., Lemma 5 in Ref.~\cite{Haah:18}) to derive the following decomposition of the operator $A_\epsilon$ into its local and small nonlocal part:
\begin{equation}
    A_\epsilon  = A_\epsilon^{(vT_a/2 +\delta r)} + \delta A, \quad \|\delta A\| \leq \sqrt{\frac{\gamma(\epsilon) T_a }{2\pi}} \|A\| e^{-c \delta r}\ ,
    \label{eq:64}
\end{equation}
where $c$ is a universal (system-independent and $\epsilon$-independent) constant, and
\beq
A_\epsilon^{(x)} =\sqrt{\frac{\gamma(\epsilon)}{2\pi T_a}}\int_{-T_a/2}^{T_a/2}e^{i\epsilon t}e^{iH^{(x)}t}A e^{-iH^{(x)}t}dt  
\eeq 
is an operator strictly local within the graph distance $x$ from the original location of $A$. We define the graph distance as the length (number of edges) of the smallest connected path of edges between two vertices.
Here $H^{(x)} = \Tr_{\overline{x}}H$ is composed of those Hamiltonian terms which involve vertices within graph distance $x$ from the original location of $A$. If we drop $\delta A$ from Eq.~\eqref{eq:64}, we can control the resulting error similarly to what is described in Sec.~\ref{errSec} below:
\begin{equation}
   {\dot\rho}^{(x)}_C(t) =-i[H +H^{(x)}_{\mathrm{LS}},\rho^{(x)}_C] + \Delta \epsilon 
  \sum_{-k^*<k<k^*}\sum_{\epsilon = k \Delta \epsilon} \left(A^{(x)}_{\epsilon} \rho^{(x)}_C A_{\epsilon}^{(x)\dag} - \frac{1}{2}\left\{ \rho^{(x)}_C , A_{\epsilon}^{(x)\dag} A^{(x)}_{\epsilon}\right\} \right)\ . \label{localCGME}
\end{equation}
This master equation has local generators $A^{(x)}_{\epsilon}$ with $x =vT_a/2 +\delta r$. We also assumed that some truncation method is applied to the Lamb shift $H_{\mathrm{LS}} \to H^{(x)}_{\mathrm{LS}}$, although we do not present it in this work.
One can now use standard methods \cite{Pichler:10} to recast Eq.~\eqref{localCGME})  as a stochastic Schr\"odinger equation, with those operators corresponding to jumps:
    \begin{equation}
        \frac{d}{dt}|\psi(t,r)\rangle =  G(t,r) |\psi(t)\rangle, \quad \rho^{(x)}_C(t) =\text{Av}_r|\psi(t,r)\rangle  \langle \psi(t,r)|
    \end{equation}
here $r$ represents a random number used to generate trajectories. We will not present the explicit form of $G$, but note that:
\begin{equation}
    G = \sum_{i=1}^n G_i\ ,
\end{equation}
where the sum is over qubits and each $G_i$ is a local operator of radius $x =vT_a/2 +\delta r$ and can be computed in $O(1)$ time in the system size $n$. The time evolution is then given by a trotterization $U^{\text{Trotter}}(r)$ of
 \begin{equation}
        U(r) = \mathcal{T}e^{-i\int_0^tG(\tau,r)d\tau } \ .
\end{equation}
$U^{\text{Trotter}}(r)$ is a \emph{geometrically local} circuit of depth $O(t)$ where each gate takes $O(1)$ time to compute. Contraction of such a circuit $|\psi(t,r)\rangle = U^{\text{Trotter}}(r)|\psi(0)\rangle$ does not necessarily have a local structure, even if the initial condition $|\psi(0)\rangle$ is a product state. However, note that in 1d $|\psi(t,r)\rangle$ has the structure of a MPS~\cite{Vidalsolo:03}:
\begin{equation}
    |\psi(t,r)\rangle = \sum_{s_1\dots s_n}\prod_i M_i^{s_i}|s_1 \dots s_n\rangle
\end{equation}
Where $M_i^{s_i}$ are $\chi \times \chi$ matrices for each $i,s_i$, and $\chi$ is called a bond dimension. A portion of the time evolution of a MPS is illustrated in Fig.~\ref{tevMP} for 2-local $G_i$, for which $U^{\text{Trotter}}=\prod_t(\prod_{i=2k}U_{i,i+1}\prod_{i=2k-1}U_{i,i+1})$ and $U_{i,i+1} = e^{-iG_i dt}$. Iterating the time step naively grows the bond dimension as $\chi =e^{O(t)}$.
    \begin{figure}
\centering
\includegraphics[width=1\columnwidth]{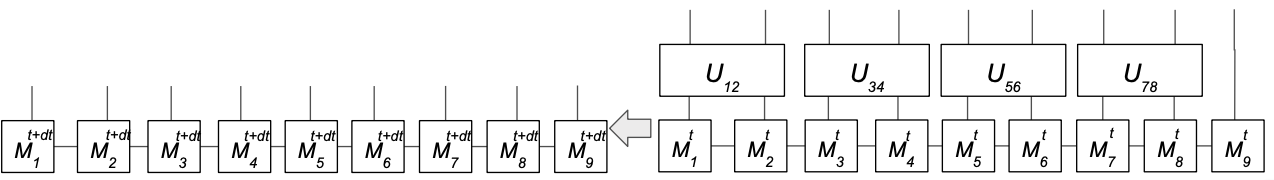}
\caption{Definition of a MPS $|\psi(t,r)\rangle = U^{\text{Trotter}}(r)|\psi(0)\rangle$. 
}
\label{tevMP}
\end{figure}
The computation of any observable in $|\psi(t,r)\rangle$ requires a contraction of such a tensor network, which takes $O(n)e^{O(t)}$ computational time, and $e^{O(t)}$ memory.   In practice, MPS calculations employ a truncation scheme for the bond dimension after every time step, which fixes the bond dimension to a constant and leads to $O(n,t)$ computational time, $O(n)$ memory. Such a method is generally not a controlled approximation, but it may be possible to develop a special version with an error estimate. Averaging over randomness adds an extra computational cost to either method. We leave the development of further details to a future publication. 

The simultaneous locality and positivity of Eq.~\eqref{localCGME} is a new result: previously, the Davies generators were only local for commuting Hamiltonians, and the Redfield master equation leads to completely positive evolution only for flat spectral densities. Thus, this is the first method for \emph{ab initio} completely positive open system evolution of MPSs. 

\subsection{Scaling with size}
\label{nTerms}


Here we discuss  the system-size dependence. So far we assumed that there is only one term $V = A\otimes B$ describing the interaction with the bath. In reality, there are $n$ terms like this, one for each qubit. If qubits are coupled to the same bath $V = \sum_i^n A_i \otimes B$, it will lead to correlated noise (collective decoherence~\cite{Duan:98}). If the qubits are coupled to different baths $V = \sum_i^n A_i \otimes B_i$, it can lead to uncorrelated noise if $\Tr[\rho_B B_i(t)B_j]=0$. We will consider the latter case since it contains the former case if one sets $B_i =B$ (though distinctly different phenomena such as decoherence-free subspaces appear in this special case~\cite{Zanardi:97c,Lidar:1998fk}). The form of the Davies-Lindblad and Redfield equations can be found elsewhere; here we present only the form of the CGME. The sum is carried through to the frequency form of the equation --- Eq.~\eqref{eq:CGME}:
\bes
\begin{align}
{\dot\rho_C}(t) &=-i[H +H_{\mathrm{LS}},\rho_C] + \sum_{\omega, \omega',i,j}\gamma^{ij}_{\omega\omega'} (A_{\omega}^j \rho_C A_{\omega'}^i - \frac{1}{2} \{\rho_C , A_{\omega'}^i A_{\omega}^j\}  ) \ ,\\
\gamma^{ij}_{\omega\omega'} &=  \frac{1}{T_a}\int_{-T_a/2}^{T_a/2}dt' \int_{-T_a/2}^{T_a/2} d\tau' \mathcal{C}_{ij} (\tau' -t')e^{-i(\omega t'+\omega' \tau')} \ , \\
H_{\mathrm{LS}} &=\frac{i}{2T_a}\sum_{ij}\int_{-T_a/2}^{T_a/2}dt_1\int_{-T_a/2}^{T_a/2}  dt_2\ \sgn(t_1-t_2)\mathcal{C}_{ij} (t_2-t_1)  A^i(t_2) A^j(t_1) \ , \\
\mathcal{C}_{ij} (t) &= \Tr[\rho_B B^i(t)B^j] \ .
\end{align}
\ees
The theoretically  optimal $T_a=\sqrt{\tau_{SB}\tau_B /5}$ is given by the same expression in terms of time scales $\tau_B$ and $\tau_{SB}$, which are now defined using $\int|\mathcal{C}| \mapsto \text{max}_{ij}\int |\mathcal{C}_{ij}|$. Such a definition keeps $\tau_{B}$ and $\tau_{SB}$ consistent with the single qubit quantities. The locality of the equation is now also dependent on the decay of $\mathcal{C}_{ij}$ with the distance between $i$ and $j$. If it can be approximated by zero outside some radius (``strong locality''), the equation remains locally generated in the sense that every term on the r.h.s. contains local terms around some qubit $i$ conjugating the density matrix, and there are polynomially many of them. If we use a weaker sense of locality, when the operators on the right and left of the density matrix are allowed to be local around different $i,j$, and neglect the Lamb shift, then no assumption on $\mathcal{C}_{ij}$ is needed. 

Let us present an explicitly Lindblad form. Let $\mu$ index the positive (see Appendix \ref{app:A}) eigenvalues of $\gamma_{ij}(\omega)= \int_{-\infty}^\infty e^{i\omega s}\mathcal{C}_{ij}(s)ds$:
\beq
\gamma_{ij}(\omega) = \sum_\mu U_{\mu i}^*(\omega) D_\mu (\omega)U_{\mu j}(\omega) \ .
 \eeq
This allows one to write:
\bes
\begin{align}
  \gamma^{ij}_{\omega\omega'} &= \int d\epsilon \sum_{\mu}   U_{\mu i}^*(\epsilon) {\tilde{\gamma}}_{\omega \omega'}(D_\mu(\epsilon), \epsilon) U_{\mu j}(\epsilon) \\
   {\tilde{\gamma}}_{\omega \omega'}(D_\mu(\epsilon), \epsilon) &=\frac{1}{2\pi T_a}\int_{-T_a/2}^{T_a/2}dt' \int_{-T_a/2}^{T_a/2} d\tau' D_{\mu}(\epsilon) e^{-i (\epsilon(\tau' -t') +\omega t'+\omega' \tau')} \ .
\end{align}
\ees
Next, note that if we let $\tilde{A} = \sum_j U_{\mu j }(\epsilon) A_j$, the previous derivation goes through.
\begin{equation}
{\dot\rho}_C(t) =-i[H +H_{\mathrm{LS}},\rho_C(t)] + \int_{-\infty}^{\infty}d\epsilon \sum_{\mu=1}^n \left(A_{\epsilon,\mu} \rho_C(t) A_{\epsilon,\mu}^\dag - \frac{1}{2}\left\{ \rho_C(t) , A_{\epsilon,\mu}^\dag A_{\epsilon,\mu}\right\} \right)\ ,
\end{equation}
where the Lindblad operators are
\begin{subequations}
\begin{align}
    \label{Atime2}
    A_{\epsilon,\mu} &=  \sqrt{\frac{\gamma(\epsilon)}{2\pi T_a}}\int_{-T_a/2}^{T_a/2}\sum_j e^{i\epsilon t} U_{\mu j }(\epsilon) A^j(t)dt \\
    &= \sum_{\omega,j} f(\epsilon,\mu,j,\omega) A^j_\omega \ , \quad f(\epsilon,\mu,j,\omega) =  \sqrt{\frac{D_\mu(\epsilon)T_a}{2\pi}} \textrm{sinc}[T_a(\epsilon-\omega)/2]U_{\mu j }(\epsilon) \ .
\end{align}
\end{subequations}
These results generalize Eqs.~\eqref{eq:Lind}-\eqref{eq:newLindops} to the $n$-qubit setting. Now the locality is more directly defined: if the operators $A_{\epsilon,\mu}$ can be approximately truncated to a ball around some point for each $\mu$, then  the CGME for $n$ qubits is local in the strong sense described above; otherwise it is local in a weak sense. In any case, it has $n$ times more generator terms than the single-qubit CGME. The latter has a constant number of terms after the discretization described in Appendix \ref{app:C}. 

Let us now discuss the range of applicability of the $n$-qubit CGME. In the case of correlated noise, the norm of each $A$ effectively grows by a factor of $n$. Thus the bound on the relaxation time becomes $\Lambda \sim  n^2/\tau_{SB}$ with $\tau_{SB}$ taken from the single qubit case. For uncorrelated noise, the bound is $\Lambda \sim n/\tau_{SB}$. In both cases, the range of applicability can be derived by directly repeating the derivation given in Sec.~\ref{errSec}:
\begin{equation}
    \|\rho_C(t) - \rho_{\text{true}}(t) \|_1 \leq O(\sqrt{\Lambda \tau_B} e^{1.5\Lambda t })\ ,
    \label{eq:193}
\end{equation}
where $\tau_B$ is system-independent, and thus is the same as the single qubit quantity. The apparently shrinking range with $n$ should not discourage us, since it can be mitigated under the assumption of locality. Namely, one can prove that the error in local density matrices obeys a stronger bound consistent with the isolated qubit result. Let us note that in an $n$-qubit system there are in fact nonlocal modes that relax with the rate $n/\tau_{SB}$, or $n^2/\tau_{SB}$ for correlated noise (e.g., superradiance), as well as modes whose relaxation is completely suppressed (e.g., subradiance, or decoherence-free subspaces, as mentioned above). These phenomena are well known~\cite{LidarWhaley:03} and experimentally documented~\cite{Kwiat:00,Viola:2001sp,Kielpinski:01,PhysRevLett.91.217904}.

The benefit expressed by Eq.~\eqref{eq:193} is that the requirement on the single qubit $\tau_B/\tau_{SB}$ value 
switches from being exponential in the system size for the Davies-Lindblad case to polynomial for the CGME and the Redfield master equation. In terms of the coupling strength, we found these ranges of applicability at the relevant times of $\Lambda^{-1}$ to be between $\sqrt{n}g\tau_B \ll1$ and $ng\tau_B \ll1$ depending on the correlations in the noise. For local observables the earlier range $g\tau_B\ll1$ for times $\sim \tau_{SB}$ is expected to hold. Here the inequality notation ($\ll$) guarantees that the error $\|\delta \rho\|_1$ or $\|\delta\rho_{\text{loc}}\|_1$ is small.

\section{The CGME with a time-dependent system Hamiltonian}
\label{TD}

We start again from the total Hamiltonian~\eqref{eq:Htot}, but this time we let the system Hamiltonian be time-dependent: $H(t)$.
We denote the interaction term $A\otimes B$ by $V$. To go to the interaction picture, we now need to specify both the start time $t_i$ and the final time $t \leq t_f$. The unitary propagator $U_{tt_i}$ satisfies
\begin{equation}
    \frac{dU_{tt_i}}{dt} =-i H(t)U_{tt_i}, \quad U_{t_it_i} =I\ .
\end{equation}
The formal solution is
\begin{equation}
    U_{t_ft_i} = \mathcal{T}\textrm{exp}\left(-i\int_{t_i}^{t_f}H(t)dt\right)\ ,
\end{equation}
where $\mathcal{T}$ represents  forward time-ordering.
Some basic additional  facts are:
\beq
U_{t_3t_1} = U_{t_3t_2}U_{t_2t_1} \ , \quad 
\frac{dU_{tt_i}}{dt_i} =i U_{tt_i}H(t)\ , \quad U_{t_1t_2}^\dag = U_{t_2t_1}\ .
\end{equation}
The relation between the Schr\"odinger picture and the interaction picture is:
\begin{equation}
    \rho(t) =U_{t0}\rho_I(t)U_{0t}, \quad V_I = V(t) =A(t,0)\otimes B(t)= U_{0t}AU_{t0} \otimes B(t) \ .
    \label{eq:48}
\end{equation}
Here $B(t)$ is given by the bath Hamiltonian in the same way as in the previous derivation.  We again assume the bath state to be stationary for simplicity, although the derivation will go through without it. At this point we repeat the steps involved in Eqs.~\eqref{eq:16}-\eqref{RedStart}, with the only change being that the $A(t)$ operators now acquire a second time variable. Namely, Eq.~\eqref{RedStart} --- the Redfield equation in the interaction picture --- is replaced by
%
\begin{subequations}
\label{eq:Redfield-t-dep}
\begin{align}
{\dot\rho}_{B,I}(t) &=\int_0^t \mathcal{C}(\tau-t) (A(t,0) \rho_{B,I}(t) A(\tau,0) - \rho_{B,I}(t) A(\tau,0) A(t,0)) d\tau +\textrm{h.c.} + \mathcal{E}_M\\
&=\int_0^t  d\tau \left[\mathcal{C}(\tau-t) (A(t,0) \rho_{B,I}(t) A(\tau,0) - \rho_{B,I}(t)A(\tau,0) A(t,0)) \right. \notag \\
&\left. \qquad \qquad +\mathcal{C}(t-\tau) (A(\tau,0) \rho_{B,I}(t) A(t,0) - A(t,0)A(\tau,0)\rho_{B,I}(t) ) \right] + \mathcal{E}_M\ ,
\end{align}
\end{subequations}
where we again introduced the Markov approximation $\rho_{B,I}(\tau) \mapsto \rho_{B,I}(t)$.

Let us digress briefly to derive the time-dependent Redfield equation. To do so, we transform Eq.~\eqref{eq:Redfield-t-dep} out of the interaction picture:
\begin{equation}
{\dot\rho}_{BM}(t) = -i[H(t),\rho_{BM}] + \int_0^t \mathcal{C}(\tau-t) (A \rho_{B}(t) A(\tau,t) - \rho_{B}(t) A(\tau,t) A +\textrm{h.c.}) d\tau\ .
\end{equation}
Changing variables $t-\tau =t'$ and switching the $dt'$ integration limit to $t\to \infty$ we obtain:
\begin{equation}
   {\dot\rho}_{R}(t) = -i[H(t),\rho_{R}] +(A \rho_{B}(t) A_f(t) - \rho_{B}(t) A_f(t) A d\tau +\textrm{h.c.}), \quad A_f(t) = \int_{0}^\infty \mathcal{C} (-t')A(t-t',t)dt'\ ,
   \label{eq:t-dep-Redfield}
\end{equation}
which is the time-dependent Redfield equation.

We now proceed with essentially the same series of transformations we used in the derivation for the time-independent Hamiltonian in Sec.~\ref{sec:derivation}. The main difference is that there are now additional time variables.
Similarly as in the transition to Eq.~\eqref{eq:22}, we time-average the r.h.s. by shifting $t\mapsto t+t_1$ and applying  
$\frac{1}{T_a} \int_{-T_a/2}^{T_a/2} dt_1$.
The new equation is:
\begin{align}
    {\dot\rho}_{BMT,I}(t) = \frac{1}{T_a}\int_{-T_a/2}^{T_a/2}dt_1\int_0^{t+t_1}  d\tau \left[\mathcal{C}(\tau-t-t_1) (A(t+t_1,0) \rho_{BMT,I}(t)A(\tau,0) - \rho_{BMT,I}(t)A(\tau,0) A(t+t_1,0)) \right. \notag \\
    \left. +\mathcal{C}(t+t_1-\tau) (A(\tau,0)\rho_{BMT,I}(t) A(t+t_1,0) - A(t+t_1,0)A(\tau,0)\rho_{BMT,I}(t))\right] \ ,
\end{align}
and as long as $T_a$ is small compared to the fastest timescale over which $\rho_I(t)$ changes, the associated error $\|\rho_{BM}(t) - \rho_{BMT}(t)\|_1$ is small (see Sec.~\ref{sec:t-ave-err2} for a detailed analysis of the time-independent case). 
We now rotate out of the interaction picture using Eq.~\eqref{eq:48}:
\begin{align}
    {\dot\rho}_{BMT}(t) &=-i[H(t),\rho_{BMT}(t)] \notag \\
    &\quad +  \frac{1}{T_a}\int_{-T_a/2}^{T_a/2}dt_1\int_0^{t+t_1}  d\tau \left[\mathcal{C}(\tau-t-t_1) (A(t+t_1,t) \rho_{BMT}(t) A(\tau,t) - \rho_{BMT}(t) A(\tau,t) A(t+t_1,t)) \right. \notag \\
    &\quad \left. +\mathcal{C}(t+t_1-\tau) (A(\tau,t) \rho_{BMT}(t) A(t+t_1,t) - A(t+t_1,t)A(\tau,t)\rho_{BMT}(t)  )\right]\ .
\end{align}
Next, we change variables to $t_2 = \tau -t$:
\begin{align}
    {\dot\rho}_{BMT}(t) &=-i[H(t),\rho(t)] \notag \\
    &\quad +  \frac{1}{T_a}\int_{-T_a/2}^{T_a/2}dt_1\int_{-t}^{t_1}  dt_2 \left[ \mathcal{C}(t_2-t_1) (A(t+t_1,t) \rho_{BMT}(t) A(t+t_2,t) - \rho_{BMT}(t) A(t+t_2,t) A(t+t_1,t)) \right. \notag \\
    &\quad \left. +\mathcal{C}(t_1-t_2) (A(t+t_2,t) \rho_{BMT}(t) A(t+t_1,t) - A(t+t_1,t)A(t+t_2,t)\rho_{BMT}(t)  )\right]\ .
\end{align}
The transformations we performed after coarse graining of the Redfield equation~\eqref{eq:Redfield-t-dep} were exact. We know what to do to recover complete positivity: as in the transition to Eq.~\eqref{eq:26} in the time-independent case, we need to neglect a part of the integral by changing the lower limit from $-t$ to $-T_a/2$:
\begin{align}
    {\dot\rho}_{BMT}(t) 
    =-i[H(t),\rho_{BMT}(t)] +  P_1 +P_2 -\rho_{BMT}(t)Q_1-Q_2\rho_{BMT}(t) + \mathcal{E}_p,
    \label{eq:rhodotPQ}
\end{align}
where we defined
\bes
\begin{align}
    P_1 &\equiv \frac{1}{T_a}\int_{-T_a/2}^{T_a/2}dt_1\int_{-T_a/2}^{t_1}  dt_2 \mathcal{C}(t_2-t_1) A(t+t_1,t) \rho_{BMT}(t) A(t+t_2,t)\\
    P_2 &\equiv\frac{1}{T_a}\int_{-T_a/2}^{T_a/2}dt_1\int_{-T_a/2}^{t_1}  dt_2 \mathcal{C}(t_1-t_2) A(t+t_2,t) \rho_{BMT}(t) A(t+t_1,t)\\
Q_1 &\equiv \frac{1}{T_a}\int_{-T_a/2}^{T_a/2}dt_1\int_{-T_a/2}^{t_1}  dt_2 \mathcal{C}(t_2-t_1)  A(t+t_2,t) A(t+t_1,t)\\
Q_2 &\equiv \frac{1}{T_a}\int_{-T_a/2}^{T_a/2}dt_1\int_{-T_a/2}^{t_1}  dt_2 \mathcal{C}(t_1-t_2) A(t+t_1,t)A(t+t_2,t)\ .
\end{align}
\ees
The same estimates on the error incurred by doing this apply, i.e., this introduces an error of order $ \|\mathcal{E}_p(\text{const}\cdot \tau_{SB})\|= O(\tau_B/T_a)$. Complete positivity yet remains to be demonstrated. 
To do so, we first note that we may interchange the order of integration since $\int_{-T_a/2}^{T_a/2}dt_1\int_{-T_a/2}^{t_1}  dt_2 = \int_{-T_a/2}^{T_a/2}dt_2\int_{t_2}^{T_a/2}  dt_1$. We then swap the integration variables $t_1$ and $t_2$ and obtain:
\begin{subequations}
\begin{align}
P_2 &=\frac{1}{T_a}\int_{-T_a/2}^{T_a/2}dt_1\int_{t_1}^{T_a/2}  dt_2 \mathcal{C}(t_2-t_1) A(t+t_1,t) \rho_{BMT}(t) A(t+t_2,t)\\
Q_2 &= \frac{1}{T_a}\int_{-T_a/2}^{T_a/2}dt_1\int_{t_1}^{T_a/2}  dt_2 \mathcal{C}(t_2-t_1) A(t+t_2,t) \rho_{BMT}(t) A(t+t_1,t)\ ,
\end{align}
\end{subequations}    
so that:
\begin{subequations}
\label{eq:55}
\begin{align}
P_1+P_2 &=\frac{1}{T_a}\int_{-T_a/2}^{T_a/2}dt_1\int_{-T_a/2}^{T_a/2}  dt_2 \mathcal{C}(t_2-t_1) A(t+t_1,t) \rho_{BMT}(t) A(t+t_2,t)\\
X &\equiv Q_1+Q_2 =\frac{1}{T_a}\int_{-T_a/2}^{T_a/2}dt_1\int_{-T_a/2}^{T_a/2}  dt_2 \mathcal{C}(t_2-t_1)  A(t+t_2,t) A(t+t_1,t)\\
Y &\equiv -i(Q_1-Q_2) = \frac{1}{iT_a}\int_{-T_a/2}^{T_a/2}dt_1\int_{-T_a/2}^{T_a/2}  dt_2 \sgn(t_1-t_2)\mathcal{C}(t_2-t_1)  A(t+t_2,t) A(t+t_1,t)\ ,
\end{align}
\end{subequations}    
and note that $X$ and $Y$ are both Hermitian.
Also, $Q_1 = \frac{1}{2} (X+iY)$ and $Q_2 =\frac{1}{2} (X-iY)$, so that
\begin{equation}
   \rho_{BMT} Q_1+Q_2\rho_{BMT} = -\frac{i}{2}[Y,\rho_{BMT}]+\frac{1}{2}\{X,\rho_{BMT}\}\ .
   \label{eq:57}
\end{equation}
%
Combining Eq.~\eqref{eq:rhodotPQ} with the expressions in Eqs.~\eqref{eq:55}-\eqref{eq:57}, we thus find
\begin{align}
    {\dot\rho}_C(t) =-i[H(t)+H_{\mathrm{LS}}(t),\rho_C(t)] +  \frac{1}{T_a}\int_{-T_a/2}^{T_a/2}dt_1\int_{-T_a/2}^{T_a/2}  dt_2 C(t_2-t_1) [A(t+t_1,t) \rho_C(t) A(t+t_2,t) \notag \\
    -\frac{1}{2}\{ A(t+t_2,t) A(t+t_1,t),\rho_C(t) \} ] \ ,
\label{eq:57-2}
\end{align}
where 
\beq
H_{\mathrm{LS}}(t) \equiv -\frac{1}{2}Y = \frac{i}{2T_a}\int_{-T_a/2}^{T_a/2}dt_1\int_{-T_a/2}^{T_a/2}  dt_2 \sgn(t_1-t_2)\mathcal{C}(t_2-t_1)  A(t+t_2,t) A(t+t_1,t)\ .
\label{eq:HLS(t)}
\eeq
Finally, we establish complete positivity by again introducing the Fourier transform $C(t) = \frac{1}{2\pi}\int_{-\infty}^{\infty} e^{-i\epsilon t}\gamma(\epsilon)d\epsilon$ [Eq.~\eqref{eq:C(t)}], where 
$\gamma(\epsilon) >0$ (Appendix~\ref{app:A}). In analogy with Eq.~\eqref{eq:newLindops} for the time-independent case, let us also define
\begin{equation}
    A_\epsilon(t) \equiv \sqrt{\frac{\gamma(\epsilon)}{2\pi T_a}}\int_{-T_a/2}^{T_a/2} e^{i\epsilon t_1}A(t+t_1,t)dt_1  \ .
\end{equation}
We then obtain from Eq.~\eqref{eq:57-2} a master equation in exactly the same form as the time-independent CGME [Eq.~\eqref{eq:Lind}], except that all the operators now depend on time: 
\begin{equation}
{\dot\rho}_C(t) =-i[H(t) +H_{\mathrm{LS}}(t),\rho_C] + \int_{-\infty}^{\infty}d\epsilon \left(A_{\epsilon}(t) \rho_C A_{\epsilon}^\dag(t) - \frac{1}{2}\left\{ \rho_C , A_{\epsilon}^\dag(t) A_{\epsilon}(t)\right\} \right)\ .
\label{eq:new-CGME}
\end{equation}
This result, which is in a manifestly completely positive form, establishes that the CGME is valid also for arbitrary time-dependent system Hamiltonians. The time-dependent CGME, Eq.~\eqref{eq:new-CGME} [also stated earlier as Eq.~\eqref{tdCoarse}], is one of the main new results of this work.


\section{Applications and Examples}
\label{scalSec}

\subsection{Application: dynamical decoupling}
\label{ddSec}

One interesting application of our time-dependent CGME is to study methods to reduce decoherence. One such method is dynamical decoupling (DD) \cite{Viola:98,Zanardi:1999fk} --- the use of pulse sequences to cancel out the system-bath interaction (for reviews see~\cite{Lidar-Brun:book,Suter:2016aa}). Even though the master equation has a Markovian appearance, the underlying bath can have a nonzero correlation time, and can therefore describe the effects of DD. At one level this might appear to be a counterintuitive result since it is widely accepted that in the limit $\mathcal{C}(t) = \delta(t)$ DD is useless,  and hence one might naively expect that one needs to use non-Markovian methods (e.g., the Magnus expansion to second order or higher~\cite{Ng:2011dn}) to describe the benefits of DD. However, it has recently been recognized that DD can in fact work in the Markovian setting as well, specifically the Davies-Lindblad master equation~\cite{Szczygielski:2015aa}, and a more abstract semigroup framework~\cite{Gough:2017aa} (see also Refs.~\cite{Addis:2015aa,Arenz:2018aa}). Our contribution here is to establish this in the framework of the Redfield master equation and, more importantly, the CGME. This ensures a wider range of applicability than the Davies-Lindblad master equation result. Since our result is obtained in a strictly Markovian setting, we differ with the conclusion reported in Ref.~\cite{Li:2018aa}, that ``success, however limited, of DD is a meaningful concept of non-Markovianity".

We consider two examples, both for a qubit with $H=0$ coupled to a purely dephasing environment: 
\begin{equation}
    H_{\textrm{tot}} = Z \otimes B + I\otimes H_\text{b}\ ,
    \label{eq:HZ}
\end{equation}
where $Z \equiv \sigma^z$ is the Pauli $z$-matrix.
The noise model~\eqref{eq:HZ} can be decoupled with a sequence of $X$ pulses spaced by $\Delta t$. Here by a pulse we mean an instantaneous unitary being applied to the qubit, through the action of an externally controlled, time-dependent qubit Hamiltonian: instead of $H=0$ we have a time-dependent system Hamiltonian $H(t) =\frac{\pi}{2} \sum_{j}\delta (t- j\Delta t)X$, where the prefactor of $\pi/2$ is chosen so that $X$-pulses are generated.

\subsubsection{Example 1: time-dependent Redfield master equation with a rectangle bath correlation function}
Consider the rectangle bath correlation function
\begin{equation}
   \mathcal{C}(t) = g^2 \theta (\tau_c-|t|) \ , 
   \label{tau_c}
\end{equation}
where we have explicitly introduced the coupling strength $g$.
To analyze this problem it will suffice to use the time-dependent version of the Redfield equation, Eq.~\eqref{eq:t-dep-Redfield}. 
We thus need to compute
\begin{equation}
    A_f(t) =\int_{0}^\infty \mathcal{C} (t')A(t-t',t)dt'\ ,
    \label{eq:201}
 \end{equation}
and note that the time-evolution of the operator $A=Z$ under $X$ pulses results in $A(t-t',t) = \pm Z$ with the sign alternating every $\Delta t$:
\begin{equation}
  A(t-t',t) =(-1)^{\lceil t/\Delta t\rceil - \lceil (t-t')/\Delta t\rceil}Z\ ,
\end{equation}
Taking the integral, we find:
 \begin{equation}
     A_f(t) =g^2 Z \int_0^{\tau_c} (-1)^{\lceil t/\Delta t\rceil - \lceil (t-t')/\Delta t\rceil} dt'= g^2 c(t) Z, \quad |c(t)| \leq \Delta t\ ,
 \end{equation}
 where $c(t)$ is the remainder after the cancellation of different signs. To show why $|c(t)| \leq \Delta t$, note that there is no way to accumulate $>\Delta t$ by integrating any interval of the function $ (-1)^{\lceil x/\Delta t\rceil}$
 \begin{equation}
    \forall ~ a,b: ~ \left|\int_a^b (-1)^{\lceil x/\Delta t\rceil} dx \right|\leq \Delta t
 \end{equation}
To see this, we note that a translation $a,b \to a+\Delta t, b+\Delta t$ changes the sign of the integral. We then split the integration interval $[a,b]$ into three subsets: $
 \{[a+2n\Delta t ,a+(2n+1)\Delta t ]\},~\{[a+(2n+1)\Delta t ,a +2(n+1)\Delta t ]\},~ [a +2\Delta t (n_{\text{max}}+1), b] $ where $n_{\text{max}}$ has been chosen closest to $b$: $n_{\text{max}} = \lceil(b-a)/2\Delta t \rceil-2$. The integrals over the first two subsets cancel due to the translation property, and the remaining subset $[a +2\Delta t (n_{\text{max}}+1), b]$ has length $L<2\Delta t$. Note that for $L=0$ and $L=2\Delta t$ the integral vanishes. Also note that the derivative of the integral with respect to $b$ is $(-1)^{\lceil b/\Delta t\rceil}$, and its absolute value is $1$. Any function with such a derivative, up to two inflection points and zeros at the end of the interval $[0,2L]$ stays within $[-\Delta t, +\Delta t]$, and the bound is saturated when $b =m\Delta t$.
 
 We compare this with the case without pulses:
 \begin{equation}
      A_f(t)= g^2\tau_c Z\ ,
 \end{equation}
Here $\tau_c$ enters the correlation function [see Eq.~\eqref{tau_c}].  
Thus the decoherence rate is reduced by a factor of at least $\Delta t/\tau_c$. For $(g \tau_c)^2 \ll 1$ our analysis is accurate with the error bounds discussed above. 

However, we note that the bath correlation function $\mathcal{C}(t) = g^2 \theta (\tau_c-|t|) $ used here for simplicity is not physically permissible since its Fourier transform $\gamma(\omega)<0$ for some $\omega$. This leads to possibly nonpositive $c(t)$ for some $t$, and then the equation does not describe a CP map. To address this we next perform a more careful calculation with a realistic bath. 
 
\subsubsection{Example 2: time-dependent CGME with an Ohmic spectral density}
Now consider a physically permissible environment with a spectral density that satisfies the KMS condition~\eqref{eq:KMS}.
We specify $\gamma(\omega)$ below.
We wish to use the time-dependent CGME [Eq.~\eqref{tdCoarse}] to study the same dynamical decoupling protocol as in the previous example. We first need to compute $A(t+t_1,t) = U^\dagger (t+t_1,t)Z U(t+t_1,t)$, where $U(t+t_1,t) =  \exp[-i\int_t^{t+t_1} H(s) ds]$. This unitary is a product of as many $X$-pulses as fit in the interval $[t,t+t_1]$, which we can express as
\beq
U(t+t_1,t) = X^{m_f-m_i} \ , \quad m_i = \lceil t/\Delta t\rceil\ , \quad m_f = \lceil (t+t_1)/\Delta t\rceil \ ,
\eeq
where $m_i$ and $m_f-1$ respectively count the number of pulses in the intervals $[0,t]$ and $[0,t+t_1]$. Conjugating $Z$ by these pulses results in alternating signs, depending on the parity of $m_f-m_i$:
\beq
A(t+t_1,t) = (-1)^{\Delta m} Z\ , \quad \Delta m = m_f-m_i\ .
\eeq
This yields the time-dependent Lindblad operators $A_\epsilon(t) $ via Eq.~\eqref{eq:14}:
\beq
A^{\text{DD}}_\epsilon(t) = f^{\text{DD}}_\epsilon(t) Z \ , \quad f^{\text{DD}}_\epsilon = \sqrt{\frac{\gamma(\epsilon)T_a}{2\pi}}\int_{-1/2}^{1/2} e^{i\epsilon \zeta T_a}(-1)^{\lceil (t+\zeta T_a)/\Delta t\rceil -\lceil t/\Delta t\rceil} d\zeta \ .
    \eeq
Note that the Lamb shift [Eq.~\eqref{eq:15}] is proportional to $Z^2=\openone$, so it drops out. The Lindblad operators without DD are simply
\beq
A^{\text{noDD}}_\epsilon(t)  =  f^{\text{noDD}}_\epsilon(t) Z \ , \quad f^{\text{noDD}}_\epsilon(t) =  \sqrt{\frac{\gamma(\epsilon)}{2\pi T_a}}\int_{-T_a/2}^{T_a/2} e^{i\epsilon t_1}dt_1 = \sqrt{\frac{\gamma(\epsilon)T_a}{2\pi}} \sinc(\epsilon T_a/2) \ .
\eeq
We may thus write the time-dependent CGME as:
\begin{equation}
{\dot\rho}_C(t) =-i[H(t),\rho_C] + r^x(t)\left( Z\rho_C Z -  \rho_C \right)\ ,
\end{equation}
where
\beq
r^{x}(t) =  \int_{-\infty}^{\infty}d\epsilon \left| f^{x}_\epsilon \right|^2\ ,
\eeq
with $x=\text{DD}$ or $x=\text{noDD}$. The decoherence rate is thus suppressed by the factor 
\beq
\xi(t)\equiv \frac{r^{\text{DD}}(t)}{r^{\text{noDD}}(t)} = \frac{ \int_{-\infty}^{\infty}d\epsilon {\gamma(\epsilon)} \left|\int_{-1/2}^{1/2} e^{i\epsilon \zeta T_a}(-1)^{\lceil (t+\zeta T_a)/\Delta t\rceil -\lceil t/\Delta t\rceil} d\zeta \right|^2}{ \int_{-\infty}^{\infty}d\epsilon {\gamma(\epsilon)}\left|\sinc\left(\frac{\epsilon T_a}{2}\right) \right|^2}\ .
\eeq
The inner integral can be evaluated under certain simplifying assumptions. For example, let us assume that $T_a$ is an even integer multiple of $\Delta t$: $T_a = 2k \Delta t$, and let us consider only times $t$ which are integer multiples of $\Delta t$: $t=\ell \Delta t$. Then $\lceil (t+\zeta T_a)/\Delta t\rceil -\lceil t/\Delta t\rceil = \lceil 2k \zeta \rceil$, and the dependence on $t$ cancels out, i.e., the suppression factor $\xi(t)$ is time-independent. We can then rewrite the inner integral as:
\bes
\begin{align}
\left|\int_{-1/2}^{1/2} e^{i\epsilon \zeta T_a}(-1)^{\lceil (t+\zeta T_a)/\Delta t\rceil -\lceil t/\Delta t\rceil} d\zeta \right|^2
&= \left|\sum_{j=0}^{k-1} (-1)^j \int_{-\frac{1}{2}+\frac{j}{k}}^{-\frac{1}{2}+\frac{j+1}{k}} e^{i\epsilon\zeta T_a}d\zeta\right|^2  \quad (T_a = 2k \Delta t,\  t=\ell \Delta t)\\
&= \left|\sum_{j=0}^{k-1} (-1)^j \frac{e^{i\epsilon T_a[2(j-k)+1]/(2k)}}{k}\sinc\left(\frac{\epsilon T_a}{2k}\right)\right|^2 \\
&= \left|\frac{2}{\epsilon T_a}\sin\left(\frac{k\pi}{2}+\frac{\epsilon T_a}{2}\right)\tan\left(\frac{\epsilon T_a}{2k}\right)\right|^2 \\
& = \left|\sinc\left(\frac{\epsilon T_a}{2}\right)\tan\left(\frac{\epsilon T_a}{2k}\right)\right|^2  \quad (k \text{ even}) \ ,
\end{align}
\ees
where in the last line we assumed for further simplicity assume that $k$ is even ($k=2k'$), so that $\sin\left(\frac{k\pi}{2}+\frac{\epsilon T_a}{2}\right) = \pm \sin\left(\frac{\epsilon T_a}{2}\right)$. Writing $t=\ell \Delta t$ and $T_a = {4k'} \Delta t$, we thus have:
\beq
\xi(\ell \Delta t) = \frac{ \int_{-\infty}^{\infty}d\o {\gamma(\o)} \left|\sinc\left({2k'}\o \Delta t\right)\tan\left(\o \Delta t\right) \right|^2}{ \int_{-\infty}^{\infty}d\o {\gamma(\o)}\left|\sinc\left({2k'}\o \Delta t\right) \right|^2}\quad ({k'}\in \mathbb{N})\ .
\label{eq:203}
\eeq
To guarantee that DD will cause suppression [$\xi(\ell \Delta t)<1$] it is sufficient for the spectral density to have a high-frequency cutoff $\o_c$ such that $|\tan\left(\o \Delta t\right)|<1$, i.e., 
\beq
\o_c \Delta t <\frac{\pi}{4}\ ,
\label{eq:DD-cond}
\eeq 
in agreement with well established results~\cite{Viola:98}.

An important physical example is a bath with an Ohmic spectral density
\beq
\gamma(\o) = 2\pi \kappa \frac{\o e^{-|\o|/\o_c}}{1-e^{-\b \o}}\ ,\label{Ohmm}
\eeq
where $\o_c$ is a high-frequency cutoff and $\kappa$ is a positive dimensionless constant. This spectral density satisfies the KMS condition. The associated bath correlation function can be expressed in terms of the Polygamma function (see Appendix I of Ref.~\cite{ABLZ:12-SI}), but unfortunately evaluating $\tau_{SB}$ and $\tau_B$ analytically using Eq.~\eqref{eq:T1tauB} is not possible in this case,
so we simply pick a convenient value of $T_a$ ($=4k'\Delta t$) rather than its optimal value $\sqrt{\tau_B \tau_{SB}/5}$. The suppression factor is plotted in Fig.~\ref{fig:DD} for a variety of parameter settings, including for $\Delta t>{\pi}/(4\o_c)$. It can be seen that the CGME correctly and consistently predicts that DD will result in the suppression of dephasing when Eq.~\eqref{eq:DD-cond} is satisfied. Thus, the CGME can be used to study DD protocols despite its Markovian appearance.

\begin{figure}[t]
         \subfigure[\ $\beta=0.2$]{\includegraphics[width=0.40\textwidth]{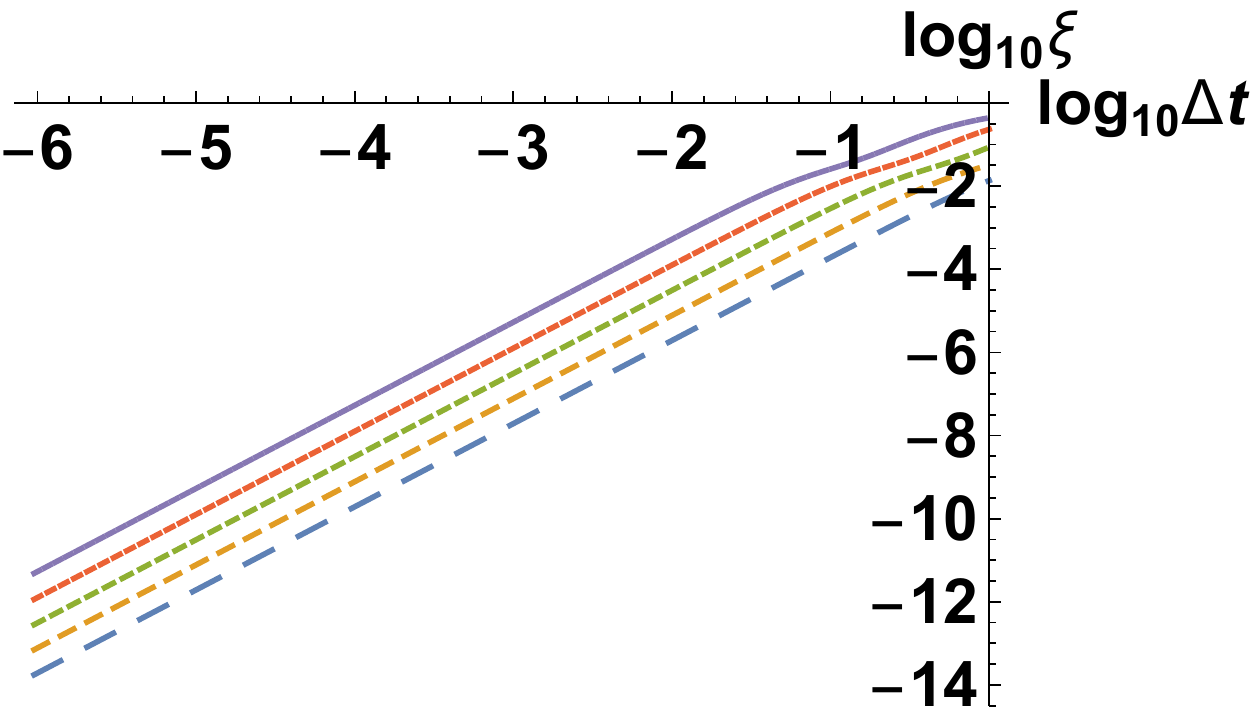}\label{DD-suppression1}}
         \subfigure[\ $\beta=5$]{\includegraphics[width=0.40\textwidth]{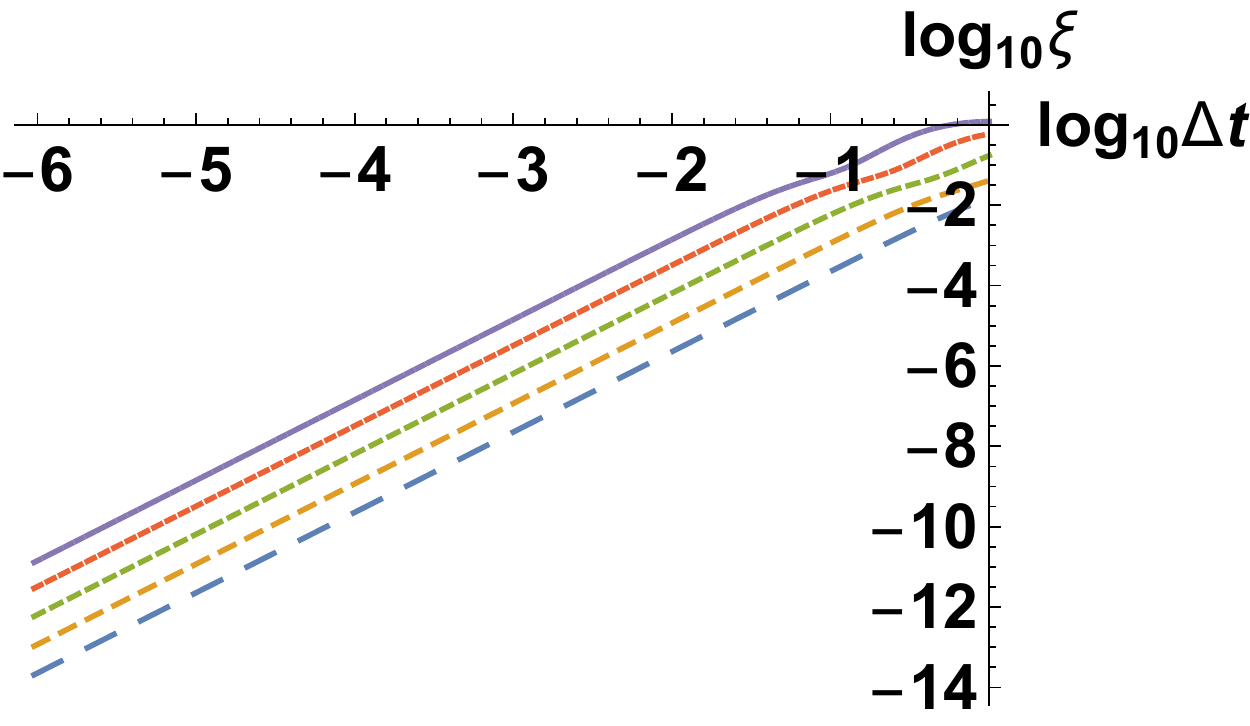}\label{DD-suppression2}}
         \subfigure{\includegraphics[width=0.15\textwidth]{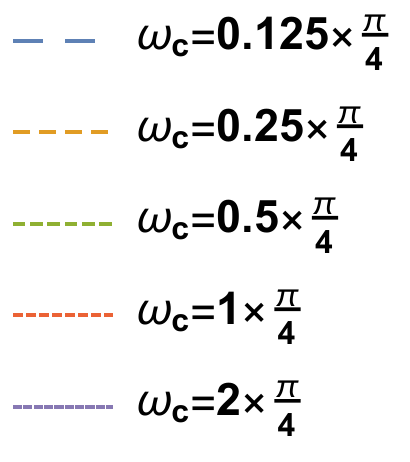}}
\caption{The DD suppression factor $\xi$ [Eq.~\eqref{eq:203}] as a function of the pulse interval $\Delta t$ for different high-frequency cutoffs $\o_c$, at two different inverse temperatures (a) $\beta=0.2$ and (b) $\beta=5$. We have set $T_a={4\Delta t}$ and $\kappa=1$. The sufficient condition~\eqref{eq:DD-cond} is satisfied everywhere, except for the case $\o_c=2\times\pi/4$ (purple line) for $\Delta t\lesssim 1$, where indeed for the $\beta=5$ case the suppression factor is slightly $>1$.}
\label{fig:DD}
\end{figure}

\subsection{Lambless master equations}
\label{sec:Lambless}

It turns out that the Lamb shift terms containing $\sgn(t)$, such as Eq.~\eqref{eq:LS-final}, are more demanding to compute numerically, but at the same time the interesting dynamics is often given by the other, relaxation terms. Therefore it is convenient to have a version of the equations in the ``Lambless" regime. Even though these equations do not have any correctness guarantees, one can still use them as toy models, and their numerical solution is often significantly faster than their complete counterparts. Let us list the resulting equations:
\begin{itemize}
    \item{Lambless Davies-Lindblad [replacing Eq.~\eqref{eq:Davies-Lind}]:}
\begin{equation}
{\dot\rho}_{LD}(t) =-i[H  ,\rho] +\frac{1}{2}\sum_{\omega} \gamma(\omega)  (A_{\omega} \rho_{LD} A_{-\omega} - \rho_{LD} A_{-\omega} A_{\omega}) +\textrm{h.c.}
\end{equation}

\item{Lambless Redfield [replacing Eq.~\eqref{eq:Red}]:}
\begin{equation}
{\dot\rho}_{LR}(t) =-i[H,\rho_{LR}(t)] +(A \rho_{LR} A_f - \rho_{LR}  A_f A +\textrm{h.c.})\ ,
\end{equation}
where
\begin{equation}
    A_f = \frac{1}{2}\sum_\omega \gamma(-\omega)A_\omega \ .
\end{equation}
This simplified form of $A_f$ follows from Eqs.~\eqref{eq:6c} and~\eqref{eq:Af} by substituting $S(\o)=0$ (vanishing Lamb shift). When the KMS condition holds we may write $A_f = \frac{1}{2}\sum_\omega e^{-\b\o} \gamma(\omega)A_\omega$.

\item{Lambless CGME [replacing Eq.~\eqref{eq:Lind}]:}

\begin{equation}
{\dot\rho}_{LC}(t) =-i[H ,\rho_{LC}] + \Delta \epsilon \sum_{k, \epsilon = \Delta \epsilon k, |k|< k^*} (A_{\epsilon} \rho_{LC} A_{\epsilon}^\dag - \frac{1}{2}\left\{ \rho_{LC} , A_{\epsilon}^\dag A_{\epsilon}\right\} ) 
\label{eq:CGME-L1}
\end{equation}
where
\begin{equation}
    A_\epsilon= \sum_\omega A_\omega \sqrt{\frac{\gamma(\epsilon)T_a}{2\pi}} \sinc\left[T_a(\epsilon-\omega)/2\right] \ , 
    \label{eq:CGME-L2}
\end{equation}
and we choose
$T_{a,\text{opt}} =\sqrt{\tau_B \tau_{SB}/5}$, $\Delta \epsilon$, and $k^*$ as prescribed in  Appendix~\ref{app:C}. We study this case numerically in Sec.~\ref{sec:CGME-L}.
\end{itemize}
The importance of the Lamb shift in open system dynamics has been analyzed before, e.g., in Refs.~\cite{ABLZ:12-SI,Vega:2010fk}

\subsection{Toy example of a slow bath at the boundary of the range of applicability, with explicit correlation function and spectral density}
\label{numeric}

Another way to speed up numerical simulations is to have a spectral density that leads to an explicit, analytically computable form of the bath correlation function $\mathcal{C}(t)$. Note that in the ``Lambless" CGME case discussed in Sec.~\ref{sec:Lambless} 
the only direct role played by $\mathcal{C}(t)$ is in the determination of $\tau_B$ and $\tau_{SB}$ [via Eq.~\eqref{eq:T1tauB}], and apart from this it suffices to specify only $\gamma(\omega)$.
The complete equations are more complicated: the Lamb terms contain $\mathcal{C}(t)$, or a Cauchy principal value integral of $\gamma(\omega)$ [Eq.~\eqref{eq:S-gamma}]. 

We present such a toy example with the property of having a bath almost at the boundary of the slowest correlation functions allowed within the range of applicability of the CGME. Baths with the same $\tau_B$ can have a different time decay of $|\mathcal{C}(t)|$ at large times $t\gg \tau_B$. A bath with $|\mathcal{C}(t)| \sim 1/t^2$ will have a logarithmically divergent $\tau_B$, which is undesirable for a toy example.%
\footnote{$|\mathcal{C}(t)| \sim 1/t^2$ is actually the borderline case realized by an Ohmic bath, for which the divergence has a cutoff depending on the total time of the experiment $T$ in Eq.~\eqref{eq:T1tauB-b}.
} 
Consider a bath with the following spectral density, which satisfies the KMS condition, such that $|\mathcal{C}(t)| \sim 1/t^4$---almost as slow as we are allowed: 
\begin{equation}
    \gamma(\omega) =\mathcal{N}\frac{e^{\beta \omega/2}}{\tau_{SB}}( e^{ - b  \beta |\omega| } -  a^{-1}e^{ - ab  \beta |\omega|})\ ,
    \label{eq:222}
\end{equation}
where $a>1$, $b>1/2$. Since $\frac{1}{\tau_{SB}}= \int_0^{\infty}|\frac{1}{2\pi}\int_{-\infty}^{\infty} e^{-i\o t} \gamma(\o)d\o|dt$ [Eqs.~\eqref{eq:gamma31},~\eqref{eq:T1tauB-a}], the normalization factor is
\begin{equation}
   \mathcal{N}^{-1} =\frac{1}{2\pi}\int_0^\infty \left| \int_{-\infty}^{\infty}e^{-i\omega t + \beta \omega/2}( e^{ - b | \beta \omega| } -  a^{-1}e^{ - ab | \beta \omega|})d\o\right| dt  \ .
\end{equation}
The difference between the two exponentials in Eq.~\eqref{eq:222} is such that for small $\omega$ the term linear in $|\omega|$ cancels, and there is no inflection in $\gamma(\omega)$. Note that for a single $e^{-\beta |\omega|}$ we would have $|\mathcal{C}(t)|\sim 1/t^2$, but after this cancellation we get $|\mathcal{C}(t)|\sim 1/t^4$. 

Our goal is now to use this example to compare the actual numerical errors to the theoretical bounds derived in this work. These bounds involve many relaxations and inequalities, so we do not expect them to be particularly tight. Also of interest is to compare the predicted optimal averaging time $T_{a,\text{theory opt}}$ to the numerically optimal $T_a$.
To be specific we set $a=1.01$, $b=0.6$, $\beta=4$, and find the normalization factor to be $\mathcal{N}\approx 21.0$. This choice of parameters 
gives rise to a maximum of $\gamma(\omega)$ 
at $\omega^* \approx 2.0$, and the peak is wide: for $0.15\leq\omega\leq 6.08, ~ \gamma(\omega) \geq 0.5 \gamma(\omega^*)$. This frequency range is where we will choose system transitions to lie; see Fig.~\ref{fig:specDen}. 

For this choice of parameters we obtain:
\begin{equation}
   \mathcal{C}(t) = \frac{325620 \mathcal{N}}{\tau_{SB}\pi(22i + 5t)(553i + 125t)(-106 - 515it + 625t^2)}\ ,
\end{equation}
so that indeed $\mathcal{C}(t)\sim 1/t^4$.
Integrals of $\mathcal{C}(t)e^{i\omega t}$ over various intervals can be expressed via the Meijer G-function. Its evaluation is a part of the numerical simulation of master equations discussed in this paper, so their computational cost is sensitive to any shortcuts we can find, and the analytic form above is helpful. We will comment on the specific runtimes below.

\begin{figure}[h]
     \centering
\subfigure[]{\includegraphics[width=0.46\textwidth]{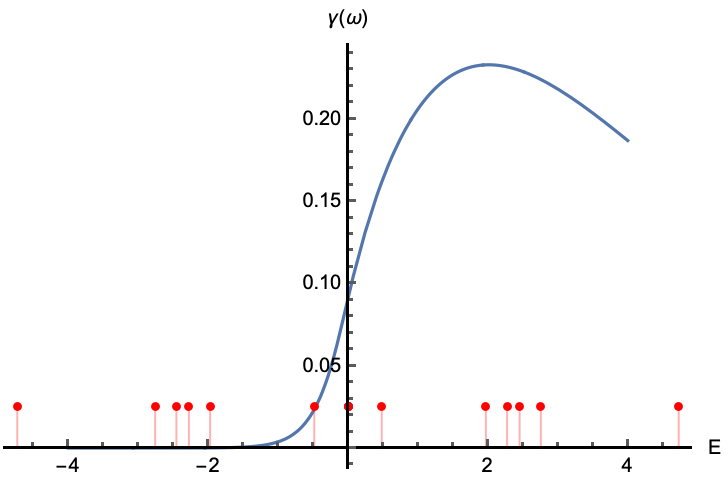}\label{fig:specDen}}
        \subfigure[]{\includegraphics[width=0.5\textwidth]{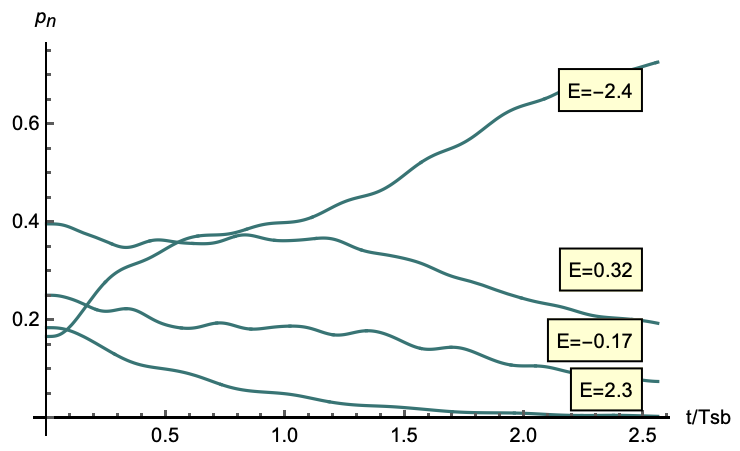}\label{RedsRel}}              
     \caption{(a) The spectral density $\gamma(\omega)$ for 
     $a=1.01$, $b=0.6$,
     $\beta=4$ and $\tau_{SB} =10.0$. The position of the transition frequencies of the Hamiltonian defined by Eq.~\eqref{eq:225} is shown by the red tickmarks. (b) Populations, i.e., the probabilities of eigenstates (labeled by energy in the figure) $p_n(t)$ for the solution $\rho_{OR}(t)$ in the region of its positivity: $0\leq t\leq 2.56\tau_{SB}$. }
\end{figure}

The bath parameter $\tau_{SB}$ can be chosen freely, while $\tau_B \approx 0.69 $ is determined by the choice of $a,b,\beta$. Note that it has dimensions of $\beta$. 

We choose $\tau_{SB} =10$ so that $\tau_B/\tau_{SB}=0.1 \ll 1$, a necessary condition for our error bounds to be small. We find $T_{a,\text{theory opt}} =\sqrt{\tau_B \tau_{SB}/5}  = 0.97$. To compare different equations, we pick an arbitrary two-qubit Hamiltonian:
\begin{equation}
    H= 0.5\sigma^z_1 -0.7 \sigma^z_2 + 0.3\sigma^z_1 \sigma^z_2 + \sigma^x_1 +\sigma^x_2\ ,
\label{eq:225}
\end{equation}
and choose $|11\rangle$ as the initial state. We let the coupling to the bath be given by
$V =\sigma_1^z \otimes B$. 

For all of our results, there will be no way to know what the true state $\rho_{\text{true}}(t)$ is. The differences found will be within the error bounds (growing at least linearly with time) of all equations. However we can simulate Eq.~\eqref{pureRedfield}, which is presumably closest to the true solution since only the Born and Markov approximations were made:
\begin{equation}
\label{origRedfield}
{\dot\rho}_{OR}(t) =-i[H,\rho_{OR}(t)] +\int_0^t \mathcal{C} (\tau-t) [A \rho_{OR}(t) A(\tau-t) - \rho_{OR}(t)  A(\tau-t) A] d\tau +\textrm{h.c.} 
\end{equation}
\begin{figure}[h]
     \centering
\subfigure[]{\includegraphics[width=0.40\textwidth]{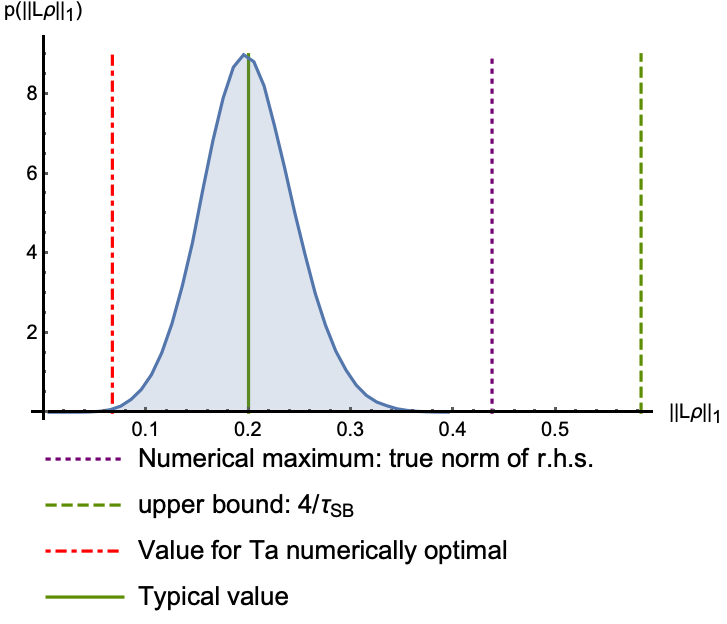}\label{rhno}}
%
        \subfigure[]{\includegraphics[width=0.56\textwidth]{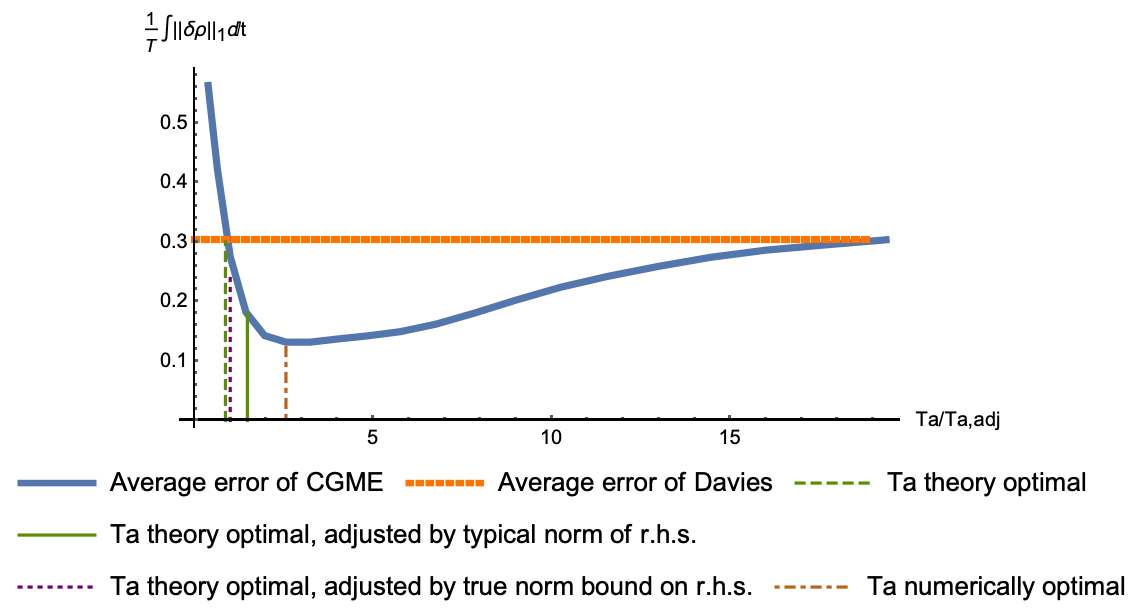}\label{pose}}      
     \caption{(a) Histogram of the norm of the right-hand side $\|\mathcal{L}^{BM,I}_t(\rho_{\text{test}})\|_1$ sampled uniformly over time and over normalized matrices from the Gaussian unitary ensemble. The upper bound $4/\tau_{SB}$ is shown as a dashed green line. Dotted: the numerically computed maximum of the norm. Solid vertical line: typical value of the norm. Dash-dotted: a value of the norm that would have explained the optimum, for comparison.   (b) Thick: the average error $ \frac{1}{T}\int_0^T \|\rho_C(t) - \rho_{OR }(t)\|_1 dt$. Thick dotted: same for Davies $ \frac{1}{T}\int_0^T \|\rho_D(t) - \rho_{OR }(t)\|_1 dt$. Dashed: position of $T_{a,\text{theory opt}}$, which does not lead to an advantage over Davies, in contrast to the following two. Dotted: $T_{a,\text{adj}}$. Dash-Dotted: $T_{a,\text{num. opt}}$.  Solid vertical line: $T_a$ adjusted by the typical value of the norm.}
\end{figure}
We call this the ``Original Redfield" equation (ORE). It is in fact possible that the CGME returns solutions closer to the true solution than this non-CP equation, but we note that in our derivation all the equations appeared as subsequent approximations to the ORE. Therefore we may characterize how good those approximations are by studying the difference $\rho_C -\rho_{OR}$. We first investigate  the solution $\rho_{OR}(t)$, and find that it becomes negative (and hence unphysical) at $t\approx 2.57\tau_{SB}$. Thus, we choose to compare the solutions of the CGME and ORE in the interval $0\leq t\leq 2.56\tau_{SB}$, where $\rho_{OR}\geq 0$ [see Fig.~\ref{RedsRel}, and hence we can assume that $c_{BM}=1$ [recall that this is always true for a CP master equation; see the discussion following Eq.~\eqref{eq:31c}]. We note that the computational cost of the ORE is the largest (by at least on order of magnitude) among our equations.%
\footnote{Precomputing the analytic form of the r.h.s. takes about $14$ minutes on a contemporary desktop computer, and the solution of the differential equation itself ~--- about a minute. In contrast, the CGME takes $20$ seconds total, and the Davies and Redfield equations ~--- even less. These numbers are for the implementation in Mathematica, where precomputing the symbolic form of the r.h.s. is possible. Faster implementations in other programming languages are certainly possible.
The Mathematica code for obtaining the plots of this section are available~\cite{urlCode}.
}

The second observation we make is about the norm of the r.h.s. of Eq.~\eqref{origRedfield}. In Fig.~\ref{rhno} we present a probability distribution of the norms of the right hand side $\|\mathcal{L}^{BM,I}_t(\rho_{\text{test}})\|_1$, where $\rho_{\text{test}}$ is taken from the Gaussian unitary ensemble (Hermitian matrices $X$ sampled from a probability distribution $\sim e^{-2^{n-1}\Tr X^2}$, where $n=2$ is the number of qubits), and subsequently normalized such that $\|\rho_{\text{test}}\|_1=1$. We also combine the data for different times, as if the time was sampled uniformly in $[0,2.56\tau_{SB}]$. We find numerically that the maximum value of the norm that can be achieved is $0.44\pm0.01$. It would suggest a higher effective $\tau_{SB}$ than our definition:
\begin{equation}
    \max_{t,\rho_{\text{test}}, \|\rho_{\text{test}}\|_1=1}\|\mathcal{L}^{BM,I}_t(\rho_{\text{test}})\|_1 = \frac{4}{\tau_{SB}^{\text{eff}}}\approx 0.44 < 0.58 \approx \frac{4}{\tau_{SB}}\ .
\end{equation}
Such an adjusted $\tau_{SB}^{\text{eff}} \approx 9.14 = 1.33 \tau_{SB}$ can be used to choose an adjusted $T_{a,\text{adj}}  =\sqrt{\tau_{SB}^{\text{eff}}/\tau_{SB}}T_{a,\text{theory opt}}  =1.12 =1.15 T_{a,\text{theory opt}}$. We could non-rigorously use the typical value of the norm instead, leading to $T_{a,\text{typ.}}\approx 1.7T_{a,\text{theory opt}}$. To see if this adjustment is justified,
we vary $T_a$ in Fig.~\ref{pose} and plot the average of the trace norm of the difference in the solutions over $t\in[0,2.56\tau_{SB}]$, i.e., $\frac{1}{T}\int_0^T \|\rho_C(t) - \rho_{OR }(t)\|_1 dt$.
We see that the true optimum is at even higher $T_{a,\text{num. opt}} \approx 2.87 \approx 3T_{a,\text{theory opt}} $, and the equation matches Davies in the limit $T_a \to \infty$ [Fig.~\ref{pose}], as formally expected (see Appendix~\ref{app:D}). Even though our adjustments of $T_a$ are not very large, the error is very sensitive to them and improves significantly.

\begin{figure}[h]
     \centering
         \subfigure[]{\includegraphics[width=0.56\textwidth]{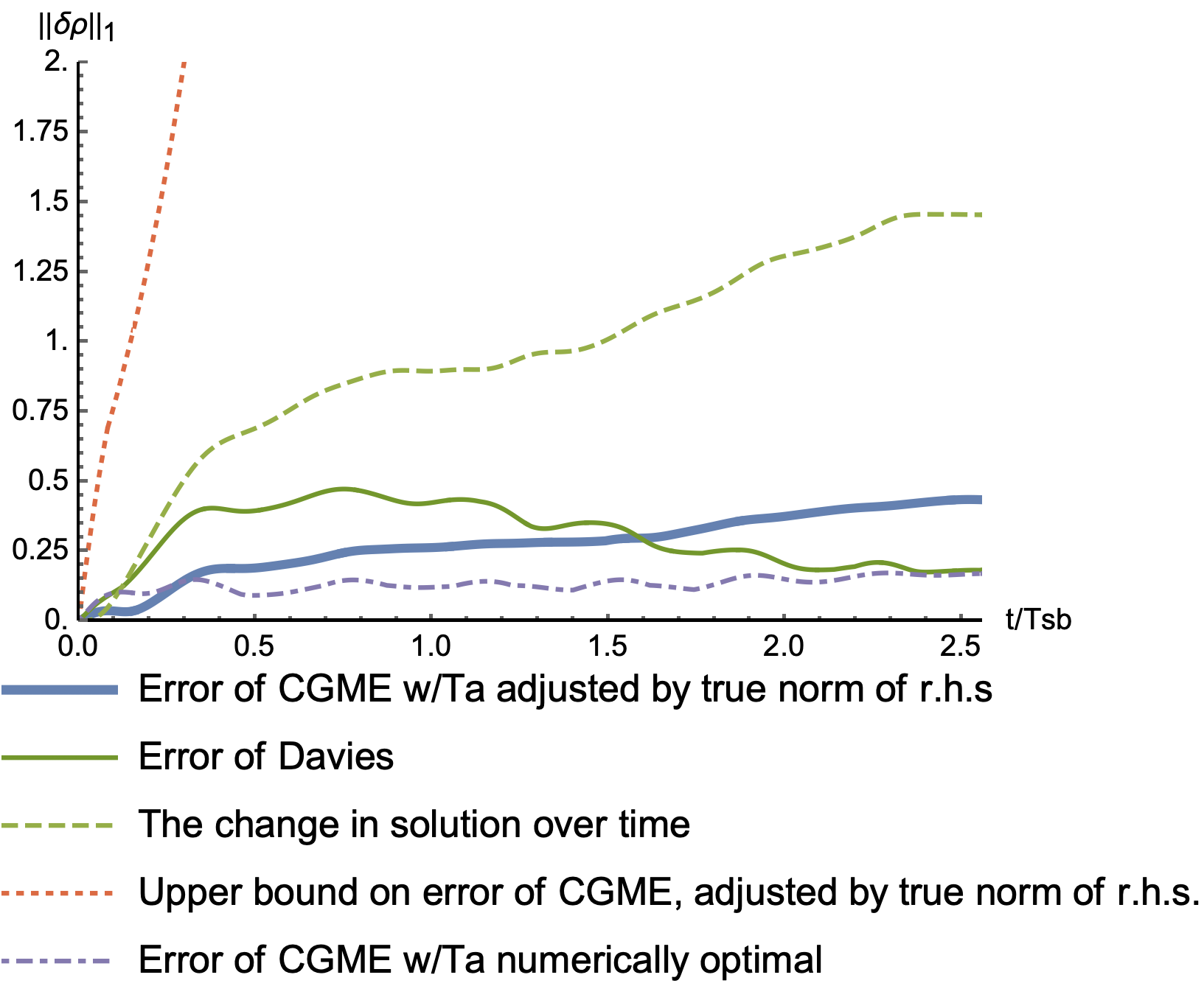}\label{earear}}
%
          \subfigure[]{\includegraphics[width=0.4\textwidth]{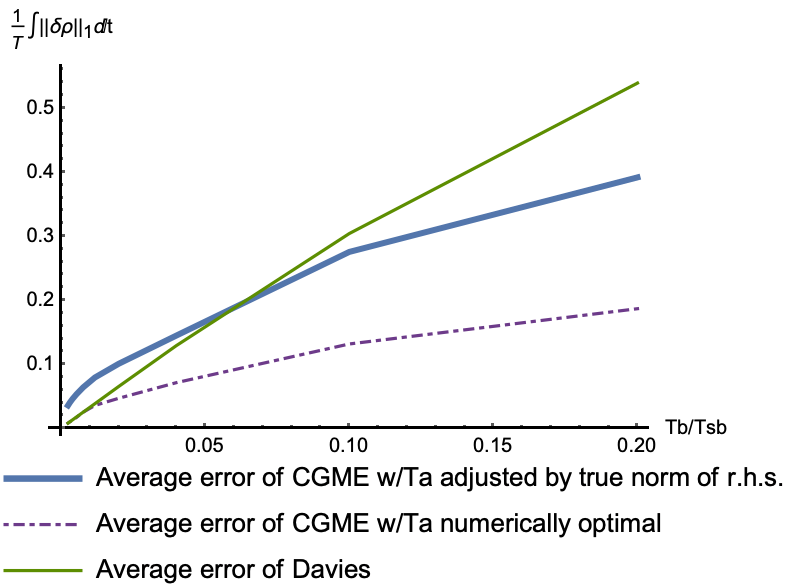}\label{firstSim}}
     \caption{(a) Thick: the error $  \|\rho_C(t) - \rho_{OR }(t)\|_1$ for $T_{a,\text{adj}}$. Thin: same for Davies $\|\rho_D(t) - \rho_{OR }(t)\|_1$. Dashed: The change in the solution $\|\rho_{OR,I }(t) -\rho_{OR,I }(0)\|_1$ induced by relaxation (our error should be small w.r.t. it). Dotted: the tightest upper bound derived in this paper [Eq.~(\ref{strongestBound})], bounding the thick line. Dash-Dotted: the error $  \|\rho_C(t) - \rho_{OR }(t)\|_1$ for $T_{a,\text{num. opt}}$. (b) The average error $ \frac{1}{T}\int_0^T \|\rho_C(t) - \rho_{OR }(t)\|_1 dt$ as a function of  varying $\tau_{SB}$, holding $\tau_B$ fixed, for $T_{a,\text{adj}}$ (thick) and $T_{a,\text{num. opt}}$ (dash-dotted). The thin line shows the average error $ \frac{1}{T}\int_0^T \|\rho_D(t) - \rho_{OR }(t)\|_1 dt$ for Davies.}
\end{figure}
We illustrate the error in more detail in Fig.~\ref{earear} by plotting the differences
\begin{equation}
    \|\rho_{C,\text{adj}}(t) - \rho_{OR }(t)\|_1, \quad\|\rho_{C,\text{num. opt}}(t) - \rho_{OR }(t)\|_1, \quad \|\rho_D(t) - \rho_{OR }(t)\|_1, \quad \|\rho_{OR,I }(t) -\rho_{OR,I }(0)\|_1 \ ,
\end{equation}
as a function of $t$, where the CGME is taken for $T_{a,\text{adj}}$ and $T_{a,\text{num. opt}}$. Also in Fig.~\ref{earear} we plot our strongest upper bound on $\|\rho_{C,\text{adj}}(t) - \rho_{OR }(t)\|_1$ [Eq.~(\ref{strongestBound}), derived in Sec. \ref{errSec} below]. We use $T_{a} \to T_{a,\text{adj}}$ and  $\Lambda \to 4/\tau_{SB}^{\text{eff}}$   Comparing this bound to the observed error in Fig.~\ref{earear}, we see that even after this adjustment, the bound is still not tight. Nonetheless, the adjustment helps to match the optimum in Fig.~\ref{pose} better, and we believe the bounds can be improved by future work in this direction.

In Fig.~\ref{firstSim} we vary $\tau_{SB}$, and evaluate 
\begin{equation}
    \frac{1}{T}\int_0^T \|\rho_{C,\text{adj}}(t) - \rho_{OR }(t)\|_1 dt, \quad    \frac{1}{T}\int_0^T \|\rho_{C,\text{num. opt}}(t) - \rho_{OR }(t)\|_1 dt, \quad \frac{1}{T}\int_0^T \|\rho_D(t) - \rho_{OR }(t)\|_1 dt \ ,\label{avErToplot}
\end{equation}
for $T =2.56\tau_{SB}$. To find $T_{a,\text{num. opt}}$,  we vary $T_a$ for each value of $\tau_{SB}$ and  choose  the one that gives the smallest values for the corresponding expression in Eq.~\eqref{avErToplot}. We again add the similar average error of Davies to the plot. Since our family of CGMEs with different $T_a$ includes Davies as a $T_a \to \infty$ limit, numerical optimization will always give an improvement over Davies. A simpler adjusted $T_a$ that does not require optimization still outperforms Davies at moderate $\tau_B/\tau_{SB}$. We also note that the $\sqrt{\tau_B/\tau_{SB}}$ dependence of its error at moderate $\tau_B/\tau_{SB}$ is consistent with the intuition from our bounds, while the linear portion suggests that our bounds are not tight in some parameter regime.

\subsection{Role of the Lamb shift in the CGME}
\label{sec:CGME-L}

We perform another test, where we drop the Lamb shift $H_{LS}$ of the CGME and obtain a corresponding solution $\rho_{CL}(t)$, as prescribed by Eqs.~\eqref{eq:CGME-L1} and~\eqref{eq:CGME-L2}. We then plot $\|\rho_{C,\text{adj}}(t) - \rho_{CL ,\text{adj}}(t)\|_1$ as a function of time, as shown in Fig.~\ref{lambBad}.
We again pick $T_{a,\text{adj}}= 1.15 T_{a,\text{opt}}$. We observe in Fig.~\ref{lambBad} that the Lamb shift has the same role as in the Davies equation: it affects the phases of the off-diagonal matrix elements of $\rho_C(t)$ in the eigenbasis. Since off-diagonal elements decay to zero for the Davies case, the effect of their phase on the norm difference first grows proportional to the phase, then disappears with the magnitude of those elements.

\begin{figure}[h]
     \centering
\includegraphics[width=0.4\textwidth]{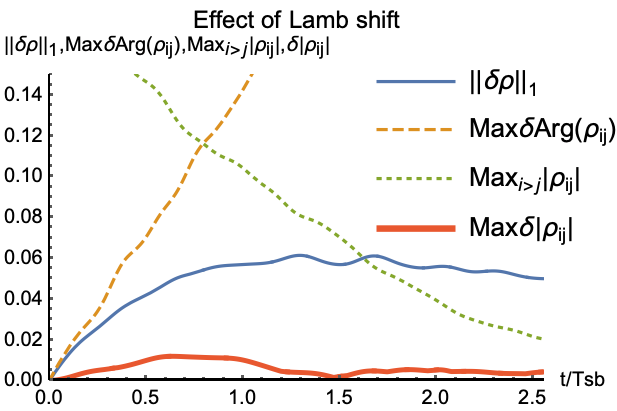}
     \caption{Thin: the error $\|\rho_C(t) - \rho_{CL }(t)\|_1$ of dropping the Lamb shift $H_{LS}$ for $T_{a,\text{adj}}$. Thick: the maximum difference in the magnitude of the matrix elements $\max_{ij}\delta|\rho_{C,ij}|$. Dashed: the maximum difference in the phase of the matrix elements. Dotted: the decay of the off-diagonal matrix elements.}
     \label{lambBad}
\end{figure}


\section{Error bounds and estimates}
\label{errSec}

In this section we give the details of the various error bounds and estimates mentioned without proof in our derivation of the CGME in Sec.~\ref{deriv1}. Our strategy is to derive the error terms successively, and then add them up using the triangle inequality. We work our way up through the various approximations, starting from Born in Sec.~\ref{Born}, followed by Markov in Sec.~\ref{sec:Markov-error}, then time-averaging (or coarse-graining) in Sec.~\ref{sec:t-ave-err2}. In the course of the latter, we bound the error due to dropping part of the integration domain to achieve complete positivity, and analyze the optimization of $T_a$. But we start with a proof of Lemma~\ref{factDif}, which plays a central role in our analysis.

\subsection{Bound on the difference of solutions of differential equations differing by a bounded operator}
\label{sec:t-ave-err}

We first present a proof of Lemma~\ref{factDif}.
\begin{proof}
The first part of the Lemma is concerned with the pair of differential equations $\dot{x} =\mathcal{L}_t x +\mathcal{E}$ and $\dot{y} = \mathcal{L}_t y$.
Their formal solution is:
\begin{equation}
    x(t) = x(0)+ \int_0^t \mathcal{L}_\tau (x(\tau)) d\tau + \int_0^t \mathcal{E}d\tau, \quad y(t) =y(0)+ \int_0^t \mathcal{L}_\tau (y(\tau)) d\tau\ .
    \eeq 
Let $\Lambda>0$ be a constant defined as follows: $\sup_{\tau,x}\|\mathcal{L}_\tau(x)\|_1\leq \Lambda$ for all $x$ such that $\|x\|_1 = 1$. Since $\mc{L}$ is linear, and writing $x=\delta/\|\delta\|_1$, it follows that $\forall\delta$: $\|\mc{L}_{\tau}(\delta)\|_1 = \|\delta\|_1 \|\mc{L}_{\tau}(x)\|_1 \leq \|\delta\|_1 \Lambda$. Then,
assuming $x(0) =y(0)$, and taking the operator norm of the difference yields, for $\delta = x(t)-y(t)$:
\begin{equation}
    \|x(t) -y(t)\|_1 \leq \Lambda \int_0^t  \|x(\tau) -y(\tau)\|_1 d\tau + t \|\mathcal{E}\|_1\ ,
\end{equation}
Many different functions $f(t)\equiv \|x(t) -y(t)\|_1$ satisfy this inequality for all $t$, but they are all upper-bounded by the function $b_1(t) = \max\|x(t) -y(t)\|_1$ that saturates the inequality:
\begin{equation}
    \dot{b}_1 = \Lambda b_1 +\|\mathcal{E}\|_1, \quad b_1(0) =0\ .
    \label{eq:107}
\end{equation}
The solution is
\bes
\begin{align}
    \label{eq:108a}    
    b_1(t) &= (e^{\Lambda t} -1)\frac{\|\mathcal{E}\|_1}{\Lambda }  \\
    &\leq (e^c-1)\frac{\|\mathcal{E}\|_1}{\Lambda} \quad \forall t\leq \frac{c}{ \Lambda} \ .
    \label{eq:108}
\end{align}
\ees
This proves the desired bound on $\|x(t) -y(t)\|_1 \leq b_1(t)$.

The second part of the Lemma is concerned with the more general, non-Markovian pair of differential equations $\dot{x}(t) =\int_0^tK_{t-\tau}(x(\tau))d\tau +\mathcal{E}$ and  $\dot{y}(t) = \int_0^tK_{t-\tau}(y(\tau))d\tau$. 
The first part of the Lemma in fact follows from this case as a corollary, but we presented it for clarity. 
We again subtract the formal solutions:
\begin{equation}
    x(t) - y(t) =\int_0^t\int_0^\tau K_{\tau-\theta}( x(\theta) - y(\theta))d\tau d\theta + \int_0^t\mathcal{E}d\tau \ .
\end{equation}
Define $\Lambda_t =\sup_{x}\|K_t(x)\|_1 $ for all $x$ such that $\|x\|_1 = 1$. By linearity of $K_t$ it follows as above that $\forall\delta$: $\|K_t(\delta)\|_1 \leq \Lambda_t \|\delta\|_1$. Therefore:
\begin{equation}
    \|x(t) - y(t)\|_1 \leq\int_0^t\int_0^\tau \Lambda_{\tau-\theta}\|x(\theta) - y(\theta)\|_1 d\tau d\theta + \int_0^t\|\mathcal{E}\|_1 d\tau \ .
\end{equation}
The upper bound $b_K(t)$ on $ \|x(t) - y(t)\|_1$ satisfies the equation:
\begin{equation}
    b_K(t) = \int_0^t\int_0^\tau \Lambda_{\tau-\theta}b_K(\theta)d\tau d\theta + \|\mathcal{E}\|_1 t\ .
\end{equation}
Taking the derivative once, we obtain:
\begin{equation}
    \dot{b}_K(t) = \int_0^t \Lambda_{t -\tau}b_K(\tau)d\tau  +  \|\mathcal{E}\|_1\ .
\end{equation}
Since the r.h.s. is positive, $b_K(t)$ is a monotonic function. So $b_K(\tau) \leq b_K(t)$, and it follows that:
\begin{equation}
        \dot{b}_K(t)\leq \int_0^t \Lambda_{t -\tau}b_K(t)d\tau  +  \|\mathcal{E}\|_1\ .
\end{equation}
In other words, we can again upper-bound $b_K(t)\leq b_K'(t)$, such that
\begin{equation}
    \dot{b}_K'(t)= \int_0^t \Lambda_{t -\tau}d\tau b_K'(t)  +  \|\mathcal{E}\|_1\ .
\end{equation}
Now let $\Lambda$ be a positive constant such that $\int_0^t \Lambda_\tau d\tau \leq \Lambda$ for all $t$. 
We thus arrive at 
\begin{equation}
    \dot{b_K}''(t)= \Lambda b_K''(t)  +  \|\mathcal{E}\|_1\ ,
\end{equation}
for a new upper-bound $b_K''(t)\geq b_K'(t)$. This is exactly the same as Eq.~\eqref{eq:107}, and hence the Lemma still holds for non-Markovian equations.
\end{proof}

Note that for our purposes $\Lambda = 4/\tau_{SB}$ is usually taken, thanks to Eq.~(\ref{eq:|L|}, \ref{newLambda}), and that the non-Markovian bound will be generalized in Sec.~\ref{Born}.
The exponential on the r.h.s. of  Eq.~\eqref{eq:108} may be troublesome. Under additional assumptions we can prove a much tighter linear in $t$ bound presented in Appendix \ref{LemmaTwo}, yet we did not find a way to recast all of our results in that tighter form.

\subsection{Born approximation error}
\label{Born}

We start with an estimate of the error associated with making the Born approximation, which we presented as the estimate $\|\mathcal{E}_B\|_1 = O\left(\frac{\tau_B}{\tau_{SB}^2}\right)$ in Eq.~\eqref{eq:Born21}. I.e., we would like to bound the error $\mathcal{E}_B = \dot\rho_{\text{true},I}(t) - \int_0^t K_{t-\tau}^{2,B}(\rho_{\text{true},I}(\tau)) d\tau$. It is difficult to do so directly since $\rho_{\text{true},I}(t)$ is unknown, so instead we will first use a proxy for $\rho_{\text{true},I}(t)$, which we call $\rho_{B4,I}(t)$. The latter is the solution of the master equation obtained by iterating the substitution of $\rho_{\text{tot}} =\rho_{\text{tot}}(0) - i \int_0^t [V(\tau),\rho_{\text{tot}}(\tau)]d\tau$ [Eq.~\eqref{eq:16b}]
into 
${\dot\rho}_{\text{tot}}(t) =-i[V(t),\rho_{\text{tot}}(t)]$ [Eq.~\eqref{eq:16a}] to 4th order, and only then applying the Born approximation. After developing an understanding of the 4th order we proceed to the infinite series. We make certain assumptions about its convergence, and comment below on their compatibility with the known convergence radius of the Dyson series. 
Let us first define the type of bath for which we can carry out this procedure in order to arrive at a bound.

We refer to a bath as exponential if its correlation function $\mathcal{C}(t)$ decays exponentially. This type of assumption is both common and convenient. For example, as we will see in more detail below, for a Gaussian bath (one which satisfies Wick's theorem) the times $t_1<t_2<t_3<t_4$ appear in the estimate of the Born error [they enter terms such as $\mathcal{C}(t_1-t_3)\mathcal{C}(t_2-t_4)$]. If $\mathcal{C}(t)$ is exponentially decaying with a characteristic time $\tau_B$, that immediately lets us conclude that only $t_1,t_2,t_3,t_4$ that are all within $\sim \tau_B$ of each other contribute to the estimate, and the bound on the Born error acquires a small factor of  $\tau_B/\tau_{SB}$. However, an exponentially decaying $\mathcal{C}(t)$ may be an unrealistically strong assumption (e.g., it is hard to satisfy at low temperatures). 
Therefore it is desirable to introduce a condition weaker than exponential decay, which we refer to as a non-exponential bath. Specifically, we consider the convergence of integrals such as
$\int_{0}^{\infty}|\mathcal{C}(t)| t^n dt$
for some integer $n\ge 0$. Such a condition was used in Ref.~\cite{ABLZ:12-SI} to control the error of the subsequent Markov approximation. Here we show how these integrals  for $n=0,1$ control our estimate of the Born error for a Gaussian bath. The physical meaning of these integrals was defined in Eq.~\eqref{eq:T1tauB}. A more general requirement on a non-Gaussian bath can be derived in the same way. However, that requirement will involve a higher-order correlation function that does not have a straightforward and concise interpretation.

Recall the steps of the derivation where we make the Born approximation, leading to Eq.~\eqref{eq:rho_Born}. 
For brevity, we drop the interaction picture subscript, replace the time argument by a subscript, and place the approximation level $B$ in the superscript: $\rho_{B,I}(t) = \rho_t^B$. Written out explicitly, Eq.~\eqref{eq:rho_Born} in terms of $\rho_t^B$ is:
\begin{equation}
{\dot\rho}^B_t =\int_0^t K_{t-\tau}^{2,B}(\rho_\tau^{B}) d\tau=\int_0^t \mathcal{C}(\tau-t) (A_t \rho^B_\tau A_\tau - \rho^B_\tau  A_\tau A_t) +\mathcal{C}(t-\tau) (A_\tau \rho^B_\tau  A_t -  A_t A_\tau\rho^B_\tau )d\tau \ .
\label{eq:Born84}
\end{equation}
The formal solution is:
\begin{equation}
{\rho}^B_\tau =\rho_0 +\int_0^{\tau}d\theta \int_0^\theta dm \mathcal{C}(m-\theta) (A_\theta \rho^B_m A_m - \rho^B_m  A_m A_\theta) +\mathcal{C}(\theta-m) (A_m \rho^B_m  A_\theta -  A_\theta A_m\rho^B_m ) .
\end{equation}
We substitute this solution back into Eq.~\eqref{eq:Born84}, a necessary step because our estimate of the Born error will come from the 4th order terms, so we need to subtract the original solution. The full expression to be subtracted is:
\bes
\label{eq:86}
\begin{align}
{\dot\rho}^B_t &=\int_0^t K_{t-\tau}^{2,B}\left(\int_0^\tau \int_0^\theta K_{\theta-m}^{2,B}(\rho_m^{B}) dm d\theta \right) d\tau=\int_0^t K_{t-m}^{4,B}(\rho_m^{B}) d\tau \\
&=\int_0^t \mathcal{C}(\tau-t) (A_t \rho_0 A_\tau - \rho_0  A_\tau A_t) +\mathcal{C}(t-\tau) (A_\tau \rho_0  A_t -  A_t A_\tau\rho_0 )d\tau
\label{firstLine}\\
&+\int_0^td\tau\int_0^{\tau}d\theta \int_0^\theta dm  \\
[ &\mathcal{C}(\tau-t)\mathcal{C}(m-\theta) A_t (A_\theta \rho^B_m A_m - \rho^B_m  A_m A_\theta)  A_\tau + \mathcal{C}(\tau-t)\mathcal{C}(\theta-m) A_t (A_m \rho^B_m  A_\theta -  A_\theta A_m\rho^B_m )  A_\tau \\
-&\mathcal{C}(\tau-t)\mathcal{C}(m-\theta) (A_\theta \rho^B_m A_m - \rho^B_m  A_m A_\theta)   A_\tau A_t - \mathcal{C}(\tau-t)\mathcal{C}(\theta-m) (A_m \rho^B_m  A_\theta -  A_\theta A_m\rho^B_m )  A_\tau A_t \\
+&\mathcal{C}(t-\tau) \mathcal{C}(m-\theta) A_\tau (A_\theta \rho^B_m A_m - \rho^B_m  A_m A_\theta) A_t+\mathcal{C}(t-\tau)\mathcal{C}(\theta-m)  A_\tau (A_m \rho^B_m  A_\theta -  A_\theta A_m\rho^B_m )  A_t \\
-& \mathcal{C}(t-\tau)\mathcal{C}(m-\theta) A_t A_\tau(A_\theta \rho^B_m A_m - \rho^B_m  A_m A_\theta)  - \mathcal{C}(t-\tau)\mathcal{C}(\theta-m) A_t A_\tau(A_m \rho^B_m  A_\theta -  A_\theta A_m\rho^B_m )]\ . \label{lastLine}
\end{align}
\ees

Now, these terms constitute only a part of the expression for the derivative of the true density matrix ${\dot\rho}^{\text{true}}_t$ to the $4$th order in the system-bath interaction. The full $4$th order expression can be obtained, as mentioned above, by substituting the solution of the total evolution, 
$ \rho_{\text{tot}} =\rho_{\text{tot}}(0) - i \int_0^t [V(\tau),\rho_{\text{tot}}(\tau)]d\tau$ 
into 
${\dot\rho}_{\text{tot}}(t) =-i[V(t),\rho_{\text{tot}}(t)]$ 
three consecutive times (to generate a 4th order commutator), and only then applying the Born approximation. We call this $\rho_t^{B4}$ since it will differ from $\rho_t^{B}$ in the 4th order in the interaction $V(t)$. Doing so, we obtain:
\bes
\begin{align}
{\dot\rho}^{B4}_t &=\int_0^t \mathcal{C}(\tau-t) (A_t \rho_0 A_\tau - \rho_0  A_\tau A_t) +\mathcal{C}(t-\tau) (A_\tau \rho_0  A_t -  A_t A_\tau\rho_0 )d\tau \\
&+\Tr_{\text{b}}[A_t\otimes B_t,\int_0^t d\tau [A_\tau\otimes B_\tau, \int_0^\tau d\theta [A_\theta \otimes B_\theta, \int_0^\theta dm [A_m \otimes B_m, \rho^{B4}_m\otimes \rho_{\text{b}}]]]] \ .
\end{align}
\ees
Defined in this way, ${\rho}^{B4}_t$ is close to the $4$th order of the true state ${\rho}^{\text{true}}_t$ up to terms of $6$th order and higher, so we use $\rho^{B4}_t - \rho^B_t$ to estimate the actual leading order error $\rho^{\text{true}}_t - \rho^B_t$. We proceed to expand all the commutators:
\bes
 \begin{align}
 \label{eq:B4}
&{\dot\rho}^{B4}_t =\int_0^t K_{t-\tau}^{4,B4}(\rho_\tau^{B4}) d\tau=\int_0^t K_{t-\tau}^{2,B}(\rho_0) d\tau \\
&+\int_0^t d\tau \int_0^\tau d\theta \int_0^\theta dm \Tr_{\text{b}} B_t B_\tau B_\theta B_m \rho_{\text{b}}  A_t A_\tau A_\theta A_m  \rho^{B4}_m
 - \Tr_{\text{b}} B_t B_\tau B_\theta \rho_{\text{b}}  B_m A_t A_\tau A_\theta   \rho^{B4}_m A_m - \dots 
 \end{align}
 \ees
 There are 16 terms total in the last line, and the order of $B_t,B_\tau, B_\theta, B_m$ and $\rho_{\text{b}}$ is exactly the same as the order of $A_t,A_\tau, A_\theta, A_m$ and $\rho^B_m$. 
 
 Let us now assume that the bath is Gaussian. By definition, a Gaussian bath obeys Wick's theorem (or Isserlis' theorem~\cite{Isserlis:1916}), which states that at all orders the higher order correlation functions decouple into products of two-point correlation functions. For example, for the 4-point correlation function:
  \begin{equation}
 \Tr_{\text{b}} B_t B_\tau B_\theta B_m \rho_{\text{b}}  =   \Tr_{\text{b}} B_t B_\tau  \rho_{\text{b}}    \Tr_{\text{b}} B_\theta B_m \rho_{\text{b}}   +  \Tr_{\text{b}} B_t B_\theta  \rho_{\text{b}}    \Tr_{\text{b}} B_\tau B_m \rho_{\text{b}} +\Tr_{\text{b}} B_t B_m \rho_{\text{b}}    \Tr_{\text{b}} B_\tau B_\theta \rho_{\text{b}}\ . \label{full}
 \end{equation}
 By definition, $\mathcal{C}(t-\tau ) =  \Tr_{\text{b}} B_t B_\tau  \rho_{\text{b}} $. Thus:
   \begin{equation}
 \Tr_{\text{b}} B_t B_\tau B_\theta B_m \rho_{\text{b}}  =  \mathcal{C}(t-\tau)   \mathcal{C}(\theta-m) 
   + \mathcal{C}(t-\theta)   \mathcal{C}(\tau-m) +\mathcal{C}(t-m)   \mathcal{C}(\tau-\theta)
 \end{equation}
 We see that in Eq.~\eqref{eq:86} only the first term was present. One can check that the order of $A$ is exactly the same, so upon subtracting Eq.~\eqref{eq:86} from Eq.~\eqref{full} only the two terms $  \mathcal{C}(t-\theta)  \mathcal{C}(\tau-m) +\mathcal{C}(t-m)   \mathcal{C}(\tau-\theta)$ remain. 
These extra terms can be interpreted in terms of the corresponding diagrammatic technique. Indeed, note that $t>\tau>\theta>m$. If in Wick's theorem we have a pairing $t\leftrightarrow\tau,~  \theta\leftrightarrow m$ and the arrows are drawn above the time axis, they do not relate to each other, and each arc can be anywhere along the time axis. But if the pairing is $t\leftrightarrow\theta,~  \tau\leftrightarrow m$ or $t\leftrightarrow m,~  \tau\leftrightarrow \theta$, then the two arcs are stuck together due to the condition $t>\tau>\theta>m$; see Fig.~\ref{pairIllustrated}.
 \begin{figure}[h]
\centering\includegraphics[width=0.5\linewidth]{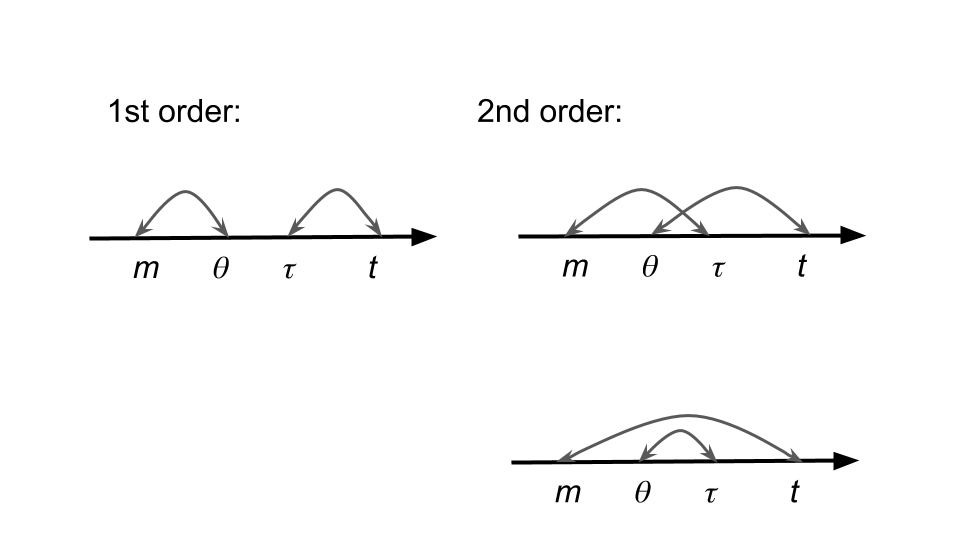}
\caption{Different possible pairings of $ \Tr_{\text{b}} B_t B_\tau B_\theta B_m \rho_{\text{b}}$ in Wick's theorem. Note that the first and second order are in $\tau_B/\tau_{SB}$, not in the coupling strength itself, in which both diagrams are 4th order.}
\label{pairIllustrated}
\end{figure}

At the moment we have two equations: the Born equation in two forms Eq.~\eqref{eq:Born84} and Eq.~\eqref{eq:86}, as well as the 4th order equation Eq.~\eqref{eq:B4}:
\begin{equation}
    {\dot\rho}^{B}_t = \int_0^t K_{t-\tau}^{2,B}(\rho_\tau^{B}) d\tau= \int_0^t K_{t-\tau}^{4,B}(\rho_\tau^{B}) d\tau, \quad {\dot\rho}^{B4}_t = \int_0^t K_{t-\tau}^{4,B4}(\rho_\tau^{B4}) d\tau,
    \label{eq:B+B4}
\end{equation}
where $K(\rho)$ are corresponding linear superoperators. $K^{2,B}$ is the 2nd order (in the bath interaction) operator from Eq.~\eqref{eq:Born84}, $K^{4,B}$ and $K^{4,B4}$ have terms up to the 4th order defined in Eq.~\eqref{eq:86} and Eq.~\eqref{eq:B4}. Define
\begin{align}
    \mathcal{E}_{B4}(\rho^{\text{any}}) &= \int_0^t (K_{t-\tau}^{4,B4}(\rho_\tau^{\text{any}}) - K_{t-\tau}^{4,B}(\rho_\tau^{\text{any}})) d\tau\ , 
\end{align}
for instance:
\beq
     \mathcal{E}_{B4}(\rho^{\text{B4}}) =  {\dot\rho}^{B4}_t - \int_0^t  K_{t-\tau}^{4,B}(\rho_\tau^{B4}) d\tau\ . 
     \eeq
Define
 \begin{align}
   \mathcal{E}_{B} &= {\dot\rho}^{\text{true}}_t - \int_0^t  K_{t-\tau}^{2,B}(\rho_\tau^{\text{true}}) d\tau \ . \label{eq:EbDef}
\end{align}
The quantity that appeared in the derivation is $\mathcal{E}_{B}$. We will now bound $\|\mathcal{E}_{B4}(\rho^{\text{any}})\|_1$ and discuss its generalization $\|\mathcal{E}_{Bk}(\rho^{\text{any}})\|_1$. We postpone the estimate for $\|\mathcal{E}_{B}\|_1$ till we develop more advanced tools. 

The difference $K_{t-\tau}^{4,B4}(\rho_\tau^{4,B4}) - K_{t-\tau}^{4,B}(\rho_\tau^{B4})$ is given by the two second order pairings in Fig.~\ref{pairIllustrated}, times the $16$ possible orders of $A_t,A_\tau, A_\theta, A_m$ and $\rho_m$:
\begin{equation}
\mathcal{E}_{B4} =\int_0^t d\tau \int_0^\tau d\theta \int_0^\theta dm (\mathcal{C}(t-\theta)   \mathcal{C}(\tau-m) +\mathcal{C}(t-m)   \mathcal{C}(\tau-\theta)) A_t A_\tau A_\theta A_m  \rho^{\text{any}}_m
  - \dots 
  \label{eq:92}
 \end{equation}
We take the trace norm of both sides:

 \begin{equation}
\|\mathcal{E}_{B4}(\rho^{\text{any}})\|_1 \leq \int_0^t d\tau \int_0^\tau d\theta \int_0^\theta dm \|(\mathcal{C}(t-\theta)   \mathcal{C}(\tau-m) +\mathcal{C}(t-m)   \mathcal{C}(\tau-\theta)) A_t A_\tau A_\theta A_m  \rho^{\text{any}}_m\|_1
  + \dots 
 \end{equation}
where we used the triangle inequality 
$\|A+B\|_1 \leq \|A\|_1 + \|B\|_1$.
We now use submultiplicativity [Eq.~\eqref{eq:submult}]. 
Noting that the first norm in the submultiplicativity inequality is the operator norm, we obtain:
 \begin{equation}
\|\mathcal{E}_{B4}(\rho^{\text{any}})\|_1 \leq \int_0^t d\tau \int_0^\tau d\theta \int_0^\theta dm \|(\mathcal{C}(t-\theta)   \mathcal{C}(\tau-m) +\mathcal{C}(t-m)   \mathcal{C}(\tau-\theta)) A_t A_\tau A_\theta A_m \| \|\rho^{\text{any}}_m\|_1
  + \dots 
 \end{equation}
The operator norm $\|A\|=1$ and we will set $c_a = \text{max}_t \|\rho^{\text{any}}_t\|_1$ such that $\|\rho^{\text{any}}\|_1 \leq c_a$:
 \begin{equation}
\|\mathcal{E}_{B4}(\rho^{\text{any}})\|_1 \leq c_a\int_0^t d\tau \int_0^\tau d\theta \int_0^\theta dm |(\mathcal{C}(t-\theta)   \mathcal{C}(\tau-m) +\mathcal{C}(t-m)   \mathcal{C}(\tau-\theta)) | 
  + \dots 
 \end{equation}
In particular, $\mathcal{E}_{B4}(\rho^B)$ is bounded with $c_B$ and $\mathcal{E}_{B4}(\rho^{B4})$ with $c_{B4}$. There are $16$ terms after the expansion of commutators:
   \begin{equation}
\|\mathcal{E}_{B4}(\rho^{\text{any}})\| \leq 16c_{a}\int_0^t d\tau \int_0^\tau d\theta \int_0^\theta dm |\mathcal{C}(t-\theta)|   |\mathcal{C}(\tau-m)| +|\mathcal{C}(t-m)|   |\mathcal{C}(\tau-\theta)| \ ,\label{ineq:Eb4}
 \end{equation}
where we have also used that $|\mathcal{C}(t-m)| = |\mathcal{C}(m-t)^*| = |\mathcal{C}(m-t)|$. What is left is to bound the integrals in Eq.~\eqref{eq:92}. It turns out that the following condition suffices:
\begin{equation}
\int_{0}^{\infty}|\mathcal{C}(t)| t^n dt \leq \frac{\tau_B^{n}}{\tau_{SB}} \ , \quad \text{for $n=0,1$}\ ,
\label{integrobound}
\end{equation}
which is automatically satisfied for $\tau_{SB}$ and $\tau_B$ as defined in Eq.~\eqref{eq:T1tauB}. We also change the order of integration variables using the condition $0\leq m\leq\theta\leq\tau\leq t$, e.g.:
\begin{equation}
    \int_0^t d\tau \int_0^\tau d\theta = \int_0^t d\theta \int_\theta^t d\tau .
\end{equation}
The first term in Eq.~\eqref{eq:92} is bounded as:
\bes
\begin{align}
\int_0^t d\tau \int_0^\tau d\theta \left[\int_0^\theta dm |\mathcal{C}(\tau-m)|\right] |\mathcal{C}(t-\theta)|  &\leq \frac{1}{\tau_{SB}} \int_0^t d\tau \int_0^\tau d\theta|\mathcal{C}(t-\theta)| 
= \frac{1}{\tau_{SB}} \int_0^t d\theta \int_\theta^t d\tau|\mathcal{C}(t-\theta)| \\
&= \frac{1}{\tau_{SB}} \int_0^t (t-\theta)d\theta |\mathcal{C}(t-\theta)| \leq \frac{\tau_B}{\tau_{SB}^2} ,
\end{align}
\ees
while the second term in Eq.~\eqref{eq:92} is bounded as:
\bes
\begin{align}
   \int_0^t d\tau \int_0^\tau d\theta \int_0^\theta dm |\mathcal{C}(t-m)|   |\mathcal{C}(\tau-\theta)| &= \int_0^t dm \int_m^t d\theta \left[\int_\theta^t d\tau |\mathcal{C}(\tau-\theta)|\right] |\mathcal{C}(t-m)|  \\ 
   &\leq  \frac{1}{\tau_{SB}} \int_0^t dm \int_m^t d\theta|\mathcal{C}(t-m)|  = \frac{1}{\tau_{SB}} \int_0^t (t-m)dm |\mathcal{C}(t-m)| \\
   &\leq \frac{\tau_B}{\tau_{SB}^2} \ .
\end{align}
\ees
Together, this yields:
\begin{equation}
\|\mathcal{E}_{B4}(\rho^{\text{any}})\|_1 \leq 32c_a\frac{\tau_B}{\tau_{SB}^2}\ ,
\end{equation}
which is the first hint for the result given in Eq.~\eqref{eq:Born21}.
The actual error term appearing in Eq.~\eqref{eq:Born20} is $ \mathcal{E}_B$ as defined in Eq.~\eqref{eq:EbDef}. We can consider repeating the construction of $\rho^{B4}$ presented here for $\rho^{B6}, \rho^{B8} \dots$ as successive approximations to $\rho^{\text{true}}$. The corresponding errors can be defined:
\begin{equation}
     \mathcal{E}_{Bk}(\rho^{\text{any}}) = \int_0^t (K_{t-\tau}^{k,Bk}(\rho_\tau^{\text{any}}) - K_{t-\tau}^{k,B}(\rho_\tau^{\text{any}})) d\tau\ , 
\end{equation}
By employing the diagrammatic technique as in Fig.~\ref{pairIllustrated}, the higher orders of perturbation theory in the system-bath coupling can be shown to give higher orders in $\tau_B/\tau_{SB}$. There are also extra factors of $4t/\tau_{SB}$ that appear in the higher orders. For instance, $\mathcal{E}_{B6}$ has the form:
\begin{equation}
\|\mathcal{E}_{B6}(\rho^{\text{any}})\|_1 \leq 32c_a\frac{\tau_B}{\tau_{SB}^2}\left(1+ \frac{8t}{\tau_{SB}}+ O\left(\frac{\tau_B}{\tau_{SB}}\right)\right) ,
\end{equation}
and the diagrams for the second term are shown in Fig.~\ref{EBkfig}(a). We do not prove this form of $\mathcal{E}_{B6}$ here. We will bound all $t$-dependent terms of the first order in $\tau_B/\tau_{SB}$, by summing the diagrams in Fig.~\ref{EBkfig}(b):
\begin{figure}[t]
\includegraphics[width=0.8\linewidth]{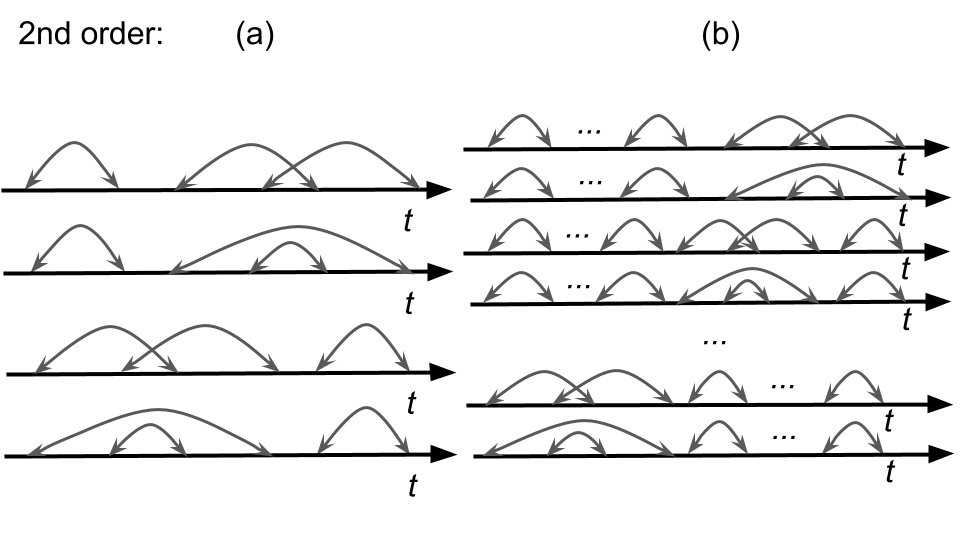}
\caption{ Diagrams contributing to the second order of (a) $\mathcal{E}_{B6}$, (b) $\mathcal{E}_{Bk}$ for arbitrary $k$.}
\label{EBkfig}
\end{figure}
\bes
\begin{align}
    \|\mathcal{E}_{Bk}(\rho^{\text{any}})\|_1 &\leq 32c_a\frac{\tau_B}{\tau_{SB}^2}\left(\sum_{m=1}^{k/2-1} \frac{m}{(m-1)!}\left(\frac{4t}{\tau_{SB}}\right)^{m-1}+ O\left(\frac{\tau_B}{\tau_{SB}}\right)\right) \\
    &\leq 32c_a\frac{\tau_B}{\tau_{SB}^2}\left((1+4 t/\tau_{SB})e^{4 t/\tau_{SB}}+ O\left(\frac{\tau_B}{\tau_{SB}}\right)\right)\ ,
\end{align}
\ees
for all $k$. Here  we fix time to be a constant: $\exists \epsilon(t), C(t)$ s.t. $O\left(\frac{\tau_B}{\tau_{SB}}\right) \leq C(t)\frac{\tau_B}{\tau_{SB}}$ for all $\frac{\tau_B}{\tau_{SB}}\leq \epsilon(t)$. This is a weak statement.  A stronger statement would include the time-dependence inside big-$O$, and have time-independent constants $\epsilon, C$: 
\begin{equation}
    \|\mathcal{E}_{Bk}(\rho^{\text{any}})\|_1 \leq 32c_a e^{4t/\tau_{SB}}\frac{\tau_B}{\tau_{SB}^2}  (1+4 t/\tau_{SB}+O(e^{4 t/\tau_{SB}}\tau_B/\tau_{SB}))
\end{equation}
We do not know how to prove such a statement as of yet. 
The proof would rely on the diagrammatic calculation of subleading terms. It is straightforward, but cumbersome, and we do not attempt it here. Instead we make two \emph{assumptions}:
\begin{enumerate}
    \item In the $k\to \infty$ limit, there is a finite radius of convergence in 
    $e^{4t/\tau_{SB}}\tau_B/\tau_{SB}$:
    \begin{equation}
         \|\mathcal{E}_{B\infty}(\rho^{\text{any}})\|_1 \leq 32c_a e^{4t/\tau_{SB}}\frac{\tau_B}{\tau_{SB}^2}  (1+4t/\tau_{SB}+O(e^{4t/\tau_{SB}}\tau_B/\tau_{SB}))\ .
         \label{eq:128}
    \end{equation}
    \item The true solution is the limit of the perturbative series:
    \begin{equation}
{\dot\rho}^{\text{true}}_t = \int_0^t  K_{t-\tau}^{\infty,B\infty}(\rho_\tau^{\text{true}}) d\tau\ .
    \end{equation}
\end{enumerate}

To prove convergence, one needs to control the factorial number of terms appearing in  Wick's theorem. Such a proof was already presented by Davies~\cite{Davies:74}. The slight adjustments needed to include an integral condition such as Eq.~\eqref{integrobound} are a somewhat tedious problem left for a future study. 

Our definition of $\mathcal{E}_{B}$ under this assumption contains $\mathcal{E}_{B\infty}$ discussed above, as well as an extra term:
\bes
\begin{align}
     \mathcal{E}_{B} &= {\dot\rho}^{\text{true}}_t - \int_0^t  K_{t-\tau}^{2,B}(\rho_\tau^{\text{true}}) d\tau  = \int_0^t  (K_{t-\tau}^{\infty,B\infty}(\rho_\tau^{\text{true}}) - K_{t-\tau}^{2,B}(\rho_\tau^{\text{true}}))  d\tau \\
    \| \mathcal{E}_{B}\|_1 &\leq\left\|\int_0^t  (K_{t-\tau}^{\infty,B\infty}(\rho_\tau^{\text{true}}) - K_{t-\tau}^{\infty,B}(\rho_\tau^{\text{true}}))  d\tau \right\|_1 + \left\|\int_0^t  (K_{t-\tau}^{\infty,B}(\rho_\tau^{\text{true}}) - K_{t-\tau}^{2,B}(\rho_\tau^{\text{true}}))  d\tau \right\|_1 \\ 
    &\leq 
    32\frac{\tau_B}{\tau_{SB}^2}  (1+O(\tau_B/\tau_{SB})) + \left\|\int_0^t  (K_{t-\tau}^{\infty,B}(\rho_\tau^{\text{true}}) - K_{t-\tau}^{2,B}(\rho_\tau^{\text{true}}))  d\tau \right\|_1\ ,
\end{align}
\ees
where $c_a=1$ since the true evolution is CP. The second term will turn out to be time-dependent, and quite tricky to estimate. We will return to this problem after we develop some extra machinery.


Let us first note the form of $K^{k,B}$, starting with $k=4$:

\begin{equation}
   \int_0^t K^{4,B}_{t-\tau}(\rho^{\text{any}}_\tau) d\tau =\int_0^t K^{2,B}_{t-\tau}\left(\rho_0 + \int_0^\tau\int_0^m K^{2,B}_{m-\theta}(\rho^{\text{any}}_m)dm d\theta\right)d\tau\ .
\end{equation}
The solution of
\begin{equation}
    \dot{\rho}^{B4} =\int_0^t K^{4,B4}_{t-\tau}(\rho_\tau^{B4}) d\tau= \int_0^t K^{4,B}_{t-\tau}(\rho^{B4}_\tau) d\tau + \mathcal{E}_{B4}
\end{equation}
is equivalent to solving a system of $4$th order equations that are our proxy for the true solution (which is $\infty$ order):
\begin{equation}
    \left\{ \begin{array}{ll}
         \dot{\rho}^{B4} = \int_0^t K^{2,B}_{t-\tau}(\rho^{(1)}_\tau) d\tau  + \mathcal{E}_{B4},& \rho^{B4}(0) = \rho_0  \\
        \dot{\rho}^{(1)} = \int_0^t K^{2,B}_{t-\tau}(\rho^{B4}_\tau) d\tau, & \rho^{(1)}(0) = \rho_0 \ .\end{array} \right.
\end{equation}
We will prove that the solution of the above system remains close to $\rho^{B}$, the solution of the equation after the Born approximation:
\begin{equation}
    \dot{\rho}^{B} = \int_0^t K^{2,B}_{t-\tau}(\rho^{B}_\tau) d\tau\ .
\end{equation}
We will also present a straightforward generalization to all $k$ which will allow one to bound $\|\mathcal{E}_B\|_1$. 

First write the equation in the compact form:
\begin{equation}
    \dot{\bf{p}} = \bf{Lp} + \bf{E}, \quad \bf{p} = \left(~^{\rho^{B4}}_{\rho^{(1)}}\right), \quad \bf{L} =  \left(~^{0~ \quad\quad ~ \int K_{t-\tau}d\tau}_{\int K_{t-\tau}d\tau ~ \quad\quad ~0}\right), \quad \bf{E} =  \left(~^{ \mathcal{E}_{B4}(\rho^{B4})}_{0 }\right)\ .
\end{equation}
Define a custom norm that is an infinity norm on the vector of operators, but for each operator the trace norm is taken:
\begin{equation}
    \|{\bf{v}}\|_c =\left\| ~^{v^1}_{v^2}\right\|_c = \text{max}(\|v^1\|_1, \|v^2\|_1)
\end{equation}
The triangle inequality is obeyed as required:
\bes
\begin{align}
    \|{\bf{v} + \bf{w}}\|_c  &= \text{max}(\|v^1 + w^1\|_1, \|v^2 +w^2\|_1) \leq \text{max}(\|v^1\|_1 + \|w^1\|_1, \|v^2\|_1 +\|w^2\|_1)  \\ 
    &\leq \text{max}(\|v^1\|_1, \|v^2\|_1) +\text{max}( \|w^1\|_1, \|w^2\|_1) = \|{\bf{v}}\|_c + \|{\bf{w}}\|_c\ .
\end{align}
\ees
The maximum of the custom norm of ${\bf{L}}$ over inputs of norm 1 is the same as before:
\begin{equation}
    \text{max}_{{\bf{p}}: \|{\bf{p}}\|_c =1}\|{\bf{Lp}}\|_c = \Lambda = \frac{4}{\tau_{SB}} \ ,
\end{equation}
so all the properties needed to prove Lemma~\ref{factDif} hold. We establish:
\begin{equation}
    \|\delta{\bf{p}}_t\|_c \leq \frac{e^{\Lambda t}-1}{\Lambda}\|{\bf{E}}\|_c
\end{equation}
in the original variables:
\begin{equation}
    \text{max}(\|\delta\rho^{(1)}_t\|_1,  \|\delta\rho^{B4}_t\|_1) \leq \frac{e^{\Lambda t}-1}{\Lambda} \|\mathcal{E}_{B4}  \|_1\ .
\end{equation}
Note that the error we are trying to bound is
\begin{equation}
   \| \mathcal{E}_B\|_1 =\left\|\dot{\rho}^{\text{true}} - \int_0^t K^{2,B}_{t-\tau}(\rho^{\text{true}}_\tau) d\tau \right\|_1 = \left\| \int_0^t \left(K^{\infty,B\infty}_{t-\tau}(\rho^{\text{true}}_\tau)-K^{2,B}_{t-\tau}(\rho^{\text{true}}_\tau\right) d\tau \right\|_1\ .
\end{equation}
Our approximation to it is bounded as:
\begin{equation}
    \left\|\dot{\rho}^{B4} - \int_0^t K^{2,B}_{t-\tau}(\rho^{B4}_\tau) d\tau \right\|_1 \leq  \left\| \int_0^t K^{2,B}_{t-\tau}(\rho^{(1)}_\tau -\rho^{B4}_\tau  ) d\tau \right\|_1 + \|\mathcal{E}_{B4}\|_1 \leq (2e^{\Lambda t}-1) \|\mathcal{E}_{B4}  \|_1 \ ,
\end{equation}
where we have used the triangle inequality: $\|\rho^{(1)}_\tau -\rho^{B4}_\tau   \|_1 \leq \|\delta\rho^{(1)}_\tau\|_1 + \|\delta\rho^{B4}_\tau   \|_1$. Now we can repeat the above proof for arbitrary $k$:
\begin{equation}
    \left\{ \begin{array}{ll}
         \dot{\rho}^{Bk} = \int_0^t K^{2,B}_{t-\tau}(\rho^{(1)}_\tau) d\tau  + \mathcal{E}_{Bk}(\rho^{Bk})& \rho^{Bk}(0) = \rho_0 \\
        \dot{\rho}^{(1)} = \int_0^t K^{2,B}_{t-\tau}(\rho^{(2)}_\tau) d\tau & \rho^{(1)}(0) = \rho_0\\ \dots & \dots\\
        \dot{\rho}^{(k/2-1)} = \int_0^t K^{2,B}_{t-\tau}(\rho^{B4}_\tau) d\tau & \rho^{(k/2-1)}(0) = \rho_0\end{array} \right.
\end{equation}
For appropriately defined ${\bf{L}}$, $\| \cdot \|_c$ the r.h.s. is still bounded by the same constant:
\begin{equation}
    \text{max}_{{\bf{p}}: \|{\bf{p}}\|_c =1}\|{\bf{Lp}}\|_c = \Lambda = \frac{4}{\tau_{SB}}\ ,
\end{equation}
which leads to
\begin{equation}
    \text{max}(\|\delta\rho^{(1)}_t\|_1, \dots, \|\delta\rho^{Bk}_t\|_1) \leq \frac{e^{\Lambda t}-1}{\Lambda} \|\mathcal{E}_{Bk}  \|_1\ ,
\end{equation}
and the error in the r.h.s. is bounded without change:
\begin{equation}
    \left\|\dot{\rho}^{Bk} - \int_0^t K^{2,B}_{t-\tau}(\rho^{Bk}_\tau) d\tau \right\|_1 \leq  \left\| \int_0^t K^{2,B}_{t-\tau}(\rho^{(1)}_\tau -\rho^{Bk}_\tau  ) d\tau \right\|_1 + \|\mathcal{E}_{Bk}\|_1 \leq (2e^{\Lambda t}-1) \|\mathcal{E}_{Bk}  \|_1 \ .
\end{equation}
Taking the $k\to \infty$ limit, we obtain:
\begin{equation}
     \|\rho^{\text{true}}_t-\rho^{B}_t\|_1 \leq \frac{e^{\Lambda t}-1}{\Lambda} \|\mathcal{E}_{B\infty} (\rho^{\text{true}}) \|_1 \leq  8(e^{4t/\tau_{SB}}-1) e^{4t/\tau_{SB}}\frac{\tau_B}{\tau_{SB}}  (1+4t/\tau_{SB}+O(e^{4t/\tau_{SB}}\tau_B/\tau_{SB})) \ ,
     \label{BorErrFirst}
\end{equation}
and, finally: 
\bes
\begin{align}
    \|\mathcal{E}_B\|_1 &\leq (2e^{\Lambda t}-1) \|\mathcal{E}_{B\infty} (\rho^{\text{true}})\|_1  \\
    &\leq 32(2e^{4t/\tau_{SB}}-1) e^{4t/\tau_{SB}}\frac{\tau_B}{\tau_{SB}^2}  (1+4t/\tau_{SB}+O(e^{4t/\tau_{SB}}\tau_B/\tau_{SB}))\ ,
    \label{eq:193b}
\end{align}
\ees
where we have used that $c_a =1$ for $\rho^{\text{true}}$. This result is the sought-after justification of the estimate $\|\mathcal{E}_B\|_1 = O\left(\frac{\tau_B}{\tau_{SB}^2}\right)$ given in Eq.~\eqref{eq:Born21}.

An immediate consequence of Eq.~\eqref{BorErrFirst} is that we can use it to obtain a bound on $c_B$, which we defined after Eq.~\eqref{eq:|L|B} as an upper bound on $\|\rho_{B,I}\|_1$. Since
\begin{equation}
   \|\rho^{B}_t\|_1 
   \leq \|\rho^{\text{true}}_t\|_1 + \|\rho^{B}_t-\rho^{\text{true}}_t\|_1  = 1 +\|\rho^{\text{true}}_t-\rho^{B}_t\|_1 \ ,
\end{equation}
we obtain:
\begin{equation}
    c_B \leq 1 + 8(e^{4t/\tau_{SB}}-1) e^{4t/\tau_{SB}}\frac{\tau_B}{\tau_{SB}}  (1+4t/\tau_{SB}+O(e^{4t/\tau_{SB}}\tau_B/\tau_{SB})) \ .
    \label{cBbound}
\end{equation}


\subsection{Markov approximation error}
\label{sec:Markov-error}

This subsection follows Ref.~\cite{ALMZ:12} initially. Before the Markov approximation but after the Born approximation, the master equation had the form given in Eq.~\eqref{eq:rho_Born}.
The error due the Markov approximation $\rho_{B,I}(\tau) \mapsto \rho_{B,I}(t)$, in transitioning to Eq.~\eqref{RedStart}, is:
\begin{equation}
  \mathcal{E}_M =   \int_0^t \mathcal{C}(\tau-t) [A(t) (\rho_{B,I}(\tau) - \rho_{B,I}(t)) A(\tau) - (\rho_{B,I}(\tau) - \rho_{B,I}(t)) A(\tau) A(t)] d\tau +\textrm{h.c.}
\end{equation}
Since $\rho_{B,I}(t) - \rho_{B,I}(\tau) = \int_\tau^t {\dot\rho}_{B,I}(t') dt'$, we have
\beq 
\|\rho_{B,I}(t) - \rho_{B,I}(\tau)\|_1 \leq \|{\dot\rho}_{B,I}\|_1(t-\tau) \leq\frac{4(t-\tau)c_B}{\tau_{SB}}\  ,
\end{equation}
where we used Eq.~\eqref{eq:|L|B} in the second inequality.
Now we can estimate the norm of the error:
\begin{equation}
   \| \mathcal{E}_M\|_1 \le 4\int_0^t |\mathcal{C}(\tau-t)| \frac{4(t-\tau)}{\tau_{SB}}   d\tau  \leq 16 \frac{\tau_B c_B}{\tau_{SB}^2}\ .
   \label{eq:E_M}
\end{equation}
We found $c_B$ and saw that it is $O(1)$ in Eq.~\eqref{cBbound}.

Let us now use the Markovian Lemma  \ref{factDif} with Eq.~\eqref{eq:E_M}:
\begin{equation}
    \|\rho^{BM}_t-\rho^{B}_t\|_1  \leq 4 \frac{\tau_B  c_{B} }{\tau_{SB}} \left(e^{4t/\tau_{SB}}-1\right)\ . \label{cBMarkov}
\end{equation} 
A total expression for the Born and Markov error is [adding Eq.~\eqref{BorErrFirst}]: 
\begin{equation}
     \|\rho^{\text{true}}_t-\rho^{BM}_t\|_1  \leq 4 \frac{\tau_B  c_{B} }{\tau_{SB}} \left(e^{4t/\tau_{SB}}-1\right)
 + 8(e^{4t/\tau_{SB}}-1) e^{4t/\tau_{SB}}\frac{\tau_B}{\tau_{SB}}  (1+4t/\tau_{SB}+O(e^{4t/\tau_{SB}}\tau_B/\tau_{SB}))\ ,
\end{equation}
which simplifies to:
\begin{equation}
     \|\rho^{\text{true}}_t-\rho^{BM}_t\|_1  \leq 4(e^{4t/\tau_{SB}}-1)\frac{\tau_B}{\tau_{SB}} \left( c_{B}+ 2 e^{4t/\tau_{SB}} (1+4t/\tau_{SB}+O(e^{4t/\tau_{SB}}\tau_B/\tau_{SB}))\right)\ , \label{BMerrorFirst}
\end{equation}
or substituting in $c_B$:
\begin{equation}
     \|\rho^{\text{true}}_t-\rho^{BM}_t\|_1  \leq \frac{\left(e^{\frac{4t}{\tau_{SB}}} -1\right)\tau_B \left(4 + 8\left(1+\frac{4\left(e^{\frac{4t}{\tau_{SB}}} -1\right)\tau_B}{\tau_{SB}}\right) e^{\frac{4t}{\tau_{SB}}} \left(1+\frac{4t}{\tau_{SB}}+O\left(e^{\frac{4t}{\tau_{SB}}}\frac{\tau_B}{\tau_{SB}}\right)\right)\right)}{\tau_{SB}}  \ .
\end{equation}
The extra factor from $c_{B}$ can be absorbed into the big-$O$ notation:
\begin{equation}
     \|\rho^{\text{true}}_t-\rho^{BM}_t\|_1  \leq \frac{\left(e^{\frac{4t}{\tau_{SB}}} -1\right)\tau_B \left(4 + 8 e^{\frac{4t}{\tau_{SB}}} \left(1+\frac{4t}{\tau_{SB}}\right)\left(1+O\left(e^{\frac{4t}{\tau_{SB}}}\frac{\tau_B}{\tau_{SB}}\right)\right)\right)}{\tau_{SB}}  \ .
\end{equation}
A loose bound follows from the above:
\begin{equation}
     \|\rho^{\text{true}}_t-\rho^{BM}_t\|_1  \leq \frac{12\left(e^{\frac{4t}{\tau_{SB}}} -1\right)e^{\frac{8t}{\tau_{SB}}}\tau_B \left(1+O\left(e^{\frac{4t}{\tau_{SB}}}\frac{\tau_B}{\tau_{SB}}\right)\right)}{\tau_{SB}} \ . \label{BMloose}
\end{equation}
Sometimes we wish to report the bound briefly as:
\begin{equation}
     \|\rho^{\text{true}}_t-\rho^{BM}_t\|_1  = \left(e^{\frac{4t}{\tau_{SB}}} -1\right)O\left(e^{\frac{8t}{\tau_{SB}}}\frac{\tau_B}{\tau_{SB}}\right) \label{bigOexplain}
\end{equation}
\begin{proof}
The original big-O notation states that for $e^{\frac{4t}{\tau_{SB}}}\frac{\tau_B}{\tau_{SB}} \leq \epsilon$ the following inequality holds:
\begin{equation}
     \|\rho^{\text{true}}_t-\rho^{BM}_t\|_1  \leq \frac{12\left(e^{\frac{4t}{\tau_{SB}}} -1\right)e^{\frac{8t}{\tau_{SB}}}\tau_B \left(1+Ce^{\frac{4t}{\tau_{SB}}}\frac{\tau_B}{\tau_{SB}}\right)}{\tau_{SB}} \leq  \frac{12\left(e^{\frac{4t}{\tau_{SB}}} -1\right)e^{\frac{8t}{\tau_{SB}}}\tau_B \left(1+C\epsilon\right)}{\tau_{SB}}\ . \label{BornObtained}
\end{equation}
We first note that $e^{\frac{4t}{\tau_{SB}}} \leq e^{\frac{8t}{\tau_{SB}}}$, so if $e^{\frac{8t}{\tau_{SB}}}\frac{\tau_B}{\tau_{SB}}  \leq \epsilon$, then  $e^{\frac{4t}{\tau_{SB}}}\frac{\tau_B}{\tau_{SB}}  \leq \epsilon$ as well, and we are within the range of the original big-O notation. By taking $12(1+C\epsilon)$ as the new constant, we prove Eq.~\eqref{bigOexplain}.
\end{proof}

\subsection{Time averaging error and its optimization}
\label{sec:t-ave-err2}

\subsubsection{General analysis of time-averaging}

We consider the time-averaging operation, which leads to a more complicated bound on the difference of the solutions. We first prove a general result:
\begin{mylemma}
\label{LemTave}
For any first order linear differential equation $\dot{\rho}(t) = \mathcal{L}_t\rho(t)$ there is a time-averaged version
\begin{equation}
      \dot{\pi}(t) = \overline{\mathcal{L}_t}\pi(t)  ,\quad \overline{\mathcal{L}_t} = \frac{1}{T_a}\int_{-T_a/2}^{T_a/2}\mathcal{L}_{t+\tau}d\tau\ .
\end{equation}
Note that we use a different symbol $\pi(t)$ for the solution of the time-averaged equation. Assume the same initial condition: $\rho(0) = \pi(0)$.
Then the deviation between the two, $\delta\rho(t)\equiv \pi(t) - \rho(t)$, is bounded as:
\begin{align}
 \|\delta\rho(t)\|_1 \leq     [b_2(\min(t,T_a/2)) + c_\rho\Lambda  T_a/4]\max(e^{\Lambda (t-T_a/2)},1) - c_\rho\Lambda  T_a/4\ ,
 \label{eq:125}
\end{align}
where $\Lambda,c_\rho$ are constants such that $\forall t,\rho_0, \|\rho_0\|_1\leq1: ~ \|\mathcal{L}_t(\rho_0)\|_1\leq \Lambda$ and the solution $\rho(t)$ with the initial condition $\rho(0)=\rho_0$ is bounded as $\|\rho(t)\|_1\leq c_\rho$. The function $b_2(t)$ is:
\beq
 b_2(t) \equiv 4c_\rho\frac{t}{T_a}+ 
c_\rho\frac{4-2\Lambda T_a -(\Lambda T_a/2)^2}{\Lambda T_a} (1 -e^{\Lambda t})   \ .
\eeq
An immediate corollary is: 
\beq
\|\delta \rho (t)\|_1 = O(\Lambda T_a)  \text{    for    } t = O(\Lambda^{-1})\ .
\eeq
\end{mylemma}
The proof can be found in Appendix~\ref{LemProof}. The $O(\Lambda T_a)$ corollary enters the estimate for $\|\rho_{BM,I}(t) -\rho_{BMT,I}(t)\|_1$ we presented in Eq.~\eqref{solErr} in the derivation (using $\Lambda = 4/\tau_{SB}$). 
We will also need a compact bound on:
\begin{equation}
     \frac{b_2(T_a/2)}{c_\rho} =   e^{\Lambda T_a/2}\frac{(\Lambda T_a/2)^2 +4(e^{-\Lambda T_a/2}   -1 +\Lambda T_a/2)}{\Lambda T_a} -\frac{ \Lambda T_a}{4} \ .
\end{equation}
One can prove that the expression in the brackets is upper-bounded by
\begin{equation}
    e^{-\Lambda T_a/2}   -1 +\Lambda T_a/2 \leq (\Lambda T_a)^2/8\ .
\end{equation}
This allows us to bound
\begin{equation}
    b_2(T_a/2) \leq c_\rho\left(  e^{\Lambda T_a/2} \frac{3\Lambda T_a}{4} -\frac{ \Lambda T_a}{4}\right)\ .
    \label{eq:143}
\end{equation}
Note that the function $b_2(t)$ is monotonic: $b_2(t) \leq b_2(T_a/2)$. One can then loosen the bound in Lemma \ref{LemTave}:
\begin{equation}
  \|\delta\rho(t)\|_1 \leq   \left\{ \begin{array}{ll} ~ c_\rho\left( e^{\Lambda T_a/2} \frac{3\Lambda T_a}{4} -\frac{ \Lambda T_a}{4} \right), &t\leq T_a/2  
   \\ ~ c_\rho\left( e^{\Lambda t} \frac{3\Lambda T_a}{4} -\frac{ \Lambda T_a}{4} \right), &t> T_a/2 
\end{array} \right. = c_\rho\left(  e^{\Lambda \max (t,T_a/2)} \frac{3\Lambda T_a}{4} -\frac{ \Lambda T_a}{4}\right) \ .
\label{Lambded}
\end{equation}

\subsubsection{Time averaging of the Redfield ME}

We now focus on the specific equations of interest to us. 
We already established in Eq.~\eqref{newLambda} that we can choose $\Lambda= 4/\tau_{SB}$ (recall that $\Lambda$ needs to upper bound the r.h.s.: $\max_{t,\rho_{\text{test}}, \|\rho_{\text{test}}\|_1=1}\|\mathcal{L}^{BM,I}_t(\rho_{\text{test}})\|_1\leq \Lambda$).
The total time-averaging error~\eqref{Lambded} is now bounded as:
\begin{align}
    \|\rho_{BM}-\rho_{BMT}\|_1 =\|\delta\rho(t)\|_1\leq c_{BM}\left(  e^{\Lambda \max (t,T_a/2)} \frac{3 T_a}{\tau_{SB}} -\frac{  T_a}{\tau_{SB}}\right) \ , 
    \label{TAforAdj}
\end{align}
We will use both the compact bound of Eq.~(\ref{TAforAdj}), and the tightest bound available which is a repetition of Lemma \ref{LemTave}:
\begin{align}
    \|\delta\rho(t)\|_1 \leq b_2(t)=     [b_2(\min(t,T_a/2)) + c_{BM}\Lambda  T_a/4]\max(e^{\Lambda (t-T_a/2)},1) - c_{BM}\Lambda  T_a/4\ , \label{repeatTAer}\\  b_2(t) = 4c_{BM}\frac{t}{T_a} + 
c_{BM}\frac{4-2\Lambda T_a -(\Lambda T_a/2)^2}{\Lambda T_a} (1 -e^{\Lambda t}), \quad t< T_a/2 \ .
\end{align}

\subsubsection{Bound on the error due to enforcing complete positivity}

Recall that after the time-averaged equation is obtained, we need to drop an extra portion of the integral to obtain complete positivity [the transition to Eq.~\eqref{eq:26}]. Before that, Eq.~\eqref{eq:11} was:
\begin{equation}
{\dot\rho}_{BMT}(t) =-i[H,\rho_{BMT}] + \sum_{\omega, \omega'}\frac{1}{T_a}\int_{-T_a/2}^{T_a/2} dt' \int_{-t}^{t'} d\tau'\mathcal{C} (\tau' -t')e^{-i(\omega t'+\omega' \tau')}   (A_{\omega} \rho_{BMT} A_{\omega'} - \rho_{BMT} A_{\omega'} A_{\omega}) +\textrm{h.c.}\ ,
\end{equation}
which changes into Eq.~\eqref{eq:26} (same terms, but the portion of the integral is subtracted):
\bes
\begin{align}
{\dot\rho}_C(t) &=-i[H,\rho_C] + \sum_{\omega, \omega'}
\frac{1}{T_a}\int_{-T_a/2}^{T_a/2} dt' \int_{-t}^{t'} d\tau'\ \mathcal{C} (\tau'-t')e^{-i(\omega t'+\omega' \tau')}(A_{\omega} \rho_C A_{\omega'} - \rho_C A_{\omega'} A_{\omega}) +\textrm{h.c.} - \mathcal{E}_p \\
\mathcal{E}_p &= \sum_{\omega, \omega'}
\frac{1}{T_a}\int_{-T_a/2}^{T_a/2} dt' \int_{-t}^{-T_a/2} d\tau'\ \mathcal{C} (\tau'-t')e^{-i(\omega t'+\omega' \tau')}(A_{\omega} \rho_C A_{\omega'} - \rho_C A_{\omega'} A_{\omega}) +\textrm{h.c.}
\end{align}
\ees
Collecting $\sum_{\omega}$ into $A(t)$ we obtain:
\begin{equation}
    \mathcal{E}_p =
\frac{1}{T_a}\int_{-T_a/2}^{T_a/2} dt' \int_{-t}^{-T_a/2} d\tau'\ \mathcal{C} (\tau'-t')(A(t') \rho_C A(\tau') - \rho_C A(\tau') A(t')) +\textrm{h.c.}\ ,
\end{equation}
so that:
\begin{equation}
    \|\mathcal{E}_p\|_1 \leq\frac{4}{T_a}\int_{-T_a/2}^{T_a/2} dt' \left|\int_{-t}^{-T_a/2} d\tau'|\mathcal{C}(\tau' -t')|\right| \ .
\end{equation}
\begin{figure}[t]
\includegraphics[width=0.8\linewidth]{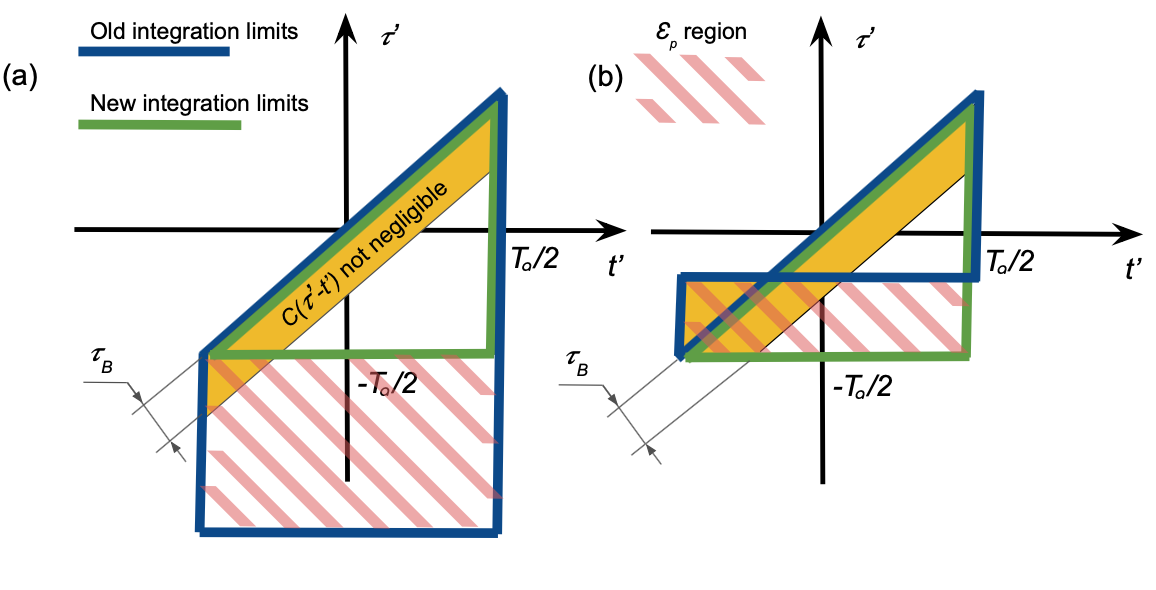}
\caption{Times contributing to $\mathcal{E}_p$ for (a) $t>T_a/2$ (b) $t<T_a/2$}
\label{regionTa}
\end{figure}
\noindent Since the resulting equation is CP, we used $\|\rho_C\|_1 =1$. We would like to bound this norm by splitting the integration domain into different times. The corresponding domains are illustrated in Fig.~\ref{regionTa}. 
For $t<T_a/2$ we note that the argument of $\mathcal{C}(\tau' - t')$ can be both positive and negative. Thus a factor of $2$ arises when we bound $\int dt'|\mathcal{C}(\tau' - t')| \leq 2/\tau_{SB}$: 
\begin{equation}
    \|\mathcal{E}_p\|_1 \leq\frac{8}{\tau_{SB}} \frac{T_a/2 -t}{T_a} \leq \frac{4}{\tau_{SB}}\ .
\end{equation}

For $t>T_a/2$, we perform a change of variables: $\tau' -t' = -\tau, ~ t' + T_a/2 = \theta$:
\begin{equation}
    \|\mathcal{E}_p\|_1 \leq\frac{4}{T_a}\int_{-T_a/2}^{-T_a/2} dt' \int_{-t}^{-T_a/2} d\tau'|\mathcal{C}(\tau' -t')| = \frac{4}{T_a}\int_{0}^{T_a} d\theta \int_{\theta}^{t + \theta -T_a/2 } d\tau|\mathcal{C}(\tau)|\ .
\end{equation}
We now change the order of integration, noting that $0<\theta<\tau$, thus obtaining the following upper bound:
\begin{equation}
    \|\mathcal{E}_p\|_1 \leq \frac{4}{T_a}\int_{0}^{T_a} d\theta \int_{\theta}^{t + \theta -T_a/2 } d\tau|\mathcal{C}(\tau)| \leq\frac{4}{T_a} \int_{0}^{t  +T_a/2 } \tau|\mathcal{C}(\tau)|d\tau\ .
\end{equation}
Recall that the total evolution time  $T$ entered the definition $\frac{\tau_B}{\tau_{SB}} = \int_0^T t |\mathcal{C}(t)| dt$ [Eq.~\eqref{eq:T1tauB-b}]. We can usually set $T=\infty$, except for an Ohmic bath [see Eq.~\eqref{Ohmm}], which is a borderline applicability case, with a logarithmic in $T$ divergence of the integral.  We define an evolution time with a collar of $T_a/2$: $T =\max_t(t+T_a/2)$:
\begin{equation}
    \|\mathcal{E}_p\|_1  \leq\frac{4}{T_a} \int_{0}^{T} \tau|\mathcal{C}(\tau)|d\tau = \frac{4\tau_B}{T_a\tau_{SB}}\ .
\end{equation}
Collecting everything, we obtain:
\begin{equation}
    \|\mathcal{E}_p\|_1(t) \leq \frac{4}{\tau_{SB}}\left(\theta(t-T_a/2)\frac{\tau_B}{T_a} +\theta(T_a/2 -t)\right) \ ,
    \label{eq:155}
\end{equation}
where $\theta(x) = 1, ~ x>0$ and zero otherwise. 

Recall that $\mathcal{E}_p$ is defined in Eq.~\eqref{eq:26} as the extra term appearing if the CGME is written with the same expression as Eq.~\eqref{eq:11} in the r.h.s.
We can now use the same idea as in the proof of Lemma~\ref{factDif} (Sec.~\ref{sec:t-ave-err}). Let us introduce $\delta \rho_2 = \rho_{BMT} - \rho_{C}$. A function $b_p$ that upper-bounds $\|\delta\rho_2\|_1$ is the solution of the differential equation
\beq
    \dot{b}_p = \Lambda b_p +\|\mathcal{E}_p\|_1(t), \quad b_p(0) =0 \ ,
    \label{eq:156}
\eeq
where we have used $\Lambda\geq\max_{t,\rho_{\text{test}}, \|\rho_{\text{test}}\|_1=1}\|\mathcal{L}^{BM,I}_t(\rho_{\text{test}})\|_1$ from Eq.~(\ref{newLambda}) as the upper bound on the norm of the superoperator on the r.h.s. of Eq.~\eqref{eq:11}. Indeed, by applying the triangle inequality to the definition of the r.h.s. of Eq.~\eqref{eq:11} we get:
\begin{equation}
    \mathcal{L}^{BMT,I}_t(\rho_{\text{test}}) = \frac{1}{T_a}\int_{-T_a/2}^{T_a/2} \mathcal{L}^{BM,I}_{t+\tau}(\rho_{\text{test}})d\tau, \quad \|\mathcal{L}^{BMT,I}_t(\rho_{\text{test}})\|_1 \leq \|\mathcal{L}^{BM,I}_t(\rho_{\text{test}})\|_1 \leq \Lambda \ .
\end{equation}
The solution to Eq.~(\ref{eq:156}) is:
\beq
    b_p(t) = e^{\Lambda t} \int_0^t e^{-\Lambda \tau}\|\mathcal{E}_p\|_1(\tau) d\tau \ .
\eeq
Substituting Eq.~\eqref{eq:155}, we obtain two cases:
\begin{align}
b_p(t) &= \frac{4}{\Lambda\tau_{SB}}e^{\Lambda t} \left(1-e^{-\Lambda t } \right),\quad  t< T_a/2\\
b_p(t) &= \frac{4}{\Lambda\tau_{SB}}e^{\Lambda t} \left((1-e^{-\Lambda T_a/2 } )  +\frac{\tau_B}{T_a}(e^{-\Lambda T_a/2 }-e^{-\Lambda t} ) \right),\quad  t> T_a/2
\ ,
\end{align}
which we can summarize as
\beq
b_p(t) \leq \frac{4}{\Lambda\tau_{SB}}[(e^{\Lambda t}  - \max(1,e^{\Lambda (t-T_a/2)})) + \frac{\tau_B}{T_a}\max(e^{\Lambda (t-T_a/2 )}-1,0 ) ] \ . \label{zeroStart}
\eeq
For a concise form we change the first term using  $1-e^{-x}\leq x$:
\beq
b_p(t) \leq e^{\Lambda t}\frac{2T_a}{\tau_{SB}} + \frac{4\tau_B}{\Lambda T_a\tau_{SB}}\max(e^{\Lambda (t-T_a/2 )}-1,0 )  \ . \label{LambdaStay}
\eeq

Let us now focus on the difference with $\rho_{BM,I}(t)$ which is the solution of Eq.~\eqref{RedStart}. 
We will keep $\Lambda\geq\max_{t,\rho_{\text{test}}, \|\rho_{\text{test}}\|_1=1}\|\mathcal{L}^{BM,I}_t(\rho_{\text{test}})\|_1$ in the bound in case it is $<4/\tau_{SB}$. Using Eqs.~\eqref{TAforAdj} and~\eqref{LambdaStay}, we have:
\begin{align}
   \|\rho_{BM,I}(t) - {\rho}_{C,I}(t)\|_1\leq c_{BM}\left(  e^{\Lambda \max (t,T_a/2)} \frac{3 T_a}{\tau_{SB}} -\frac{  T_a}{\tau_{SB}}\right) +  e^{\Lambda t}\frac{2T_a}{\tau_{SB}} + \frac{4\tau_B}{\Lambda T_a\tau_{SB}}\max(e^{\Lambda (t-T_a/2 )}-1,0 )\ . \label{fullWLam}
\end{align}
The tightest bound available that starts at zero at $t=0$ is the combination of Eqs.~\eqref{repeatTAer} and~\eqref{zeroStart}: 
\bes
\label{strongestBound}
\begin{align}
   \|\rho_{BM,I}(t) - {\rho}_{C,I}(t)\|_1\leq [b_2(\min(t,T_a/2)) + c_{BM}\Lambda  T_a/4]\max(e^{\Lambda (t-T_a/2)},1) - c_{BM}\Lambda  T_a/4+\\  +\frac{4}{\Lambda\tau_{SB}}[(e^{\Lambda t}  - \max(1,e^{\Lambda (t-T_a/2)})) + \frac{\tau_B}{T_a}\max(e^{\Lambda (t-T_a/2 )}-1,0 ) ]\ . \end{align}
\ees
where for $  t< T_a/2$
\begin{equation}
 b_2(t) = 4c_{BM}\frac{t}{T_a} + 
c_{BM}\frac{4-2\Lambda T_a -(\Lambda T_a/2)^2}{\Lambda T_a} (1 -e^{\Lambda t})
\end{equation}
We may optimize $T_a$ to minimize the bound Eq.~(\ref{fullWLam}), which we proceed to do next.

\subsubsection{Optimization of $T_a$}
\label{meatOptim}
We have used $\Lambda \leq 4/\tau_{SB}$ to obtain the form of the bound in Eq.~(\ref{fullWLam}), but we kept both $\Lambda$ and $\tau_{SB}$ to demonstrate the possibility of making the bound tighter by using different values for $\Lambda$ and $4/\tau_{SB}$ evaluated independently. To that end, we denote $c_{\Lambda} = 4/\Lambda \tau_{SB}\geq 1$ and  treat it as a constant below. At large $t \gg \Lambda^{-1}, ~ t>T_a/2$ the leading terms of Eq.~(\ref{fullWLam})  are:
\begin{equation}
   \|\rho_{BM,I}(t) - {\rho}_{C,I}(t)\|_1\leq c_{BM}e^{\Lambda t}\frac{3T_a}{\tau_{SB}} + e^{\Lambda t}  \frac{2T_a}{\tau_{SB}}   +c_{\Lambda}e^{\Lambda(t-T_a/2 )}\frac{\tau_B}{T_a} +o(e^{\Lambda t}) \ ,
   \label{weakBound}
\end{equation}
We can loosen the bound again via $e^{-\Lambda T_a/2} \leq 1$:
\begin{equation}
   \|\rho_{BM,I}(t) - {\rho}_{C,I}(t)\|_1\leq e^{\Lambda t} \left(\frac{(3c_{BM}+2)T_a}{\tau_{SB}}    +c_{\Lambda}\frac{\tau_B}{T_a}\right) +o(e^{\Lambda t}) \ .
\end{equation}
Minimizing the expression in the brackets, we find that $T_a = \sqrt{c_{\Lambda}\tau_{SB}\tau_B/(3c_{BM}+2)}$ [$=\sqrt{\tau_{SB}\tau_B/5}$ if $c_{BM}=c_{\Lambda}=1$, which is Eq.~\eqref{eq:T_a-opt}] for any $\tau_B/\tau_{SB}$ leads to an optimal bound at large enough times. The asymptotic behavior of the bound is:
\begin{equation}
   \|\rho_{BM,I}(t) - {\rho}_{C,I}(t)\|_1\leq 2e^{\Lambda t}\sqrt{\frac{c_\Lambda(3c_{BM}+2)\tau_B}{\tau_{SB}}}  + o(e^{\Lambda t}) \ . \label{startApp0}
\end{equation}
One may be concerned that the relaxation $e^{-\Lambda T_a/2} \leq 1$ reduced the quality of our bound. We checked this for $c_{BM}=c_{\Lambda}=1$, though we do not present that calculation here. The result is that after the optimization of the original Eq.~(\ref{strongestBound}), which is a more complicated function of $T_a$, the bound is unchanged to the leading order, and in fact acquires a multiplicative correction of the form $\left( 1-\frac{17}{15\sqrt{5}}\frac{\sqrt{\tau_B}}{\sqrt{\tau_{SB}}} + O(\tau_B/\tau_{SB})) \right)$.
We do not pursue this route since the effect is small, and the calculation is unreasonably long. Instead we use $T_a = \sqrt{c_{\Lambda}\tau_{SB}\tau_B/(3c_{BM}+2)}$ obtained above and find that the full expression for the error from Eq.~(\ref{fullWLam}) can be bounded as 
\begin{equation}
     \|\rho_{BM,I}(t) - {\rho}_{C,I}(t)\|_1\leq 2\sqrt{\frac{c_{\Lambda}(3c_{BM}+2)\tau_B}{\tau_{SB}}}\left(e^{\Lambda t+1}-3/5\right)\ ,  \label{looseBound}
\end{equation}
and it holds for all positive $\tau_B/\tau_{SB}$, $t/\tau_{SB}$ and $c_{BM},c_{\Lambda} \geq 1$. The derivation can be found in Appendix~\ref{boundingDetails}. This simple form allows us to finally bound $c_{BM}$. Since for a CP coarse-grained equation $\|\rho_C\|_1 =1$:
\begin{equation}
     \|\rho_{BM,I}(t)\|_1 \leq 1 +  2\sqrt{\frac{c_{\Lambda}(3c_{BM}+2)\tau_B}{\tau_{SB}}}\left(e^{\Lambda t+1}-3/5\right), \quad c_{BM}= 1 + 2\sqrt{\frac{c_{\Lambda}(3c_{BM}+2)\tau_B}{\tau_{SB}}}\left(e^{\Lambda t+1}-3/5\right)
\end{equation}
This is a quadratic equation on $c_{BM}$. Its $\geq 1$ solution is:
\begin{equation}
    c_{BM} = 1 + \frac12 \left( \sqrt{20 X + 9 X^2} +3 X \right) = 1 + O\left(\sqrt{\frac{\tau_B}{\tau_{SB}}}e^{\Lambda t}\right), \quad X  = 4\frac{c_{\Lambda}\tau_B}{\tau_{SB}}\left(e^{\Lambda t+1}-3/5\right)^2 \label{cBMexact}
\end{equation}
And here the big-$O$ notation means that there exist constants $x_0,C$ such that for $0\leq x\leq x_0, O(x) \leq C x$. We note that the choice of $T_a = \sqrt{c_{\Lambda}\tau_{SB}\tau_B/(3c_{BM}+2)}$ depends on the total evolution time, if we use the above bound on $c_{BM}$. We would like to approximate it by a time-independent $T_a$. Note that for $\sqrt{\frac{\tau_B}{\tau_{SB}}}e^{\Lambda t} \ll 1$ we can approximate $c_{BM}\approx 1$ and consequently $T_a =\sqrt{c_{\Lambda}\tau_{SB}\tau_B/5}$. Note also  that our construction was optimal for asymptotic $e^{\Lambda t} \gg 1$ (for the discussion of small times, see Appendix \ref{optAside}). We combine the two conditions as $1\ll e^{\Lambda t} \ll \sqrt{\frac{\tau_{SB}}{\tau_B}}$. In units of actual time, this corresponds to
\begin{equation}
   1\ll \Lambda t \ll \text{ln}\sqrt{\frac{\tau_{SB}}{\tau_B}}\ .
\end{equation}
This is a fairly narrow range for any realistic parameters, but it still becomes $[1,\infty]$ in the weak coupling limit. We choose to work with  $T_a =\sqrt{c_{\Lambda}\tau_{SB}\tau_B/5} $ that is approximately optimal within that range. From a purely mathematical point of view, there was no particular reason to choose it this way. From a physical point of view, we would like to work with errors of $\leq 10\% $, which would correspond to a range like this. We have used $T_{a,\text{adj}} =\sqrt{c_{\Lambda}\tau_{SB}\tau_B/5}$ and $T_{a,\text{theory}} =\sqrt{\tau_{SB}\tau_B/5}$ for the numerical simulations in Section \ref{numeric}. For theoretical purposes of the remainder of this Section, we can always use $c_{\Lambda}=1$ which is proven. We will keep the notation $e^{\Lambda t} = e^{4t/\tau_{SB}}$ for brevity. The error bound for $T_a =\sqrt{\tau_{SB}\tau_B/5} $ is derived as Eq.~\eqref{sqrt5allT}: 
\begin{equation}
     \|\rho_{BM,I}(t) - {\rho}_{C\sqrt{5},I}(t)\|_1\leq \sqrt{\frac{\tau_B}{5\tau_{SB}}} \left( \left(3e^{\Lambda t+1} -1\right)c_{BM} + \left(7e^{\Lambda t+1} -5\right)\right)\ . \label{simplifySqrt5}
\end{equation}
The constant $c_{BM}$ can be obtained numerically independently, and then one just needs to substitute it. For the theoretical use of the CGME, that route is not available, and instead one needs to substitute its upper bound Eq.~\eqref{cBMexact}. For the rest of this section, we will loosen the resulting bound to make the expression more compact.  Since we optimized for $e^{\Lambda t}\gg 1$, dropping constants will not be significant in the parameter regime of interest. We will also update the expression for $c_{BM}$ using $\sqrt{a+b} \leq \sqrt{a} +\sqrt{b}$ in Eq.~\eqref{cBMexact}:
\begin{equation}
     c_{BM} \leq 1 + \sqrt{5X} +3X, \quad \sqrt{X} \leq  2\sqrt{\frac{\tau_B}{\tau_{SB}}}e^{\Lambda t+1}\ . \label{cBMX}
\end{equation}
Now we simplify Eq.~\eqref{simplifySqrt5} by relaxing the constant terms:
\begin{equation}
     \|\rho_{BM,I}(t) - {\rho}_{C\sqrt{5},I}(t)\|_1\leq \sqrt{\frac{\tau_B}{5\tau_{SB}}}e^{\Lambda t
     +1} \left( 3 c_{BM} + 7 \right)  \ .
\end{equation}
With the relaxed form of $c_{BM}$ above this becomes:
 \begin{equation}
     \|\rho_{BM,I}(t) - {\rho}_{C\sqrt{5},I}(t)\|_1\leq \alpha(t) e (2\sqrt{5} + 6e \alpha(t) +\frac{36e^2}{\sqrt{5}}\alpha^2(t)) ,\quad \alpha(t) =\sqrt{\frac{\tau_B}{\tau_{SB}}}e^{\Lambda t}\ .
\end{equation}
The error near $t=0$ is doubled as compared to the longer formula~\eqref{simplifySqrt5}. Moreover, at each of decreasing timescales $t \sim \Lambda^{-1}$, $t\sim T_a$, $t=0$ the error gets progressively worse in our approximations, as compared to the tightest bound possible: the error at $t=0$ is supposed to be $0$. Nevertheless, we proceed, as this route allows us to obtain a concise form and the very small $t$ region is less important for our purposes. The bound can be further relaxed  with rounding, using $6e^2 \alpha \leq 13-2\sqrt{5}e + \frac{9e^3\alpha^2}{13-2\sqrt{5}e}$:
\begin{equation}
    \|\rho_{BM,I}(t) - {\rho}_{C\sqrt{5},I}(t)\|_1\leq 13\sqrt{\frac{\tau_B}{\tau_{SB}}}e^{4t/\tau_{SB}}\left(1 + 29 \frac{\tau_B}{\tau_{SB}}e^{8t/\tau_{SB}}\right)\ .
    \label{compactErr}
\end{equation}
where we plugged in $\Lambda = 4/\tau_{SB}$. We will use this form while computing the total error.
We still need to include the Markov and Born errors. Adding the bound in Eqs.~\eqref{compactErr} and~\eqref{BornObtained}, and using the triangle inequality, we finally obtain:
\begin{equation}
     \|\rho^{\text{true}}_t-\rho^{C\sqrt{5}}_t\|_1  \leq \frac{\left(e^{\frac{4t}{\tau_{SB}}} -1\right)e^{\frac{8t}{\tau_{SB}}}\tau_B \left(12+  O\left(e^{\frac{4t}{\tau_{SB}}}\frac{\tau_B}{\tau_{SB}}\right)\right)}{\tau_{SB}}  +\frac{13\sqrt{\tau_B}e^{\frac{4t}{\tau_{SB}}}\left(1 +  \frac{29\tau_B e^{\frac{8t}{\tau_{SB}}}}{\tau_{SB}}\right)}{\sqrt{\tau_{SB}}}\ . \label{BoundObtained}
\end{equation}
Recall that this bound is applicable within the region $e^{\frac{4t}{\tau_{SB}}}\frac{\tau_B}{\tau_{SB}} \leq C$ for some constant $C$ (from big-$O$), but within that region any triple $\{t,\tau_B, \tau_{SB}\}\geq 0$ is allowed. In particular, for sufficiently small $\tau_B/\tau_{SB}$, the error at any value of $t/\tau_{SB}$ can be shown to be small. As a corollary, we proved the validity in the weak coupling limit, but the above inequality is more general than that. In particular, this proves Eq.~\eqref{eq:58}, which uses $T_a= \sqrt{\tau_B \tau_{SB}/5}$. 

We will finally show how to simplify Eq.~(\ref{BoundObtained}) to the form $O\left(e^{\frac{6t}{\tau_{SB}}}\sqrt{\frac{\tau_B}{\tau_{SB}}}\right)$ presented in Eq.~\eqref{eq:59b}. Denoting $x =e^{\frac{6t}{\tau_{SB}}}\sqrt{\frac{\tau_B}{\tau_{SB}}} $, we can relax the bound as follows:
\begin{equation}
     \|\rho^{\text{true}}_t-\rho^{C\sqrt{5}}_t\|_1  \leq x^2 \left(12+  O\left(e^{\frac{4t}{\tau_{SB}}}\frac{\tau_B}{\tau_{SB}}\right)\right) +13x\left(1 +  29x^2\right)\ . \label{bigOwork}
\end{equation}
We then use the property of Big O notation: for any function $f(t/\tau_{SB},\tau_B/\tau_{SB})$, if $f(t/\tau_{SB},\tau_B/\tau_{SB}) = O\left(e^{\frac{4t}{\tau_{SB}}}\frac{\tau_B}{\tau_{SB}}\right)$, then also $f(t/\tau_{SB},\tau_B/\tau_{SB}) = O(x^2)$. We obtain:
\begin{equation}
     \|\rho^{\text{true}}_t-\rho^{C\sqrt{5}}_t\|_1  \leq x^2 \left(12+  O\left(x^2\right)\right) +13x\left(1 +  29x^2\right)\ .
\end{equation}
From this, the result of Eq.~\eqref{eq:59b} immediately follows:
\begin{equation}
     \|\rho^{\text{true}}_t-\rho^{C\sqrt{5}}_t\|_1  \leq O(x) = O\left(e^{\frac{6t}{\tau_{SB}}}\sqrt{\frac{\tau_B}{\tau_{SB}}}\right)\ .
\end{equation}

\subsection{The Davies-Lindblad error}
\label{DavSec}

We write the CGME with an explicit dependence on $T_a$:
\begin{equation}
{\dot\rho_C}(t) =-i[H +H_{\mathrm{LS}},\rho_C] + \sum_{\omega, \omega'}\gamma_{\omega\omega'}^{T_a} (A_{\omega} \rho_C A_{\omega'} - \frac{1}{2}\left\{ \rho_C , A_{\omega'} A_{\omega}\right\} )\ .
\end{equation}
The error bound is known for arbitrary $T_a$: see Eq.~\eqref{weakBound}. We now wish to compare the CGME with finite $T_a$ to the limit $T_a \to \infty$ of the same equation:
\begin{equation}
{\dot\rho_C}(t) =-i[H +H_{\mathrm{LS}},\rho_C] + \sum_{\omega, \omega'}\gamma_{\omega\omega'}^{T_a =\infty} (A_{\omega} \rho_C A_{\omega'} - \frac{1}{2}\left\{ \rho_C , A_{\omega'} A_{\omega}\right\} )  + \mathcal{E}_R\ .
\end{equation}
In Appendix \ref{app:D}, we show that this limit is just the Davies-Lindblad equation:
\begin{equation}
{\dot\rho_D}(t) =-i[H +H_{\mathrm{LS}},\rho_D] + \sum_{\omega, \omega'}\gamma_{\omega\omega'}^{T_a =\infty} (A_{\omega} \rho_D A_{\omega'} - \frac{1}{2}\left\{ \rho_D , A_{\omega'} A_{\omega}\right\} )\ .
\end{equation}
We also show in Appendix \ref{app:D} that the difference between the right-hand sides is bounded as:
\begin{equation}
    \|\mathcal{E}_R\|_1 \leq 4\sum_{\omega\omega'} |\gamma_{\omega \omega'}^{T_a} - \gamma_{\omega \omega'}^{T_a=\infty}|\leq \frac{4^n}{\tau_{SB}} O\left(\frac{1}{T_a\delta E}\right)\ .
\end{equation}
Here the system Hilbert space dimension $2^n$ appeared, and we note that $1/\delta E$ could contain an additional exponential dependence on $n$. The accumulated error in the solution is bounded as:
\begin{equation}
    \|\rho_{\text{true}}(t) -\rho_{\text{D}}(t)\|_1 \leq (e^{\Lambda t} -1) \frac{\|\mathcal{E}_R\|_1}{\Lambda} + \|\rho_{\text{true}} -\rho_{\text{C}}\|_1\ .
\end{equation}
In the long time limit the leading terms are: 
\begin{equation}
     \|\rho_{\text{true}}(t) -\rho_{\text{D}}(t)\|_1 \leq e^{4t/\tau_{SB}} 4^n O\left(\frac{1}{T_a\delta E}\right) + \frac{5T_a}{\tau_{SB}}e^{4t/\tau_{SB}} +e^{4t/\tau_{SB}}\frac{\tau_B}{T_a} +\|\rho_{\text{true}} -\rho_{\text{BM}}\|_1+o(e^{4t/\tau_{SB}}) \ .
\end{equation}
Assuming that $\frac{4^n}{\delta E} \gg \tau_B $, we find the optimal $T_a =2^nO(\sqrt{\tau_{SB}/\delta E})$. Using that together with a bound for $t\geq T_a/2$ we get:
\begin{equation}
\|\rho_{\text{true}}(t) -\rho_{\text{D}}(t)\|_1\leq  2^n O\left( \frac{e^{4t/\tau_{SB}}}{\sqrt{\tau_{SB}\delta E}} \left(1 +\frac{\tau_B}{\tau_{SB}} \right) \right) +\|\rho_{\text{true}} -\rho_{\text{BM}}\|_1 \leq  2^n O\left( e^{12t/\tau_{SB}} \left(\frac{1}{\sqrt{\tau_{SB}\delta E}} +\frac{\tau_B}{\tau_{SB}} \right) \right) \ .
 \label{LblBound}
\end{equation}
Assuming a constant system size $n=O(1)$, this is the form presented in Eq.~\eqref{eq:72}, our second main error bound. 

Then, assuming $\frac{4^n}{\delta E} \gg \tau_B $, we can neglect the second term. Recall the discussion in Sec.~\ref{sec:times}: $g = 1/\sqrt{\tau_B \tau_{SB}}$, where $g$ is the coupling strength $g$. Thus, we can express $\sqrt{1/\tau_{SB}\delta E} =g \sqrt{\tau_B/\delta E}$. If we set $t<c\tau_{SB}$, we can write $ \|\rho_{\text{true}}(t) -\rho_{\text{D}}(t)\|_1 =O(\sqrt{1/\tau_{SB}\delta E}) = O(g \sqrt{\tau_B/\delta E})$, as stated in Eq.~\eqref{eq:1} in the Introduction.

\subsection{Redfield limits of integration error}
\label{errLimits}

There is one more error left to analyze: the error due to changing $\int^{t} \to \int^{\infty}$ in the Redfield equation. We already had the bound given in Eq.~\eqref{eq:31c}:
\begin{equation}
    \|\mathcal{E}_l\|_1 \leq 4c_{BM}\int_t^{\infty} |\mathcal{C} (t')|dt' \leq \frac{4c_{BM}}{t}\int_t^{T}t' |\mathcal{C} (t')|dt' + 4c_{BM}\int_T^{\infty} |\mathcal{C} (t')|dt'  \leq \frac{4c_{BM}\tau_B}{t\tau_{SB}} + \frac{4c_{BM}\epsilon_T}{\tau_{SB}}\ ,
    \label{eq:184}
\end{equation}
where $\epsilon_T= \tau_{SB}\int_T^{\infty} |\mathcal{C} (t')|dt'$ [Eq.~\eqref{eq:eps_T}]. We augment Eq.~\eqref{eq:184} by a trivial bound:
\begin{equation}
    \|\mathcal{E}_l\|_1 \leq 4c_{BM}\int_t^{\infty} |\mathcal{C} (t')|dt' \leq \frac{4c_{BM}}{\tau_{SB}} \min (1, \tau_B/t + \epsilon_T)\ .
\end{equation}
Since the bound on $\|\mathcal{E}_l\|_1$ is time dependent, we cannot directly use Lemma~\ref{factDif} for Eq.~\eqref{eq:31a}. Instead, we repeat its proof and integrate the corresponding differential equations as was done in the proof of Lemma~\ref{LemTave}. We define a function  $b_R(t)\geq \|\rho_{BM}(t) - \rho_R(t)\|_1$ (recall that $\rho_{BMl} = \rho_R$) that satisfies the following differential equation, analogous to Eq.~\eqref{eq:107}:
\begin{equation}
    \dot{b}_R = \frac{4b_R}{\tau_{SB}} + \frac{4c_{BM}}{\tau_{SB}} \min (1, \tau_B/t + \epsilon_T)
\end{equation}
Integrating it, we obtain:
\begin{equation}
    \|\rho_{BM}(t) - \rho_R(t)\|_1 { \leq \delta(t)} = c_{BM} e^{4t/\tau_{SB}} \int_0^{4t/\tau_{SB}}\min (1, \frac{4\tau_B}{x\tau_{SB}} + \epsilon_T)e^{-x}dx\ .
    \label{eq:182}
\end{equation}
As shown in Appendix~\ref{app:Red-details}, this results in the error bound:
\bes
\begin{align}
    \|\rho_{BM}(t) - \rho_R(t)\|_1  &\leq c_{BM}e^{4t/\tau_{SB}}  \left[  \frac{4\tau_B}{\tau_{SB}}\left(1 - \ln\left(1-e^{-\frac{4\tau_B}{\tau_{SB} (1-\epsilon_T)}}\right)\right)+\epsilon_T \right]\\
    &\leq c_{BM}e^{4t/\tau_{SB}} \left[\frac{4\tau_B}{{\tau_{SB}}} \left({1 + \frac{1}{e}} +\max\left(\text{ln}\frac{\tau_{SB}(1-\epsilon_T)}{4\tau_B},0\right) \right) + \epsilon_T\right]\ .
    \label{eq:183}
\end{align}
\ees
In other words, the $O(\tau_B/\tau_{SB})$ error scaling of the master equation just after Born-Markov approximation, noted in Eq.~\eqref{BMloose} acquires a logarithmic correction.

We will now show how to simplify the total error of Redfield to obtain Eq.\eqref{err:Redfield}. We add the error of $\rho_{BM}$ from Eq.~(\ref{BornObtained}):
\begin{align}
     \|\rho^{\text{true}}_t-\rho^{R}_t\|_1  \leq &\frac{\left(e^{\frac{4t}{\tau_{SB}}} -1\right)e^{\frac{8t}{\tau_{SB}}}\tau_B \left(12+  O\left(e^{\frac{4t}{\tau_{SB}}}\frac{\tau_B}{\tau_{SB}}\right)\right)}{\tau_{SB}}  +\\&+c_{BM}e^{4t/\tau_{SB}} \left[\frac{4\tau_B}{{\tau_{SB}}} \left({1 + \frac{1}{e}} +\max\left(\text{ln}\frac{\tau_{SB}(1-\epsilon_T)}{4\tau_B},0\right) \right) + \epsilon_T\right]\ . 
\end{align}
We will set $\epsilon_T =0$ in this estimate and use the following bound on $c_{BM}$ from Eq.~(\ref{cBMX}):
\begin{equation}
     c_{BM} \leq 1 + \sqrt{5X} +3X, \quad \sqrt{X} \leq  2\sqrt{\frac{\tau_B}{\tau_{SB}}}e^{\frac{4t}{\tau_{SB}}+1}\ . 
\end{equation}
where we denoted $x =e^{\frac{12t}{\tau_{SB}}}\frac{\tau_B}{\tau_{SB}}$.
After a series of simplifications similar to the ones done after Eq.~(\ref{bigOwork}) we arrive at:
\begin{equation}
     \|\rho^{\text{true}}_t-\rho^{R}_t\|_1  \leq O(x)(1 +\max\left(\text{ln}\frac{\tau_{SB}}{4\tau_B},0\right))\ .
\end{equation}
Since big O notation is defined for $x\leq \epsilon$ for arbitrary small $\epsilon$, it follows that $\tau_B/\tau_{SB} \leq \epsilon$ is also small, so the logarithm is the leading term:
\begin{equation}
     \|\rho^{\text{true}}_t-\rho^{R}_t\|_1  \leq O(x)\text{ln}\frac{\tau_{SB}}{\tau_B} = O\left(e^{\frac{12t}{\tau_{SB}}}\frac{\tau_B}{\tau_{SB}}\right)\text{ln}\frac{\tau_{SB}}{\tau_B}\ .
\end{equation}
This is Eq.~\eqref{err:Redfield}.

\section{Summary and Conclusions}
\label{sec:conc}

Our primary goal in this work was to derive a Markovian master equation with a complete positivity guarantee that is in addition capable of incorporating arbitrary time-dependent driving, and that has a larger range of applicability than the Davies-Lindblad Markovian master equation. The master equation that achieves these desiderata is based on the coarse-graining approach introduced in Ref.~\cite{Majenz:2013qw}, which only considered the time-independent (no driving) case. We rederived this result by showing that it results from the non-CP Redfield master equation after neglecting part of the time-integration domain, which restored complete positivity. We then included arbitrary time-dependent driving and obtained a new CGME that is still CP, and retains essentially the same form as its time-independent counterpart. This is our central new result, given in Eq.~\eqref{tdCoarse}.

This time-dependent CGME depends on one free parameter, the coarse graining time $T_a$. We showed that it is optimal to choose it as proportional to the geometric mean of two key timescales, the bath correlation time $\tau_B$ and the fastest system decoherence timescale $\tau_{SB}$, given in Eq.~\eqref{eq:T1tauB}. These timescales are related to the bath spectral density $\g(\o)$ via the inequalities $2/\tau_{SB} \geq \g(\o)$ and $2\tau_B/\tau_{SB} \geq |\g'(\o)|$. The $1$-parameter family of CGME equations involves a continuous integral, or equivalently, a sum over  exponentially many terms in the frequency domain. Alongside this $1$-parameter family, we also derived an approximation in terms of $3$-parameter family ($T_a$, $\Delta \epsilon$, $k^*$), given in Eq.~\eqref{eq:discrete-approx}. We showed that if $[H,A]=O(1)$, where $H$ is the system Hamiltonian and $A\otimes B$ is the interaction with the bath [Eq.~\eqref{eq:Htot}], then there is no system-size dependence in $\Delta \epsilon, k^*$, where $k^*$ is the number of terms replacing the continuous integral.

Along the way we provided rigorous error bounds on three different master equations: Redfield, Davies-Lindblad, and the CGME. These bounds  [Eqs.~\eqref{err:Redfield}-\eqref{eq:59b}] are expressed in terms of the dimensionless ratio $\tau_B/\tau_{SB}$, and in the Davies-Lindblad case also in terms of the quantity $1/\sqrt{\tau_{SB}\delta E}$, where $\delta E$ is the minimum level spacing. These results are valid even when $\| B\|$ diverges, so they apply in particular to, e.g., oscillator baths. When the system-bath coupling strength $g$ is finite we can write $\tau_B/\tau_{SB} = (g \tau_B)^2$ and $1/\sqrt{\tau_{SB}\delta E} = g \sqrt{\tau_B/\delta E}$.
These bounds provide a new perspective on the range of applicability of Markovian master equations: the Redfield and CG master equations apply in the regime where $g \tau_B \ll 1$, but the Davies-Lindblad master equation applies when in addition the minimum level spacing is large relative to $g^2 \tau_B$. While these are all sufficient and not necessary conditions, if we informally view a large level spacing as a requirement, then it is a very onerous one indeed, since in many cases of interest the spacing shrinks rapidly (polynomially or even exponentially) with problem size, e.g., in adiabatic quantum computing~\cite{Albash-Lidar:RMP}. Thus the guaranteed range of applicability of the Davies-Lindblad master equation is severely limited, while the CGME achieves complete positivity without imposing this cost.

To summarize the latter in the simplest possible terms, what we have shown is that the CGME and Redfield master equations are valid for fast baths, in the sense of a very simple integral condition on the bath correlation function. Ohmic and super-Ohmic baths satisfy this condition, provided the system-bath coupling strength is weak but constant in system size (for a local system). For the Davies-Lindblad master equation the coupling strength needs to be exponentially small in the system size. However, the Redfield equation and the CGME are not on equal footing: the Redfield equation provides no guarantee of complete positivity, while the CGME does.

We list a few potential applications of the CGME:
\begin{itemize}
\item The CGME opens the door to studies of open system dynamics with arbitrary time-dependent drives $H(t)$. In particular, the drive need not be adiabatic. In fact, the {time-dependent} CGME even correctly describes a system driven by $\delta$-function pulses, as demonstrated in Sec.~\ref{ddSec} for dynamical decoupling. {Moreover, the time-dependent CGME overcomes an inherent problem in applying the Davies-Lindblad ME to circuit model quantum computing, in particular in the context of fault-tolerant quantum error correction. The latter ME is incompatible with the assumption of arbitrarily fast gates needed there~\cite{PhysRevA.73.052311}, but the time-dependent CGME can be consistently applied in this context. This opens the door to a rigorous analysis of fault-tolerance in a first-principles Markovian setting.}
\item Earlier studies of the lifetime of the toric code at finite temperature took advantage of the fact that the Davies generators are local for commuting models~\cite{Alicki:07a}. The CGME  provides local generators of open system finite temperature dynamics for non-commuting models (see Sec.~\ref{posLoc}). This opens the door to studying the finite-temperature lifetime of, e.g., topological subsystem codes~\cite{Bombin:2010aa} in the thermodynamic limit.
\item Many-body localization effects in quantum annealing --- previously studied only in the closed system setting~\cite{Altshuler2010} --- can now be extended to finite temperature open system quantum annealing. 
\end{itemize}

We hope that this work, including the applications above, will stimulate further research into the uses of the CGME, the only master equation we know of that provides a complete positivity guarantee, is locally generated, and applies for arbitrarily driven open quantum systems.

\acknowledgements
The research is based upon work (partially)
supported by the Office of the Director of National Intelligence (ODNI), Intelligence
Advanced Research Projects Activity (IARPA), via the U.S. Army Research Office contract
W911NF-17-C-0050. The views and conclusions contained herein are those of the authors and
should not be interpreted as necessarily representing the official policies or
endorsements, either expressed or implied, of the ODNI, IARPA, or the U.S. Government. The
U.S. Government is authorized to reproduce and distribute reprints for Governmental
purposes notwithstanding any copyright annotation thereon.

\newpage 

\appendix

\section{Proof that $\g(\o) > 0$}
\label{app:A}

\paragraph{Trivial case.} First consider $V = A \otimes B$. The function $\gamma(\omega)$ is the Fourier transform of the correlation function:
\begin{equation}
    \gamma(\omega) = \int_{-\infty}^\infty e^{i\omega t}\mathcal{C}(t) dt = \int_{-\infty}^\infty e^{i\omega t}\Tr[\rho_B B(t) B] dt
\end{equation}
Our derivations in this paper are for a stationary bath $\rho_B(t) = \rho_B$, or $[\rho_B,H_B] =0$. Note that $[\sqrt{\rho_B},H_B] =0$ follows.
\begin{equation}
    \gamma(\omega)= \sum_{kl} \int_{-\infty}^\infty e^{i\omega t}\Tr[\sqrt{\rho_B} e^{iH_Bt}\sqrt{\rho_B} B e^{-iH_Bt} B] dt
\end{equation}
Inserting the sum over bath eigenstates:
\begin{equation}
    \gamma(\omega)= \sum_{kl} \int_{-\infty}^\infty e^{i(\omega+ E_{kl}) t }\langle k|\sqrt{\rho_B} B |l\rangle \langle l| B\sqrt{\rho_B}|k\rangle dt
\end{equation}
Integration results in delta-functions:
\begin{equation}
    \gamma(\omega)= 2\pi\sum_{kl} \delta(\omega -E_{lk})|\langle l| B\sqrt{\rho_B}|k\rangle|^2 \geq 0
\end{equation}
This is enough for all the results in this paper. A more complicated interaction with environment, however, requires a different proof.
\paragraph{General case.}  We consider the case $V = \sum_i A_i\otimes B_i$, so that $\g(\o)$ becomes a matrix. $\rho_B$ will still be assumed stationary: $[\rho_B,H_B] =0$. Repeating the same steps, we arrive at
\begin{equation}
    \gamma_{ij}(\omega)= 2\pi\sum_{kl} \delta(\omega -E_{lk})\langle k|\sqrt{\rho_B} B_i |l\rangle \langle l| B_j\sqrt{\rho_B}|k\rangle
\end{equation}
Thus we need to prove that the matrix
\begin{equation}
    M_{ij} = \langle k|\sqrt{\rho_B} B_i |l\rangle \langle l| B_j\sqrt{\rho_B}|k\rangle
\end{equation}
is positive. Let $v_i = \langle k|\sqrt{\rho_B} B_i |l\rangle$, then
\begin{equation}
    M_{ij} = v_i v_j^*\ .
\end{equation}
For positivity, we need to prove that for any vector  $\vec{w}$ with elements $w_i$ it holds that $\sum_{ij}w_iM_{ij}w_j^* \geq 0$. Indeed:
\begin{equation}
    \forall \vec{w}, \quad \sum_{ij}w_iM_{ij}w_j^* = \sum_{ij} w_iv_i v_j^* w_j = \left|\sum_i v_i w_i\right|^2\geq 0\ .
\end{equation}
Thus
\begin{equation}
    \gamma(\omega) \geq 0
\end{equation}
as a matrix.
\paragraph{Bochner's theorem} Our proof above used the eigendecomposition of the bath, and $\gamma(\omega)$ was a sum of Dirac $\delta$ functions. In practice, a smooth $\gamma(\omega)$ is used, corresponding to, e.g., a bath with a continuous spectrum. 
In this sense it is desirable to avoid using bath eigenstates and $\delta$-functions altogether.  Here we give another proof that uses Bochner's theorem, as suggested, e.g., in the textbook~\cite{Breuer:book}, and which avoids the issue mentioned above.

\begin{mylemma}[Bochner]
If a function $\phi(t)$ is of positive type, i.e. $\forall n\in\mathbb{Z}, ~ t_i\in \mathbb{R}, ~\xi_i\in\mathbb{C}$
\begin{equation}
    \sum_{i,j=1}^n \phi(t_i-t_j)\xi_i\xi_j^* \geq 0,
\end{equation}
and also $\phi(0)$ is finite and $\phi(t)$ is continuous at $t=0$, then its Fourier transform (if it exists) is non-negative:
\begin{equation}
    \int_{-\infty}^\infty\phi(t) e^{i\omega t} dt \geq 0
\end{equation}
\end{mylemma}

We again consider many interaction terms $V = \sum_i A_i\otimes B_i$ (as in Sec.~\ref{nTerms}) and a stationary $\rho_B$ s.t. $[\rho_B,H_B]=0$. 
Since $\g(\o)$ is Hermitian we can diagonalize it using a unitary transformation:
\beq
 D \equiv U\gamma U^\dagger \Rightarrow D_{\alpha\beta} = \sum_{i,j}U_{\alpha i}\gamma_{ij}U_{\beta j}^* \ .
 \eeq
$ D $ is diagonal so we only need to consider the diagonal elements (i.e., the eigenvalues of $ \gamma $). Substituting $ \gamma_{ij} = \int_{-\infty}^\infty e^{i\omega s}\mathcal{C}_{ij}(s)ds$, where $ \mathcal{C}_{ij}(s) = \Tr[\rho_B B_i(s) B_j]$, gives
\beq
D_\alpha = \int_{-\infty}^\infty e^{i\omega s}\pen{\sum_{i,j}U_{\alpha i}\mathcal{C}_{ij}(s)U_{\alpha j}^*}ds\ .
\eeq
	We wish to show that $ D_\alpha $ is non-negative for each $ \alpha $. To do this we must consider the function in parenthesis. $ D_\alpha $ is the Fourier transform of this function so if we can show that it is of \emph{positive type} then $ D_\alpha $ must be positive by Bochner's theorem \cite{Reed:1975aa}. Define the following function with $ \cen{t_i} $ an arbitrary time partition:
\beq
f_{mn}^\alpha \equiv \sum_{i,j}U_{\alpha i}\mathcal{C}_{ij}(t_m-t_n)U_{\alpha j}^* \ .
\eeq
Note that \bes
\begin{align}
\mc{B}_{\a\b}(t,t-\tau) &\equiv 
\langle B_{\alpha}(t)B_{\beta}(t-\tau)\rangle_{B}=\Tr\bigl(e^{iH_{B}t}B_{\alpha}e^{-iH_{B}t}e^{iH_{B}(t-\tau)}B_{\beta}e^{-iH_{B}(t-\tau)}\rho_{B}\bigr)\\
  &=\Tr\bigl(e^{-iH_{B}(t-\tau)}e^{iH_{B}t}B_{\alpha}e^{-iH_{B}t}e^{iH_{B}(t-\tau)}B_{\beta}\rho_{B}\bigr)=
  \Tr\bigl(e^{iH_{B}\tau}B_{\alpha}e^{-iH_{B}\tau}B_{\beta}\rho_{B}\bigr)=\langle B_{\alpha}(\tau) B_{\beta}\,\rangle_{B} \\ 
  & = \mc{B}_{\a\b}(\tau,0) \equiv \mc{B}_{\a\b}(\tau)\ ,
  \label{eq:477c}
\end{align}
\ees
	Now use the property $ \expv{B_\alpha(s)B_\beta(0)} = \expv{B_\alpha(t)B_\beta(t-s)} $ [Eq.~\eqref{eq:477c}] to write $ f_{mn}^\alpha $ as
\beq
 f_{mn}^\alpha = \sum_{i,j}U_{\alpha i}\Tr\ben{\rho_BB_i(t_m)B_j(t_n)}U_{\alpha j}^* = \Tr\pen{\rho_B\sum_iU_{\alpha i}B_i(t_m)\sum_jU_{\alpha j}^*B_j(t_n)} \ .
 \eeq
	We need to show that $ f^\alpha $ is a positive matrix. For arbitrary $ \ket{v} $ we have
\bes
\begin{align}
\bra{v}f^\alpha\ket{v} &= \sum_{m,n}v_m^*v_nf_{mn}^\alpha = \Tr\ben{\pen{\sum_{i,m}v_m^*U_{\alpha i}\sqrt{\rho_B}B_i(t_m)}\pen{\sum_{j,n}v_nU_{\alpha j}^*B_j(t_n)\sqrt{\rho_B}}} \\
& = 
\Tr\pen{M_\alpha^\dagger M_\alpha} \geq 0 \ ,
\end{align}
\ees
	where $M_\alpha \equiv \sum_{i,m}v_m^*U_{\alpha i}\sqrt{\rho_B}B_i(t_m)$.
	
We have established that $ \bra{v}f^\alpha\ket{v} \geq 0 $ for any time partition $ \cen{t_i} $. Therefore $ D_\alpha $ is non-negative by Bochner's theorem. Consequently, $ \gamma\geq 0$ as a matrix since all its eigenvalues are non-negative.

\section{Lamb shift simplification}
\label{app:LS-simp}
The Lamb shift was given in Eq.~\eqref{eq:HLS1} as
$    H_{\mathrm{LS}}=\frac{i}{2T_a}\int_{-T_a/2}^{T_a/2}dt_1\int_{-T_a/2}^{T_a/2}  dt_2 \sgn(t_1-t_2)\mathcal{C}(t_2-t_1)  A(t_2) A(t_1)$. Recall also that 
$A(t) = \sum_\omega A_{\omega}e^{-i\omega t}$ [Eq.~\eqref{eq:A(t)}], 
so we can write:
\begin{equation}
    H_{\mathrm{LS}}=\sum_{\omega_1\omega_2}F_{\omega_1\omega_2}A_{\omega_2}A_{\omega_1}\ ,
\label{eq:HLS2}
\end{equation}
where 
\begin{subequations}
\begin{align}
F_{\omega_1\omega_2}&= \frac{i}{2T_a}\int_{-T_a/2}^{T_a/2}dt_1\int_{-T_a/2}^{T_a/2}  dt_2 \sgn(t_1-t_2)\mathcal{C}(t_2-t_1)  e^{-i(\omega_2t_2 + \omega_1 t_1)}\\
&=\frac{i}{2T_a}\int_{-T_a/2}^{T_a/2}dt_1\int_{-T_a/2}^{T_a/2}  dt_2 \sgn(\theta_-)\mathcal{C}(-\theta_-)  e^{-i(\omega_-\theta_- + \omega_+ \theta_+)}\ ,
\end{align}
\end{subequations}  
where $\omega_\pm = \frac{1}{2}(\omega_1 \pm \omega_2)$ and  $\theta_\pm = t_1 \pm t_2$. This change of variables rotates the square by $\pi/4$ into a rhombus, and the integration limits correspondingly change to:
\bes
\begin{align}
      &\theta_->0: ~ -T_a+\theta_-\leq   \theta_+\leq T_a -\theta_-\\  
      &\theta_-<0: ~ -T_a-\theta_-\leq   \theta_+\leq T_a +\theta_- \ .
\end{align}
\ees
The Jacobian is $1/2$. 
We can now perform the integral over $\theta_+$: 
\bes
\begin{align}
    F_{\omega_1\omega_2}  &= \frac{i}{4T_a}\left( -\int_{-T_a}^0 d\theta_- \int_{ -T_a-\theta_-}^{T_a +\theta_- }d\theta_+  +  \int_{0}^{T_a} d\theta_- \int_{-T_a+\theta_-}^{T_a -\theta_-}d\theta_+\right) \mathcal{C}(-\theta_-)  e^{-i(\omega_-\theta_- + \omega_+ \theta_+)}\\
    &=\frac{1}{4T_a\o_+}\left( \int_{-T_a}^0 d\theta_- ( e^{-i\omega_+(T_a+\theta_-)} -e^{i\omega_+(T_a+\theta_-)} )  -  \int_{0}^{T_a} d\theta_-( e^{-i\omega_+(T_a-\theta_-)} -e^{i\omega_+(T_a-\theta_-)}) \right) \mathcal{C}(-\theta_-)  e^{-i\omega_-\theta_- }\\
     &= \frac{1}{4T_a\o_+}\left( \int_{-T_a}^0 d\theta (e^{-iT_a\o_+-i\omega_1\theta } -e^{iT_a\o_++i\omega_2\theta} )\mathcal{C}(-\theta)  -  \int_{0}^{T_a} d\theta( e^{-iT_a\o_++i\omega_2\theta} -e^{iT_a\o_+-i\omega_1\theta}) \mathcal{C}(-\theta)\right)   \\
     &= \frac{1}{4T_a\o_+}\left( \int_0^{T_a} d\theta (e^{-iT_a\o_++i\omega_1\theta } -e^{iT_a\o_+-i\omega_2\theta} )\mathcal{C}(\theta)  -  \int_{0}^{T_a} d\theta( e^{-iT_a\o_++i\omega_2\theta} -e^{iT_a\o_+-i\omega_1\theta}) \mathcal{C}^*(\theta)\right)   \\
     &= \frac{1}{4T_a\o_+} \int_0^{T_a} d\theta (e^{-iT_a\o_++i\omega_1\theta } -e^{iT_a\o_+-i\omega_2\theta} )\mathcal{C}(\theta)  +  (  e^{-iT_a\o_++i\omega_1\theta}-e^{iT_a\o_+-i\omega_2\theta})^* \mathcal{C}^*(\theta)  
\end{align}
\ees
where 
we used the identity $\mathcal{C}(\theta) = \mathcal{C}^*(-\theta)$ and 
a change of variables $\theta \mapsto -\theta$ to make all the integration domains positive. When $\omega_+ = 0$ the expression above should be understood as a lim$_{\omega_+ \to 0}$. 
Thus:
\begin{equation}
    F_{\omega_1\omega_2}  =\frac{1}{2T_a\o_+} \Re  \int_{0}^{T_a} d\theta \left( e^{i(\o_1\theta-T_a\o_+)} - e^{-i(\o_2\theta-T_a\o_+)} \right) \mathcal{C}(\theta)\ ,
    \label{eq:Fw1w2-copy}
\end{equation}
which is Eq.~\eqref{eq:Fw1w2} in the main text.

Recall that $A_\omega^\dag = A_{-\omega}$; 
for $H_{\mathrm{LS}}= \sum_{\omega_1\omega_2}F_{\omega_1\omega_2}A_{\omega_2}A_{\omega_1}$ to be Hermitian, it is thus sufficient if
\begin{equation}
    F_{\omega_1\omega_2} = F_{\omega_1\omega_2}^* = F_{-\omega_2,-\omega_1}\ .
\end{equation}
This is indeed satisfied by Eq.~\eqref{eq:Fw1w2-copy}, thus ensuring the Hermiticity.


\section{The Davies-Lindblad equation is the infinite coarse-graining time limit of the CGME}
\label{app:D}

The RWA used to derive the Davies-Lindblad equation leaves something to be desired. We simply dropped terms with different Bohr frequencies, without a rigorous mathematical justification. Here we repeat, with some extra clarifications, the derivation given in Ref.~\cite{Majenz:2013qw} which shows that the Davies-Lindblad equation is the infinite coarse-graining timescale limit of the CGME. 

%
First we rewrite Eq.~\eqref{eq:18} in the form
\begin{align}
\label{eq:gammaww'}
\gamma_{\omega \omega'} (T_a)  =  \frac{1}{T_a}  \int_0^{T_a} d s \int_0^{T_a} ds'  e^{i ( \omega' s - \omega s' )}  \mathcal{C}(s-s')  \ ,
 \end{align}
 which we can do after changing variables $\omega' \to -\omega'$ as remarked below Eq.~\eqref{eq:CGME}.
Our goal is to show that $\lim_{T_a\to\infty} \gamma_{\omega \omega'} (T_a) = \gamma(\omega)\delta_{\o\o'}$. 

\begin{mylemma}
The following equivalent form holds for $\g_{\omega \omega'} (T_a)$:
\beq
\g_{\omega \omega'} (T_a) = \frac{1}{T_a}e^{i\frac{\omega'-\omega}{2}T_a}\int_0^{T_a}dv\ \cos\left(\frac{\omega'-\omega}{2}(v-T_a)\right)\int_{-v}^{v}du\ e^{i\frac{\omega+\omega'}{2}u}\mathcal{C}(u) \ .
\label{eq:g-simplified}
\eeq
\end{mylemma}

\begin{proof}
Let us rewrite $\omega' s - \omega s'$ in terms of a sum and difference of Bohr frequencies: 
\beq
\omega' s - \omega s' = \frac{1}{2}(\o'-\o)v+\frac{1}{2}(\o'+\o)u\ ,
\eeq
where $u=s-s'$ and $v=s+s'$. After this change of variables $\mathcal{C}(s-s') = \mathcal{C}(u)$, and since $s=(v+u)/2$ and $s'=(v-u)/2$, the Jacobian of the transformation is $1/2$. In terms of the new variables the integration domain is a rhombus, bounded between the lines $u=v$ and $u=-v$ for $v\in [0,T_a]$ and the lines $u=2T_a-v$ and $v-2T_a$ for $v\in [T_a,2T_a]$.
Thus:
\begin{align}
\g_{\omega \omega'}(T_a)= \frac{1}{2T_a}\int_0^{T_a}dv\ e^{i\frac{\omega'-\omega}{2} v}\int_{-v}^{v}du\ e^{i\frac{\omega+\omega'}{2}u}\mathcal{C}(u)+\frac{1}{2} \int_{T_a}^{2 T_a}dv\ e^{i\frac{\omega'-\omega}{2} v}\int_{-(2T_a-v)}^{2 T_a-v}du\ e^{i\frac{\omega+\omega'}{2}u}\mathcal{C}(u) \ .
\end{align}
To get the integration limits to be the same we make a change of variables from $v$ to $2T_a-v$ in the second double integral:
\bes
\begin{align}
\g_{\omega \omega'}(T_a)&=\frac{1}{2T_a}\int_0^{T_a}dv\ e^{i\frac{\omega'-\omega}{2} [(v-T_a)+T_a]}\int_{-v}^{v}du\  e^{i\frac{\omega+\omega'}{2}u}\mathcal{C}(u)+ \frac{1}{2} \int_{0}^{T_a}dv\ e^{-i\frac{\omega'-\omega}{2} [(v-T_a)-T_a]}\int_{-v}^{v}du\ e^{i\frac{\omega+\omega'}{2}u}\mathcal{C}(u)\\
&=\frac{1}{T_a}e^{i\frac{\omega'-\omega}{2}T_a}\int_0^{T_a}dv\ \cos\left(\frac{\omega'-\omega}{2}(v-T_a)\right)\int_{-v}^{v}du\ e^{i\frac{\omega+\omega'}{2}u}\mathcal{C}(u) \ ,
\end{align}
\ees
which is Eq.~\eqref{eq:gammaww'}.
\end{proof}

\subsection{The $\o=\o'$ case}

For $\omega=\omega'$ we now have:

\begin{equation}
 \g_{\omega \omega}(t)=\frac{1}{T_a}\int_0^{T_a} dv\ \int_{-v}^{v}du\ e^{i\omega u}\mathcal{C}(u) \ .
\end{equation}
Let $U = \int_{-v}^{v}du\ e^{i\omega u}\mathcal{C}(u)$. 
%
Then $dU = \left(e^{i\o v} \mathcal{C}(v)+e^{-i\o v} \mathcal{C}(-v)\right)dv$, and integrating by parts 
gives:
\begin{align}
\g_{\omega\omega}(T_a)&=\frac{1}{T_a}\left[v \int_{-v}^{v}du\ e^{i\omega u}\mathcal{C}(u)\right]_0^{T_a}-\frac{1}{T_a}\int_0^{T_a}dv\ v\left( e^{i\omega v}\mathcal{C}(v)+e^{-i\omega v}\mathcal{C}(-v)\right)\ .
\label{eq:539}
\end{align}
Consider the second integral:
\bes
\begin{align}
\left| \frac{1}{T_a} \int_0^{T_a}dv\ v e^{i\omega v}\mathcal{C}(v) \right| &\leq \frac{1}{T_a} \int_0^{T_a}dv\ v \left| \mathcal{C}(v) \right| \leq \frac{1}{T_a}\int_0^{\infty}dv\ v \left| \mathcal{C}(v) \right| \\
&= \frac{\tau_B}{\tau_{SB} T_a}  \stackrel{T_a\to\infty}{\longrightarrow} 0\ ,
\end{align}
\ees
where in the last step we used Eq.~\eqref{eq:T1tauB-b}. Since $\mathcal{C}(v) =\mathcal{C}^*(-v)$, the third integral in Eq.~\eqref{eq:539} satisfies the same bound and limit. We are thus left with
\beq
\lim_{T_a \to \infty} \g_{\o\o}(T_a) =\int_{-\infty}^{\infty}du\ e^{i\omega u}\mathcal{C}(u)=\gamma(\omega) \  ,
\eeq
with the last equality being due to Eq.~\eqref{eq:gamma31}.

\subsection{The $\o\neq\o'$ case}
For $\omega\neq\omega'$ we also perform integration by parts of Eq.~\eqref{eq:g-simplified}, but we shall see that this time the boundary terms vanish. 
We write 
$\g_{\omega \omega'} (T_a) = \frac{1}{T_a}e^{i\frac{\omega'-\omega}{2}T_a} \int_0^{T_a}dV\  U(v)$,
where now $dV = \cos\left(\frac{\omega'-\omega}{2}(v-T_a)\right)dv$ and $U(v) = \int_{-v}^{v}du\ e^{i\frac{\omega+\omega'}{2}u}\mathcal{C}(u)$. Then 
\bes
\begin{align}
V(v) &= \frac{2}{\omega'-\omega}\sin\left(\frac{\omega'-\omega}{2}(v-T_a)\right) \\
dU/dv &= e^{i\frac{\omega+\omega'}{2}v} \mathcal{C}(v)+e^{-i\frac{\omega+\omega'}{2}v} \mathcal{C}(-v)\\
\left[U(v) V(v) \right]_0^{T_a} &= U(T_a) V(T_a) - U(0)V(0) = 0 \ .
\label{eq:695c}
\end{align}
\ees
Therefore:
\begin{equation}
\g_{\omega \omega'}(T_a)=-\int_{0}^{T_a}V dU = -\frac{2 e^{i\frac{\omega'-\omega}{2}T_a}}{(\omega'-\omega)T_a}\int_0^{T_a}dv \sin\left( \frac{(\omega'-\omega)}{2}(v-T_a)\right)\left[e^{i\frac{\omega+\omega'}{2}v}\mathcal{C}(v)+e^{-i\frac{\omega+\omega'}{2}v}\mathcal{C}(-v)\right] \ .
\end{equation}
Changing from $v$ to $-v$ in the second term we get
\bes
\begin{align}
\g_{\omega \omega'}(T_a)=&-\frac{2 e^{i\frac{\omega'-\omega}{2}T_a}}{(\omega'-\omega)T_a}\Bigg[\int_0^{T_a}dv\sin\left(\frac{(\omega'-\omega)}{2}(v-T_a)\right)e^{i\frac{\omega+\omega'}{2}v}\mathcal{C}(v)+\int_{-T_a}^{0}dv\sin\left(\frac{(\omega'-\omega)}{2} (-v-T_a)\right)e^{i\frac{\omega+\omega'}{2}v}\mathcal{C}(v)\Bigg]\\
=&\frac{e^{i\frac{\omega'-\omega}{2}T_a}}{(\omega'-\omega)T_a}\int_{-T_a}^{T_a}dv\left[\sin\left(\frac{\omega'-\omega}{2}T_a\right)\left(e^{i \omega v}+e^{i \omega' v}\right)+\frac{\mathrm{sgn(v)}}{i}\cos\left(\frac{\omega'-\omega}{2}T_a\right)\left(e^{i \omega v}-e^{i \omega' v}\right)\right]\mathcal{C}(v) \ ,
\end{align}
\ees
where we used the angle sum identity for the sine in the last equality. Thus:
\beq
\lim_{T_a \to \infty} \g_{\omega \omega'}(T_a) = \lim_{T_a \to \infty} \frac{e^{i\frac{\omega'-\omega}{2}T_a}}{(\omega'-\omega)T_a}\left[\sin\left(\frac{\omega'-\omega}{2}T_a\right)(\gamma(\omega)+\gamma(\omega'))+2\cos\left(\frac{\omega'-\omega}{2}T_a\right)(S(\omega)-S(\omega'))\right] \ .
\eeq
where $S(\omega)$ is the Lamb shift amplitude in Eq.~\eqref{eq:LS-final}.  Since nothing cancels with the overall $T_a^{-1}$, we find that the $\omega \neq \omega'$ term vanishes.  

A similar calculation could be done for the Lamb shift term~\eqref{eq:HLS1}, showing that $\Im(x_{\o\o'}) \to \Im(x_{\o\o})\delta_{\o\o'}$.  Therefore, the Davies-Lindblad (RWA) equation can be understood as the $T_a\to\infty$ limit of the CGME.


\section{Derivation of the discrete approximation Eq.~\eqref{eq:discrete-approx}}
\label{app:C}

The starting point is the time-independent CGME given in Eq.~\eqref{eq:Lind}.
We can obtain a discrete sum just by discretizing $\int d\epsilon$. Below is the estimate of the number of terms  $(2k^*-1)$ in the sum for that case. The error is:
\begin{equation}
    \mathcal{E}_k = \int_{-\infty}^\infty d\epsilon (A_{\epsilon} \rho_C A_{\epsilon}^\dag - \frac{1}{2}\left\{ \rho_C , A_{\epsilon}^\dag A_{\epsilon}\right\} )-\Delta \epsilon \sum_{k, \epsilon = \Delta \epsilon k, ~ |k|<k^*} (A_{\epsilon} \rho_C A_{\epsilon}^\dag - \frac{1}{2}\left\{ \rho_C , A_{\epsilon}^\dag A_{\epsilon}\right\} )\ ,
\end{equation}
where $A_\epsilon$ is given in Eq.~\eqref{Atime}.
The discretization of an integral introduces errors as follows:
\begin{equation}
    \left\|\int d\epsilon f(\epsilon)-\Delta \epsilon \sum_{k, \epsilon = \Delta \epsilon k, ~ |k|<k^*}f(\epsilon) \right\|_1 \leq \frac{k^*-0.5}{2} \max\|f'\|_1\Delta \epsilon^2 +\left(\int_{-\infty}^{-\Delta \epsilon (k^*-0.5) } + \int_{\Delta \epsilon (k^*-0.5)}^\infty\right) d\epsilon \|f(\epsilon)\|_1 \ ,
\label{eq:C2}
\end{equation}
where $f$ is an operator or operator-valued function.
\begin{proof}
By the triangle inequality:
\bes
\begin{align}
    \left\|\int d\epsilon f(\epsilon)-\Delta \epsilon \sum_{k, \epsilon = \Delta \epsilon k, ~ |k|<k^*}f(\epsilon) \right\|_1 \leq \left\|\int_{-\Delta \epsilon (k^*-0.5) }^{\Delta \epsilon (k^*-0.5) } d\epsilon f(\epsilon)-\Delta \epsilon \sum_{k, \epsilon = \Delta \epsilon k, ~ |k|<k^*}f(\epsilon) \right\|_1 
    \\+\left(\int_{-\infty}^{-\Delta \epsilon (k^*-0.5) } + \int_{\Delta \epsilon (k^*-0.5)}^\infty\right) d\epsilon \|f(\epsilon)\|_1 \ .
\end{align}
\ees
We apply the triangle inequality to pull the sum out:
\begin{equation}
     \left\|\int_{-\Delta \epsilon (k^*-0.5) }^{\Delta \epsilon(k^*-0.5) } d\epsilon f(\epsilon)-\Delta \epsilon \sum_{k, \epsilon = \Delta \epsilon k, ~ |k|<k^*}f(\epsilon) \right\|_1 \leq \sum_{k, \epsilon = \Delta \epsilon k, ~ |k|<k^*} \left\|\int_{\epsilon-\Delta \epsilon/2 }^{\epsilon +\Delta \epsilon/2 } d\mu f(\mu)-\Delta \epsilon f(\epsilon) \right\|_1\ .
\end{equation}
To bound this, we use $O(t) = O(0) + \int_0^t O'(\theta)d\theta$, so that $\int_{-1/2}^{1/2} O(t)dt - O(0)  = \int_{-1/2}^{1/2} \int_0^t O'(\theta)d\theta dt$, which implies
\begin{equation}
\left\|\int_{-1/2}^{1/2} O(t)dt - O(0)\right\|_1\leq \frac{1}{4}\max_{\theta\in [0,1]}\|O'\|_1\ ,
\label{eq:C3}
\end{equation}
which applied to the interval $[\epsilon, \epsilon +\Delta \epsilon]$ reads:
\begin{equation}
\left\|\int_{\epsilon-\Delta \epsilon/2 }^{\epsilon +\Delta \epsilon/2 } d\mu f(\mu) - f(\epsilon)\right\|_1\leq \frac{\Delta \epsilon^2}{4}\max_{\mu\in [\epsilon, \epsilon +\Delta \epsilon]}\|f'\|_1\ .
\end{equation}
Substituting this yields:
\begin{equation}
     \left\|\int_{-\Delta \epsilon (k^*-0.5) }^{\Delta \epsilon (k^*-0.5) } d\epsilon f(\epsilon)-\Delta \epsilon \sum_{k, \epsilon = \Delta \epsilon k, ~ |k|<k^*}f(\epsilon) \right\|_1 \leq (2k^* -1) \frac{\Delta \epsilon^2}{4}\max_{\mu\in [-\Delta \epsilon k^*, \Delta \epsilon k^*]}\|f'\|_1\ ,
\end{equation}
which concludes the proof.
\end{proof}
If we use the proven Eq.~\eqref{eq:C3} for
\begin{equation}
    f_\epsilon = (A_{\epsilon} \rho_C A_{\epsilon}^\dag - \frac{1}{2}\left\{ \rho_C , A_{\epsilon}^\dag A_{\epsilon}\right\} )\ , \label{fDef}
\end{equation}
it will express the error $\mathcal{E}_k$ as the sum of:
\begin{equation}
    \mathcal{E}_k =\mathcal{E}_{k1} +\mathcal{E}_{k2} ,\quad \mathcal{E}_{k1} = \frac{k^*-0.5}{2} \max\|f'\|_1\Delta \epsilon^2, \quad  \mathcal{E}_{k1} =\left(\int_{-\infty}^{-\Delta \epsilon (k^*-0.5) } + \int_{\Delta \epsilon (k^*-0.5)}^\infty\right) d\epsilon \|f(\epsilon)\|_1\ .
\end{equation}
Let us start with the $\mathcal{E}_{k2}$ term  corresponding to truncating the integral at $k^*$.
As $\epsilon$ becomes far removed from all the frequencies $\omega$ in the system, this integral becomes small as $1/\epsilon$. To prove this, we integrate by parts:
\begin{equation}
    A_\epsilon =  \sqrt{\frac{\gamma(\epsilon)}{2\pi T_a}}\int_{-T_a/2}^{T_a/2}A(t) e^{i\epsilon t} dt = \sqrt{\frac{\gamma(\epsilon)}{2\pi T_a}}\frac{1}{i\epsilon} \left(A(T_a/2) e^{i\epsilon T_a/2} - A(-T_a/2) e^{-i\epsilon T_a/2} +i\int_{-T_a/2}^{T_a/2}[H,A(t)] e^{i\epsilon t}dt  \right)\ .
\end{equation}
Taking the norm and recalling that $\| A\|=1$, we get:
\begin{equation}
    \|A_\epsilon\| \leq \sqrt{\frac{\gamma(\epsilon)}{2\pi T_a}} \frac{2 + \|[H,A]\| T_a}{|\epsilon|} \leq \sqrt{\frac{1}{\pi \tau_{SB} T_a}}\frac{2 + \|[H,A]\|T_a}{|\epsilon|}\ .
\end{equation}
Note that for a local $n$ qubit Hamiltonian with $O(1)$ local terms and a local operator $A$ with $\|A\| =1$, the commutator $\|[H,A]\| = O(1)$. Combining, we find:
\bes
\begin{align}
    \|\mathcal{E}_{k2}\|_1 &\leq \left(\int^{-\Delta\epsilon (k^*-0.5)}_{-\infty}+\int_{\Delta\epsilon (k^*-0.5)}^{\infty}\right) \|A_\epsilon\|^2 d\epsilon \leq 2 \int_{\Delta\epsilon (k^*-0.5)}^{\infty}\frac{1}{\pi T_a\tau_{SB}}\frac{(2 + \|[H,A]\| T_a)^2}{|\epsilon|^2}d\epsilon \\
    &= 2 \frac{1}{\pi T_a\tau_{SB}}\frac{(2 + \|[H,A]\| T_a)^2}{\Delta\epsilon (k^*-0.5)}\ .
\end{align}
\ees
We enforce that the error is $\sqrt{\tau_B/\tau_{SB}} \frac{1}{\tau_{SB}}$, so as to match the other errors in our equation:
\begin{equation}
\sqrt{\frac{\tau_B}{\tau_{SB}}} \frac{1}{\tau_{SB}} =   2 \frac{1}{\pi T_a\tau_{SB}}\frac{(2 + \|[H,A]\| T_a)^2}{\Delta\epsilon (k^*-0.5)} \quad \Rightarrow\quad  \Delta\epsilon (k^*-0.5) = 2 \frac{1}{\pi T_a}\sqrt{\frac{\tau_{SB}}{\tau_B}}(2 + \|[H,A]\| T_a)^2\ .
\label{eq:C8}
\end{equation}

Let us now focus on the term with the derivatives $\mathcal{E}_{k1}$:
\begin{equation}
    \|\mathcal{E}_{k1}\|_1 \leq 2(k^*-0.5) \max\|A_\epsilon'\|\|A_\epsilon\| \Delta \epsilon^2\ .
\end{equation}
The necessary norm bounds are:
\begin{equation}
     \|A_\epsilon\| \leq \sqrt{\frac{\gamma(\epsilon) T_a}{2\pi}} ,\quad \|A_\epsilon' \| \leq\frac{1}{2}\sqrt{\frac{|\gamma'(\epsilon)| T_a}{2\pi\gamma(\epsilon)}} + \frac{T_a}{4}\sqrt{\frac{\gamma(\epsilon) T_a}{2\pi}}\ .
\end{equation}
Using Eqs.~\eqref{eq:22new} and~\eqref{eq:23} we have the bound
\begin{equation}
     \|A_\epsilon\| \|A_\epsilon' \| \leq\frac{1}{2}\frac{\sqrt{|\gamma'(\epsilon)|} T_a}{2\pi} + \frac{T_a}{4}\frac{\gamma(\epsilon) T_a}{2\pi}  \leq \frac{1}{2}\frac{\sqrt{2\tau_B/\tau_{SB}} T_a}{2\pi} + \frac{1}{4}\frac{ T_a^2}{2\pi \tau_{SB}}\ .
\end{equation}
Substituting $T_a = \sqrt{\tau_B \tau_{SB}/5}$, we get:
\begin{equation}
     \|A_\epsilon\| \|A_\epsilon' \| \leq\frac{\sqrt{2/5} \tau_B}{4\pi} + \frac{ \tau_B}{40\pi} = \frac{(10\sqrt{2/5}+1) \tau_B}{40\pi }\ ,
\end{equation}
and the error bound is then:
\begin{equation}
    \|\mathcal{E}_{k1}\|_1 \leq 2(k^*-0.5) \frac{(10\sqrt{2/5}+1) \tau_B}{10\pi }\Delta \epsilon^2\ .
\end{equation}
This leads to the step-size choice:
\begin{equation}
    2(k^*-0.5) \frac{(10\sqrt{2/5}+1) \tau_B}{10\pi }\Delta \epsilon^2 = \frac{1}{\tau_{SB}}\sqrt{\frac{\tau_B}{\tau_{SB}}}, \quad \Delta\epsilon = \frac{1}{\tau_{SB}\sqrt{\tau_B \tau_{SB}} (k^*-0.5)\Delta \epsilon}\frac{5\pi }{(10\sqrt{2/5}+1) }\ .
\end{equation}
Combining this with Eq.~\eqref{eq:C8}, the solution is:
\begin{equation}
    \Delta\epsilon =\frac{1}{\tau_{SB}}\sqrt{\frac{\tau_B}{\tau_{SB}}}\frac{1}{(2 + \|[H,A]\| T_a)^2}\frac{\sqrt{5}\pi^2 }{2(10\sqrt{2/5}+1) }, \quad k^* = \text{ceil}\left( \sqrt{\frac{\tau_{SB}}{\tau_B}}^3(2 + \|[H,A]\| T_a)^4\frac{4(10\sqrt{2/5}+1) }{\pi^3 } +0.5\right) \ .
    \label{eq:C15}
\end{equation}
If $[H,A]=O(1)$, there is no system-size dependence in $\Delta \epsilon, k^*$. 

We remark that the bounds we used to derive Eq.~\eqref{eq:C15} can be significantly tightened; our goal here was to show that they are $O(1)$. We do not actually recommend to use Eq.~\eqref{eq:C15}  for numerical calculations: one can obtain accurate results with much larger $\Delta \epsilon$ and much smaller $k^*$.


%


\section{An error bound that is linear in $t$}
\label{LemmaTwo}
\begin{mylemma}
\label{LemLinear}
Let $\mathcal{L}$ be a linear superoperator and consider the pair of equations
\begin{equation}
    \dot{x}(t) = \mathcal{L}x(t) +\mathcal{E}, \quad  \dot{y}(t) = \mathcal{L}y(t)\ .
\end{equation}
If any solution $y(t)$ possesses the property that $\forall t, y(0)$ such that $\|y(0)\|_p\leq 1: \|y(t)\|_p\leq c_y$, then for any finite $t$:
\begin{equation}
    \|x(t)-y(t)\|_p \leq c_y t\|\mathcal{E}\|_p\ .
\end{equation}
The result holds in the same form for the operator ($p=\infty$, omitted in our notation) and trace ($p=1$) norms.
\end{mylemma}

\begin{proof}
The linear differential equation for $y$ has a unique solution in a neighborhood of any initial condition. A trivial corollary is that the evolution operator is reversible for any finite $t$:
\begin{equation}
  e^{\mathcal{L}t}e^{-\mathcal{L}t} =\openone \ .
\end{equation}
This is an equality between superoperators. Now we can express the difference as:
\begin{equation}
    x(t) -y(t) =e^{\mathcal{L}t}(e^{-\mathcal{L}t}x(t) - y(0)) \ .
    \label{eq:119}
\end{equation}
The explicit form of the solution $x(t)$ is:
\begin{equation}
    x(t) = e^{\mathcal{L}t}x(0) + e^{\mathcal{L}t}\int_0^t e^{-\mathcal{L}\tau}\mathcal{E} d\tau\ ,
\end{equation}
which yields, using Eq.~\eqref{eq:119}:
\begin{equation}
    x(t) - y(t) =e^{\mathcal{L}t} \int_0^t e^{-\mathcal{L}\tau}\mathcal{E} d\tau = \int_0^t e^{\mathcal{L}(t- \tau)}\mathcal{E} d\tau  \ .
\end{equation}
We now take the norm of both sides:
\begin{equation}
    \|x(t) - y(t)\|_p \leq  t \|e^{\mathcal{L}(t- \tau)}\mathcal{E} \|_p \ .
    \label{eq:122}
\end{equation}
By linearity of the time evolution and the norm:
\begin{equation}
    \|e^{\mathcal{L}(t- \tau)}\mathcal{E} \|_p  = \|\mathcal{E}\|_p \|e^{\mathcal{L}(t- \tau)}y \|_p, \quad y = \frac{\mathcal{E}}{\|\mathcal{E}\|_p}
\end{equation}
Note that $\|y\|_p=1$. Since we can take $y$ as the initial condition of the equation $\dot{y}(t) = \mathcal{L}y(t)$, and since $e^{\mathcal{L}(t- \tau)}y$ will then be the solution at time $t-\tau$, we have, by assumption of the lemma:
\begin{equation}
    \|e^{\mathcal{L}(t- \tau)}y \|_p \leq c_y\ .
\end{equation}
Note that we had to make these assumptions for any $y$ of bounded norm, not just density operators. It now follows from Eq.~\eqref{eq:122} that:
\begin{equation}
    \|x(t) - y(t)\|_p \leq c_y t \|\mathcal{E}\|_p\ .
\end{equation}
\end{proof}
In particular, if $\mc{L}$ is a positive trace-preserving map then $c_y=1$. 

The only error for which we can apply Lemma~\ref{LemLinear} in this work is the Markov error, which we present here:
\begin{equation}
     \|\rho^{BM}_t-\rho^{B}_t\|_1  \leq 16 \frac{\tau_B c_B c_{BM} t}{\tau_{SB}^2} \ ,
\end{equation}
where $c_y =c_{BM}$ and we used Eq.~\eqref{eq:E_M}, which contains $c_{B}$. We instead used a loose bound from Eq.~\eqref{cBMarkov} since it results in a more compact expression.
Apart from this we have no occasion to use Lemma~\ref{LemLinear} in this work, since not all the equations appearing in our derivation are invertible. This is why we instead use Lemma~\ref{factDif} or prove the necessary results independently.

\section{Proof of Lemma \ref{LemTave}}
\label{LemProof}
\begin{proof}
The formal solutions allow us to write:
\begin{equation}
\delta\rho(t)=   \int_0^t \overline{\mathcal{L}_\theta}\pi(\theta)  d\theta - \int_0^t  \mathcal{L}_{\theta}\rho(\theta)d\theta \ .
\end{equation}
Substituting $\pi(t) = \rho(t) + \delta\rho(t)$ gives:
\begin{equation}
 \delta\rho(t)= \int_0^t (\overline{\mathcal{L}_\theta}\rho(\theta) -\mathcal{L}_{\theta}\rho(\theta))  d\theta + \int_0^t \overline{\mathcal{L}_\theta}\delta\rho(\theta)  d\theta\ .
\end{equation}
Making the time-averaging explicit:
\begin{equation}
 \delta\rho(t)= \frac{1}{T_a}\int_0^t d\theta \int_{-T_a/2}^{T_a/2} d\tau [\mathcal{L}_{\theta+\tau}\rho(\theta) -\mathcal{L}_{\theta}\rho(\theta)]   + \int_0^t \overline{\mathcal{L}_\theta}\delta\rho(\theta)  d\theta\ .
 \label{eq:113}
\end{equation}
We would like to change the integration variables of the first term $\mathcal{L}_{\theta+\tau}\rho(\theta)$ in such a way that  it becomes $\mathcal{L}_{\theta'}\rho(\theta' -\tau')$. The change of variables is $\theta + \tau = \theta', ~ \tau = \tau'$. In the second term $-\mathcal{L}_{\theta}\rho(\theta)$ the change is trivial $\theta =\theta', ~ \tau = \tau'$. This introduces transient effects at the beginning and the end of the evolution, as we want to revert the integration region of the first term $\mathcal{L}_{\theta'}\rho(\theta' -\tau')$ to be the same as the second term. This is illustrated in Fig.~\ref{timeAvtrick} and the corresponding integrals are:
\begin{figure}[t]
\includegraphics[width=0.8\linewidth]{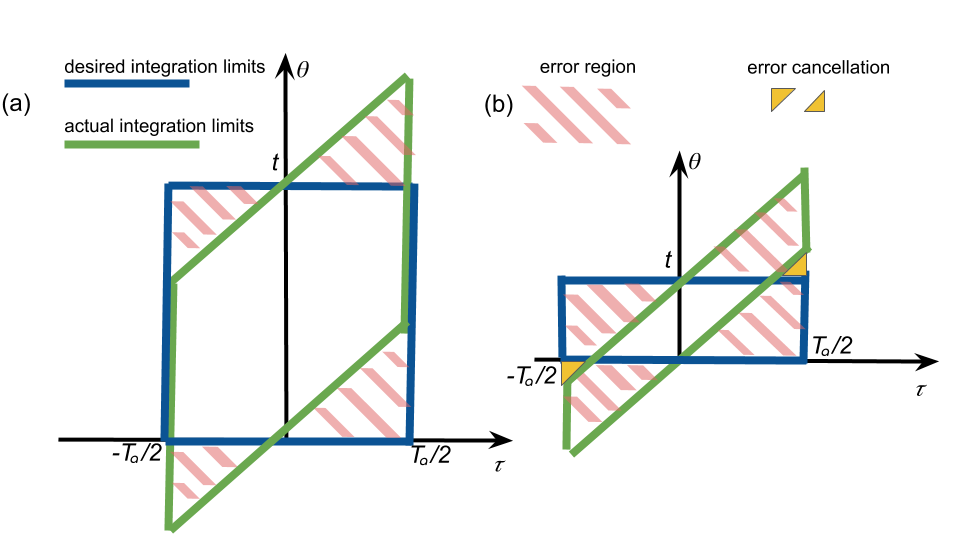}
\caption{Times contributing to $\delta \rho(t)$ for (a) $t>T_a/2$ (b) $t<T_a/2$.  }
\label{timeAvtrick}
\end{figure}
\begin{equation}
   \int_0^t d\theta \int_{-T_a/2}^{T_a/2}d\tau =\int_{-T_a/2}^{T_a/2}d\tau'\int_{\tau'}^{t+\tau'}d\theta'  = \int_{-T_a/2}^{T_a/2}d\tau' \left(\int_0^{t}d\theta' + \int_t^{t+\tau'}d\theta' -\int_0^{\tau'}d\theta' \right)
\end{equation}
After this change of variables and collecting all the contributions from the $\mathcal{L}_{\theta}\rho(\theta)$ term in Eq.~\eqref{eq:113} into a single term, we get:
\begin{equation}
  \delta\rho(t)= \frac{1}{T_a}\int_{-T_a/2}^{T_a/2}d\tau' \int_0^{t}d\theta'\mathcal{L}_{\theta'}(\rho(\theta' -\tau')-\rho(\theta')) +{\frac{1}{T_a}\int_{-T_a/2}^{T_a/2}d\tau' \left(\int_t^{t+\tau'}d\theta' -\int_0^{\tau'}d\theta' \right)\mathcal{L}_{\theta'}\rho(\theta' -\tau')} + \int_0^t \overline{\mathcal{L}_\theta}\delta\rho(\theta)  d\theta .
\end{equation}
We call the middle term transient because its contribution does not grow with $t$. One can interpret it as two instantaneous kicks to the solution at time $0$ and $t$. For $t\leq T_a/2$ the integrals in the second term cancel over the area $(T_a/2 - t)^2$, as indicated by the yellow triangles in Fig.~\ref{timeAvtrick}. In fact by splitting the integral above differently we could get a slightly tighter (parametrically better for the transient) bound, but working with rectangular regions is easier. Now we can take the norm: \begin{equation}
 \|\delta\rho(t)\|_1\leq \frac{\Lambda }{T_a}\int_0^t d\theta' \int_{-T_a/2}^{T_a/2} d\tau' \|\rho(\theta'-\tau') -\rho(\theta')\|_1   + {\frac{c_\rho \Lambda }{2T_a}\left(T_a^2 -\left(\max\left(T_a-2t,0\right)\right)^2\right)} + \Lambda \int_{0}^t\|\delta\rho(\theta)\|_1d\theta \ ,
\end{equation}
where $\Lambda$ is a number such that
\beq
\Lambda  \geq \sup_{\theta \in [0,t],\|\rho\|_1= 1} \|\mathcal{L}_{\theta}(\rho)\|_1\ .
\eeq
The difference between $\rho$'s at nearby time points is bounded trivially as:
\begin{equation}
   \dot\rho = \mathcal{L}\rho\quad \Rightarrow \quad \|\rho(\theta' -\tau') - \rho(\theta')\|_1 \leq c_\rho\Lambda |\tau'| \ ,
\end{equation}
which can be used to simplify:
\bes
\begin{align}
 \|\delta\rho(t)\|_1 &\leq \frac{c_\rho\Lambda }{T_a}\int_0^t d\theta' \int_{-T_a/2}^{T_a/2} d\tau' \Lambda |\tau'|   + \frac{c_\rho \Lambda \min(T_a,4t - 4t^2/T_a)}{2} + \Lambda \int_{0}^t\|\delta\rho(\theta)\|_1d\theta   \\ 
 &\leq \frac{c_\rho\Lambda ^2 t T_a}{4}   + \frac{c_\rho\Lambda \min(T_a,4t - 4t^2/T_a)}{2} + \Lambda \int_{0}^t\|\delta\rho(\theta)\|_1d\theta\ .
\end{align}
\ees
We analyzed a similar inequality in the proof of Lemma~\ref{factDif}, and as we saw there this amounts to solving for $b_2(t)$ that saturates the inequality and upper bounds $\delta\rho(t)$. Note that despite the $\min$ function there is no abrupt slope change on the r.h.s. and the derivative matches, so there is no discontinuity in $\dot{b}_2$. The equations for $b_2$ such that $\|\delta\rho(t)\|_1 \leq b_2(t)$ are:  
\bes
\begin{align}
    \dot{b}_2(t) &= \frac{c_\rho\Lambda ^2  T_a}{4}   + c_\rho\Lambda (2 - 4t/T_a) + \Lambda b_2(t), \quad b_2(0)=0,\quad t\leq T_a/2 \\
      \dot{b}_2(t) &= \frac{c_\rho\Lambda ^2  T_a}{4}   + \Lambda b_2(t), \quad t\geq T_a/2 
\end{align}
\ees
Solving the first equation, we find:
\bes
\begin{align}
b_2(t) &= 4c_\rho t/T_a + c_\rho\frac{4-2\Lambda T_a -(\Lambda T_a/2)^2}{\Lambda T_a} (1 -e^{\Lambda t}) \label{dFull}\\
    b_2(T_a/2) &= c_\rho\frac{4- (\Lambda T_a/2)^2}{\Lambda T_a} + c_\rho e^{\Lambda T_a/2}\frac{(\Lambda T_a/2)^2 + 2\Lambda T_a -4}{\Lambda T_a}  = \frac{c_\rho\Lambda T_a }{2}+ O(\Lambda ^2T_a^2)
\end{align}
\ees
where the last estimate is in the limit $\Lambda T_a\ll 1$. Now we use that as the initial condition for the second equation to find:
\begin{equation}
   b_2(t) = [b_2(T_a/2) + c_\rho\Lambda  T_a/4]e^{\Lambda (t-T_a/2)} - c_\rho\Lambda  T_a/4,   \quad t>T_a/2 \ .
\end{equation}
 Notice that $\dot{b}_2>0$ always. We can express a bound $\|\delta\rho(t)\|_1 \leq b_2(t)$ for all $t$ as follows:
\begin{equation}
    \|\delta\rho(t)\|_1\leq (b_2(\min(t,T_a/2)) + c_\rho\Lambda  T_a/4)\max(e^{\Lambda (t-T_a/2)},1) - c_\rho\Lambda  T_a/4\ .
\end{equation}
The simplest observation is that for $t\leq c/\Lambda $:
\begin{equation}
   \|\delta\rho(t)\|_1 \leq O(\Lambda T_a)\ .
\end{equation}
\end{proof}

\section{Derivation of the bound~\eqref{looseBound}}
\label{boundingDetails}

Here we discuss the derivation of the error introduced by time-averaging and restoring complete positivity, starting from Eq.~(\ref{fullWLam}):
\begin{align}
   \|\rho_{BM,I}(t) - {\rho}_{C,I}(t)\|_1\leq c_{BM}\left(  e^{\Lambda \max (t,T_a/2)} \frac{3 T_a}{\tau_{SB}} -\frac{  T_a}{\tau_{SB}}\right) +  e^{\Lambda t}\frac{2T_a}{\tau_{SB}} + \frac{4\tau_B}{\Lambda T_a\tau_{SB}}\max(e^{\Lambda (t-T_a/2 )}-1,0 )\ . 
\end{align}
Similarly to Eq.~\eqref{weakBound}, we loosen the bound by using $e^{-\Lambda T_a/2} \leq 1$. For $t>T_a/2$ and $c_{\Lambda} = 4/\Lambda \tau_{SB}\geq 1$, we obtain:
\begin{align}
   \|\rho_{BM,I}(t) - {\rho}_{C,I}(t)\|_1\leq ((3c_{BM}+2)e^{\Lambda t}-c_{BM}) \frac{ T_a}{\tau_{SB}}   +(e^{\Lambda t}-1 )\frac{c_\Lambda\tau_B}{ T_a}\ .
   \label{startApp1}
\end{align}
At $T_a =\sqrt{c_\Lambda\tau_B\tau_{SB}/(3c_{BM}+2)}$ for $t>T_a/2$ we have:
\begin{align}
   \|\rho_{BM,I}(t) - {\rho}_{C,I}(t)\|_1\leq \sqrt{\frac{c_\Lambda(3c_{BM}+2)\tau_B}{\tau_{SB}}}\left(2e^{\Lambda t}-\left(1 +\frac{c_{BM}}{(3c_{BM}+2)}\right)\right)\ .
\end{align}
We can also give a bound that works for all times. The simplest way to do so would be to employ the monotonicity of the error:
\begin{align}
   \|\rho_{BM,I}(t) - {\rho}_{C,I}(t)\|_1\leq\sqrt{\frac{c_\Lambda(3c_{BM}+2)\tau_B}{\tau_{SB}}}\left(2e^{\Lambda \max(t,T_a/2)}-\left(1 +\frac{c_{BM}}{(3c_{BM}+2)}\right)\right), \quad T_a = \sqrt{\frac{c_\Lambda\tau_B \tau_{SB}}{3c_{BM}+2}} \ ,
   \label{allT}
\end{align}
and this holds for all $t$ and all $\tau_B/\tau_{SB}$. We note that the more careful optimization in App.~\ref{optAside} allows us to get a tighter bound, but the difference is most drastic for large $\tau_B/\tau_{SB}$ when the bound is weaker than the trivial $ \|\rho_{BM,I}(t) - {\rho}_{C,I}(t)\|\leq 1 +c_{BM}$.
 
Note that for $t\leq T_a/2$ the expression in Eq.~\eqref{allT} grows exponentially with $\sqrt{\tau_B/\tau_{SB}}$. But we can have a bound linear in $\sqrt{\tau_B/\tau_{SB}}$ if we combine Eq.~\eqref{allT} with the trivial bound $ \|\rho_{BM,I}(t) - {\rho}_{C,I}(t)\|\leq 1 +c_{BM}$:
\begin{align}
   \|\rho_{BM,I}(t) - {\rho}_{C,I}(t)\|&\leq \text{min}\left(\sqrt{\frac{c_\Lambda(3c_{BM}+2)\tau_B}{\tau_{SB}}}\left(2e^{\Lambda \max(t,T_a/2)}-\left(1 +\frac{c_{BM}}{(3c_{BM}+2)}\right)\right), 1 +c_{BM}\right) \\
   &\leq  \sqrt{\frac{c_\Lambda(3c_{BM}+2)\tau_B}{\tau_{SB}}}\left(2e^{\Lambda t+1}-6/5\right)\ .
   \label{eq:F4}
\end{align}
\begin{proof} Here we prove Eq.~\eqref{eq:F4}. Note that the bound is monotonic in $t$, so the lowest value of $\tau_B/\tau_{SB}= \epsilon^{*2}$ for which the minimum is $2$ for all $t$ is given by equating the expressions at $t=0$:
\bes
\begin{align}
    \sqrt{\frac{c_\Lambda(3c_{BM}+2)\tau_B}{\tau_{SB}}}\left(2e^{2T_a/c_\Lambda\tau_{SB}}-\left(1 +\frac{c_{BM}}{(3c_{BM}+2)}\right)\right)= 1 +c_{BM},  \\
    \epsilon^*\sqrt{c_\Lambda(3c_{BM}+2)}\left(2e^{2\epsilon^*/\sqrt{c_\Lambda(3c_{BM}+2)}}-\left(1 +\frac{c_{BM}}{(3c_{BM}+2)}\right)\right)= 1 +c_{BM}
     \label{maxFrac}
\end{align}
\ees
Relaxing this to an inequality, we find:
\begin{equation}
   \epsilon^* \sqrt{c_\Lambda(3c_{BM}+2)}\left(2-\left(1 +\frac{c_{BM}}{(3c_{BM}+2)}\right)\right)\leq 1 +c_{BM}, \quad  \frac{2\epsilon^* \sqrt{c_\Lambda}}{\sqrt{3c_{BM}+2}}\leq 1, \quad e^{2\epsilon^*/\sqrt{c_\Lambda(c_{BM}+2})}\leq e^{1/c_\Lambda}\leq e \ .\label{eIneq}
\end{equation}
where we used $c_\Lambda\geq 1$. Thus any bound on Eq.~\eqref{allT} for $\tau_B/\tau_{SB}\leq \epsilon^{*2}$ that is monotonic in both $\tau_B/\tau_{SB}$ and $t$ will also be a bound for all $\tau_B/\tau_{SB}$. We choose to use the following in the exponent:
\begin{equation}
    \frac{\max(t,T_a/2)}{\tau_{SB}} \leq \frac{t}{\tau_{SB}} + \frac{1}{2}\sqrt{\frac{c_\Lambda\tau_B}{(3c_{BM}+2)\tau_{SB}}} \leq \frac{t}{\tau_{SB}} + \frac{\epsilon^* \sqrt{c_\Lambda}}{2\sqrt{(3c_{BM}+2)}}, \quad \text{for} \quad \frac{\tau_B }{\tau_{SB}}\leq \epsilon^{*2}\ .
\end{equation}
Thus we arrive at
\begin{align*}
 \|\rho_{BM,I}(t) - {\rho}_{C,I}(t)\|_1\leq\sqrt{\frac{c_\Lambda(3c_{BM}+2)\tau_B}{\tau_{SB}}}\left(2e^{2\epsilon^*/\sqrt{c_\Lambda(3c_{BM}+2)}}e^{\Lambda t}-\left(1 +\frac{c_{BM}}{(3c_{BM}+2)}\right)\right)\ ,
\end{align*}
which evaluates to the r.h.s of Eq.~\eqref{eq:F4} using Eq.~\eqref{eIneq} and $c_{BM}\geq 1  \Rightarrow c_{BM}/(3c_{BM}+2) \geq 1/5$.
\end{proof} 
We will now redo this calculation for $T_a = \sqrt{\tau_B\tau_{SB}/5}$ and $c_\Lambda =1$. Using this in Eq.~\eqref{startApp1}:
\begin{equation}
     \|\rho_{BM,I}(t) - {\rho}_{C\sqrt{5},I}(t)\|_1\leq \sqrt{\frac{\tau_B}{5\tau_{SB}}} \left( \left(3e^{\Lambda t} -1\right)c_{BM} + \left(7e^{\Lambda t} -5\right)\right)\ .
\end{equation}
Now we find $\epsilon^*$:
\begin{equation}
     1+c_{BM}= \frac{\epsilon^*}{\sqrt{5}} \left( \left(3e^{2\epsilon^*/\sqrt{5}} -1\right)c_{BM} + \left(7e^{2\epsilon^*/\sqrt{5}} -5\right)\right) \geq \frac{\epsilon^*}{\sqrt{5}} (2c_{BM}+2), \quad \Rightarrow \quad \frac{2\epsilon^*}{\sqrt{5}} \leq 1\ .
\end{equation}
This allows one to bound:
\begin{equation}
     \|\rho_{BM,I}(t) - {\rho}_{C\sqrt{5},I}(t)\|_1\leq \sqrt{\frac{\tau_B}{5\tau_{SB}}} \left( \left(3e^{\Lambda t+1} -1\right)c_{BM} + \left(7e^{\Lambda t+1} -5\right)\right) \label{sqrt5allT}
\end{equation}
for all $t$.

\section{Optimal time}
\label{optAside}

We begin with Eq.~\eqref{startApp1}, and use it to write down the optimum $t$-dependent coarse-graining time:
\begin{equation}
    T_a(t) = \sqrt{c_\Lambda\tau_B \tau_{SB}\frac{(e^{\Lambda t}-1 )}{(3c_{BM}+2)e^{\Lambda t}-c_{BM}}}\ .
\end{equation}
In principle we could consider a family of equations with different $T_a(t)$ and solve each one of them just to obtain $\rho(t)$ at one point. This is not a significant overhead. We proceed with the goal of obtaining a single equation since it is conceptually simpler. To this end we now define the bound relevance time $t^*$, which is the largest time for which our bound is still better than the trivial bound $ \|\rho_{BM,I}(t) - {\rho}_{C,I}(t)\|\leq 1+c_{BM}$:
\begin{equation}
    1+c_{BM}(t^*)=((3c_{BM}(t^*)+2)e^{\Lambda t^*}-c_{BM}(t^*)) \frac{ T_a}{\tau_{SB}}   +(e^{\Lambda t^*}-1 )\frac{c_\Lambda\tau_B}{T_a}\ .
\end{equation}
We would like to optimize $T_a$ at $t=t^*$, which would mean that $t^*$ is maximized. The corresponding equation for $t^*$ is:
\begin{equation}
   1+c_{BM}(t^*) =2 \sqrt{c_\Lambda\frac{\tau_B}{ \tau_{SB}}(e^{\Lambda t^*}-1 )((3c_{BM}(t^*)+2)e^{\Lambda t^*}-c_{BM}(t^*))}\ ,
\end{equation}
and solving the resulting quadratic equation yields:
\begin{equation}
    e^{\Lambda t^*} = \frac{4c_{BM}(t^*)+2 + Y}{2(3c_{BM}(t^*)+2)}, \quad Y =\sqrt{ 4( c_{BM}(t^*) +1)^2 +2(1+c_{BM}(t^*))^2 (3c_{BM}(t^*)+2)(\tau_{SB}/\tau_B c_\Lambda)}\ .
\end{equation}
Thus $T_a$ is:
\begin{equation}
    T_a = \sqrt{c_\Lambda\tau_B \tau_{SB}\frac{-2c_{BM}(t^*)-2 + Y}{(3c_{BM}(t^*)+2)(2c_{BM}(t^*) +2 + Y)}}\ .
\end{equation}
And the error is:
\bes
\begin{align}
     \|\rho_{BM}(t) - {\rho}_{C*}(t)\|_1&\leq  ((3c_{BM}(t^*)+2)e^{\Lambda\text{max}(t,T_a/2)}-c_{BM}(t^*)) \sqrt{\frac{c_\Lambda\tau_B }{\tau_{SB}} \frac{-2c_{BM}(t^*)-2 + Y}{(3c_{BM}(t^*)+2)(2c_{BM}(t^*) +2 + Y)}}   \\
     &+(e^{\Lambda\text{max}(t,T_a/2)}-1 )\sqrt{\frac{c_\Lambda\tau_B}{\tau_{SB}} \frac{(3c_{BM}(t^*)+2)(2c_{BM}(t^*) +2 + Y)}{-2c_{BM}(t^*)-2 + Y}}
\end{align}
\ees
We  should only use these expressions if $t^*>T_a/2$.
For $\tau_{SB}\ll \tau_B$ (when $t^*$ is not  likely to be $>T_a/2$), this changes the asymptotic form of $T_a$ to $T_a \sim \tau_{SB}$. In the regime $\tau_{SB}\gg \tau_B$ that we normally address with the CGME, we expect these formulas to be a minor  improvement on the loose bounds presented in the main text.  Their main advantage is a tighter bound on the error for $t>T_a/2$, $t\sim\tau_{SB}$.

\section{Proof of the bound~\eqref{eq:183} for the Redfield limits of integration error}
\label{app:Red-details}

Starting from Eq.~\eqref{eq:182}, we first send the integration limit to infinity:
\beq
\int_0^{4t/\tau_{SB}}\min (1, \frac{4\tau_B}{x\tau_{SB}} + \epsilon_T)e^{-x}dx \leq\int_0^{\infty}\min (1, \frac{4\tau_B}{x\tau_{SB}} + \epsilon_T)e^{-x}dx = \epsilon_T + \int_0^{\infty} \min (1-\epsilon_T, \frac{4\tau_B}{x\tau_{SB}}  )e^{-x}dx  \ .
\eeq
The $\min$ function evaluates to $1- \epsilon_T$ when $x\leq x^*$ and to $\frac{4\tau_B}{x\tau_{SB}}$ when $x>x^*$, where \beq
x^* = \frac{4\tau_B}{\tau_{SB} (1-\epsilon_T)}\ .
\eeq
Therefore:
\bes
\begin{align}
\int_0^{\infty} \min (1-\epsilon_T, \frac{4\tau_B}{x\tau_{SB}}  )e^{-x}dx  
&=  \int_0^{x^*}(1-\epsilon_T) e^{-x}dx + \int_{x^*}^\infty \frac{4\tau_B}{x\tau_{SB}} e^{-x}dx \\
&= (1-\epsilon_T)(1-e^{-x^*}) +  \frac{4\tau_B}{\tau_{SB}}E_1(x^*) \ ,
\end{align}
\ees
where $E_1(x)$ denotes the exponential integral function: $E_1(x)\equiv \int_x^\infty t^{-1}e^{-t}dt$. A tight bound in terms of elementary functions is given by~\cite{Alzer:1997}[Thm.~2]:
\beq
E_1(x) \leq -\ln(1-e^{-x}) \quad (0<x<\infty) \ .
\eeq
Thus
\bes
\begin{align}
    \|\rho_{BM}(t) - \rho_R(t)\|_1  &\leq  c_{BM}e^{4t/\tau_{SB}} \left[ \epsilon_T+(1-\epsilon_T)(1-e^{-x^*}) - \frac{4\tau_B}{\tau_{SB}}\ln(1-e^{-x^*}) \right] \\
    &\leq c_{BM}e^{4t/\tau_{SB}}  \left[ \epsilon_T + \frac{4\tau_B}{\tau_{SB}}\left(1 - \ln\left(1-e^{-\frac{4\tau_B}{\tau_{SB} (1-\epsilon_T)}}\right)\right)\right] \ ,
    \label{eq:183new}
\end{align}
\ees
where in the second line we used that $\forall x$: $1-e^{-x} < x$.

Alternatively, and using more elementary methods, assuming $x^*<1$, we obtain:
\bes
\begin{align}
   \int_0^{\infty}\min (1, \frac{4\tau_B}{x\tau_{SB}} + \epsilon_T)e^{-x}dx &=\epsilon_T +(1-\epsilon_T) (1-e^{-x^*}) + \int_{x^*}^1 \frac{4\tau_B}{x\tau_{SB}} e^{-x}dx  +\int_{1}^\infty \frac{4\tau_B}{x\tau_{SB}} e^{-x}dx \\ 
   &\leq \epsilon_T +(1-\epsilon_T) x^* + \int_{x^*}^1 \frac{4\tau_B}{x\tau_{SB}} dx  + \int_1^\infty\frac{4\tau_B}{\tau_{SB}}e^{-x}dx  \\
   &=  \epsilon_T +(1-\epsilon_T) x^* -  \frac{4\tau_B}{\tau_{SB}} \text{ln}x^* + \frac{4\tau_B}{e\tau_{SB}}    \\
   &= \epsilon_T + \frac{4\tau_B}{\tau_{SB} } +  \frac{4\tau_B}{\tau_{SB}} \text{ln}\frac{\tau_{SB} (1-\epsilon_T)}{4\tau_B} + \frac{4\tau_B}{e\tau_{SB}} \quad (x^*< 1 ) \ ,
\end{align}
\ees
while when $x^*\geq 1$ the term with the logarithm is absent:
\begin{align}
   {\int_0^{\infty}\min (1, \frac{4\tau_B}{x\tau_{SB}} + \epsilon_T)e^{-x}dx \leq\epsilon_T + \frac{4\tau_B}{\tau_{SB} } + \frac{4\tau_B}{e\tau_{SB}} \quad (x^*\geq 1 ) } \ .
\end{align}
Combining these two cases directly yields Eq.~\eqref{eq:183}, which is an upper bound on the tighter bound presented in Eq.~\eqref{eq:183new}.

\nocite{apsrev41Control}
\bibliography{refs}

\end{document}